\nonstopmode

\newcommand{\ie}{i.e.{}}
\newcommand{\eg}{e.g.{}}
\newcommand{\eV}{\U{eV}}
\newcommand{\Cal}[1]{{\cal #1}}
\newcommand{\U}[1]{\,{\rm{#1}}}
\newcommand{\euler}{\mathrm e}
\newcommand{\Sum}{\sum\limits}
\newcommand{\Int}{\int\limits}
\newcommand{\differential}{\>\mathrm d}
\newcommand{\E}[1]{\times 10^{#1}}
\newcommand{\nbh}{\hbox{-}}
\newcommand{\xray}{x\nbh{}ray}
\newcommand{\Xray}{X\nbh{}ray}
\newcommand{\mat}[1]{\hbox{\boldmath{$#1$}\unboldmath}}
\newcommand{\mats}[1]{\hbox{\boldmath{$\scriptstyle{#1}$}\unboldmath}}
\renewcommand{\vec}[1]{\hbox{\boldmath{$#1$}\unboldmath}}
\newcommand{\diag}{\textbf{diag}}
\newcommand{\transpose}{{}^{\textrm{\scriptsize T}}}
\newcommand{\unitmatrix}{\mat{\mathbbm{1}}}
\newcommand{\eref}[1]{(\ref{#1})}
\newcommand{\Eref}[1]{Equation~(\ref{#1})}
\newcommand{\sref}[1]{Sec.~\ref{#1}}
\newcommand{\Sref}[1]{Section~\ref{#1}}
\newcommand{\fref}[1]{Fig.~\ref{#1}}
\newcommand{\atopa}[2]{\genfrac{}{}{0pt}{}{#1}{#2}}

\documentclass[a4paper,aip,jmp,10pt,twocolumn,longbibliography,showkeys,preprintnumbers]{revtex4-1}
\usepackage{amsmath,amsthm,amssymb,bbm,graphicx}
\usepackage[nativepdf,bookmarks,bookmarksopen,bookmarksnumbered,raiselinks,breaklinks,%
pdftitle={Nonlinearity in the sequential absorption of multiple photons},%
pdfauthor={Christian Buth},%
pdfsubject={Atomic Physics, Mathematical Physics},%
pdfkeywords={rate equations, system of linear first-order ordinary differential
equations, sequential absorption of multiple photons,
multiphoton process, separable, linked, simultaneous,
Volterra integral equation of the second kind,
analytical solution, ultrafast, intense, x rays, photoabsorption, decay,
light-matter interaction, nitrogen atom}]{hyperref}
\usepackage[anythingbreaks]{breakurl}
\begin{document}
\title{Nonlinearity in the sequential absorption of multiple photons}
\author{Christian Buth}
\thanks{World Wide Web: \href{http://www.christianbuth.name}
{www.christianbuth.name}, electronic mail}
\email{christian.buth@web.de}
\affiliation{Theoretische Chemie, Physikalisch-Chemisches Institut,
Ruprecht-Karls-Universit\"at Heidelberg, Im Neuenheimer Feld~229,
69120~Heidelberg, Germany}
\date{13 September 2017}

\begin{abstract}
I classify multiphoton absorption into separable, linked, and simultaneous
processes.
The first and second types can be distinguished when the rate-equation
approximation is valid whereas the third type refers to the case when the full
description of multiphoton absorption is essential.
For this purpose, rate equations are solved analytically without decay
processes which shows that even if many photons are absorbed the
interaction with the light field is linear and one has the case
of separable multiphoton absorption.
Next a short-pulse approximation is investigated in which I first solve the
rate equations without decay processes and then solve only rate equations
for the ensuing decay.
Finally, the full rate equations are examined and a successive
approximation of the underlying Volterra integral equation
of the second kind is derived leading to linked
multiphoton absorption by the involved decay widths.
The three methods are applied to a nitrogen atom in intense and
ultrafast x~rays from free-electron lasers~(FELs).
The linearity theorem barely approximates the results in the presence
of decay processes which is also not satisfactorily corrected for
by the short-pulse approximation.
The successive approximation gives excellent agreement with the
numerically-exact solution of the rate equations.
\end{abstract}


\preprint{arXiv:1612.07105}
\maketitle

\theoremstyle{plain}
\newtheorem{lemma}{Lemma}
\newtheorem{theorem}{Theorem}
\renewcommand{\qedsymbol}{q.e.d.{}}
\theoremstyle{definition}
\newtheorem{definition}{Definition}
\theoremstyle{remark}
\newtheorem{remark}{Remark}

\section{Introduction}

Rate equations which are systems of linear first-order
ordinary differential equations~\cite{Walter:GD-00} have recently
come into focus in the study of the intense and ultrafast interaction
of x~rays with matter [\eg, Refs.~\onlinecite{Rohringer:XR-07,%
Son:HA-11,*Son:EH-11,Buth:UA-12,Son:MC-12,*Son:EM-15,Rudek:UE–12,%
Rudek:RE-13,Fukuzawa:DI-13,Ho:TT-14,Liu:RE-16,Buth:NX-17}]
due to the development of \xray~lasers, especially the
state-of-the-art \xray~FELs such as the
the Linac Coherent Light Source~(LCLS)~\cite{LCLS:CDR-02,Emma:FL-10}
in Menlo Park, California, USA,
the SPring-8 Angstrom Compact free electron LAser~(SACLA)~\cite{Ishikawa:CX-12}
in Sayo-cho, Sayo-gun, Hyogo, Japan,
the SwissFEL~\cite{SwissFEL:CDR-12} in Villigen, Switzerland,
and the European X-Ray Free-Electron Laser~(XFEL)~\cite{Altarelli:TDR-06}
in Hamburg, Germany.
Such rate equations have been used for a long time in the optical
regime~\cite{LHuillier:MC-83,LHuillier:RG-83,LHuillier:LP-83,Crance:DS-85,%
Lambropoulos:ME-87,Delone:MP-00}
and, meanwhile, they have become frequently the basis for
an understanding of the interaction with x~rays,~\cite{Rohringer:XR-07,%
Hoener:FA-10,Young:FE-10,Son:HA-11,*Son:EH-11,Buth:UA-12,Liu:RE-16,Buth:NX-17,%
Obaid:FL-17} even in the presence of resonances,~\cite{Xiang:RA-12,%
Rudek:UE–12,Rudek:RE-13,Fukuzawa:DI-13,Ho:TT-14} until coherent phenomena become
important.~\cite{Rohringer:RA-08,*Rohringer:PN-08,Kanter:MA-11,Cavaletto:RF-12,%
Cavaletto:FC-13,Adams:QO-13,Cavaletto:HF-14,Li:CR-16}
This approach is required as the x~rays from FELs are so intense that
multiple x~rays can be absorbed in the course of the interaction
unlike experiments at synchrotrons which are limited to
one-\xray-photon processes.~\cite{Als-Nielsen:EM-01,Adams:QO-13}
Hence---although the interaction with the x~rays remains well-described by
few-photon absorption cross sections in the cases considered for this
work~\cite{Makris:MM-09}---a
perturbative treatment of the interaction of the \xray~pulse with matter is
typically no longer a viable approach as there is a substantial
ground-state depletion.

In \xray~science the term ``multiphoton absorption'' is frequently
used to refer to the case of the sequential absorption of several
photons [\eg, Refs.~\onlinecite{Rohringer:XR-07,Makris:MM-09,%
Hoener:FA-10,Young:FE-10,Son:HA-11,*Son:EH-11,Buth:UA-12,Son:MC-12,*Son:EM-15,%
Liu:RE-16,Buth:NX-17,Obaid:FL-17}] whereas for optical light predominantly
the simultaneous absorption of several photons is meant.~\cite{Delone:MP-00}
The attribute \emph{simultaneous}, thereby, refers to the fact that a
few-photon absorption cannot be meaningfully broken up into isolated photon
absorption events with a smaller number of photons.
\Xray~absorption may occur either non-resonantly~\cite{Rohringer:XR-07,%
Young:FE-10,Son:HA-11,*Son:EH-11,Buth:UA-12,Liu:RE-16,Buth:NX-17,Obaid:FL-17}
or resonantly-enhanced.~\cite{Xiang:RA-12,Son:MC-12,*Son:EM-15,Rudek:UE–12,%
Rudek:RE-13,Fukuzawa:DI-13,Ho:TT-14}
The latter refers to a form of resonance-enhanced multiphoton
ionization~(REMPI)~\cite{Delone:MP-00}
which is termed in this context resonance-enabled \xray~multiple
ionization~(REXMI).~\cite{Rudek:UE–12,Rudek:RE-13}
Unlike optical photons, x~rays typically have enough energy
to eject electrons from atoms until maybe the highest charge states
where simultaneous absorption of multiple
x~rays plays only a diminishing role.~\cite{Doumy:NA-11,Tamasaku:TP-14}

In many applications, the \xray~energy is far above the ionization
threshold of the involved atoms in all charge states
such as in structure determination of biomolecules~\cite{Neutze:BI-00}
and condensed matter.~\cite{Als-Nielsen:EM-01}
This is the case I focus on in this study.
It is also what was examined in the initial works on matter in FEL x~rays, \eg,
Refs.~\onlinecite{Hoener:FA-10,Young:FE-10,Cryan:AE-10,Fang:DC-10,Fang:MI-12}.
There the population of charge states
due to ionization by the sequential absorption of multiple x~rays
and ensuing electronic decay processes are investigated by
numerically solving rate equations.~\cite{Rohringer:XR-07,Son:HA-11,%
*Son:EH-11,Buth:UA-12,Son:MC-12,*Son:EM-15,Liu:RE-16,Buth:NX-17}
The solutions are then used to compute experimentally accessible
quantities such as ion yields.
Yet a more formal treatment of the underlying systems of differential
equations has not been pursued to date.

Here I would like to investigate the solution of such rate equations
analytically.
I formulate and prove the linearity theorem of multiphoton
absorption,~\footnote{%
The essence of the linearity theorem of multiphoton absorption
has already been stated in the third paragraph in the left column
on page~2 of Ref.~\onlinecite{Buth:UA-12}.}
a short-pulse approximation that uses this theorem with decay equations,
and I investigate the successive approximation of the solution of
the rate equations.
This allows me to classify multiphoton absorption into simultaneous,
linked, and separable processes.
Namely, specialized to a two-photon process, there is the well-known
simultaneous absorption of two photons~\cite{Goppert:EZ-31}
for x~rays.~\cite{Rohringer:XR-07,Doumy:NA-11}
This changes when the rate-equation approximation without simultaneous
two-photon terms is applicable.~\cite{LHuillier:MC-83,LHuillier:RG-83,%
LHuillier:LP-83,Crance:DS-85,Lambropoulos:ME-87,Delone:MP-00,%
Rohringer:XR-07,Makris:MM-09,Buth:UA-12,Buth:NX-17}
Then linked two-\xray-photon absorption can be discussed where decay
widths establish the link between two distinct one-\xray-photon absorptions.
In this article, I also turn to the third case of separable multiphoton
absorption hitherto not analyzed which is covered by the linearity theorem.
The reader should always bear in mind that linked and separable
multiphoton absorption are approximations to simultaneous absorption.
The core results of this article can be found in~\sref{sec:decayeqn}
where the analytical solution of the rate equations with decay are determined;
\sref{sec:ratewodecay} serves as an introduction to the subject neglecting
decay processes.

This article is structured as follows.
\Sref{sec:ratewodecay} treats rate equations without decay processes
which are introduced in~\sref{sec:rateonly}.
The linearity theorem of multiphoton absorption is proven in
\sref{sec:linmulph}.
Rate equations with decay processes are treated in~\sref{sec:decayeqn};
they are formulated in~\sref{sec:introdecay}.
In~\sref{sec:soldecay}, I examine rate equations with decay processes
in a short-pulse approximation.
\Sref{sec:succapprox} contains the derivation of a Volterra integral
equation from the rate equations with decay and its successive
approximation.
Results and discussions are in~\sref{sec:results} where the time evolution
of the charge states of a nitrogen is predicted using the approximations
derived before.
Conclusions are drawn in~\sref{sec:conclusion}.

Equations are formulated in atomic units.~\cite{Hartree:WM-28,Szabo:MQC-89}
All details of the calculations in this article are provided in the
Supplementary Data.~\cite{SuppData}\nocite{Mathematica:pgm-V11}

\section{Rate equations without decay processes}
\label{sec:ratewodecay}
\subsection{Closed system of rate equations}
\label{sec:rateonly}

I set out from rate equations which describe \xray~nonlinear optical processes
of atoms~\cite{Rohringer:XR-07,Buth:UA-12,Buth:NX-17} and
molecules.~\cite{Buth:UA-12,Liu:RE-16}
The discussion is restricted to the case that the \xray~energy
is sufficiently high to ionized all electrons of the atom
by one-photon absorption.~\footnote{%
For conceptual ease, I restrict myself here in considering the final
state to be electron-bare which, however, is not necessarily required.
Instead the single final state may also be a configuration which
still has electrons in deep lying shells which is energetically to
high in energy such that it cannot be ionized by the photons with
the given energy and is not amenable to REXMI.~\cite{Rudek:UE–12,Rudek:RE-13}}
The systems are assumed to be initially neutral with
$N \in \mathbb N$~electrons and in a nondegenerate ground state.
The electronic configurations of the system are enumerated by a double
index of the electron number~$n \in \{0, \ldots, N\}$ and the number
of configurations~$K^{(n)} \in \mathbb N$ in charge state~$n$.
The configurations are arranged with decreasing energy of configurations
in each charge state.
In total there are~$K = \Sum_{n = 0}^N K^{(n)}$ configurations.
There is only one neutral state and one electron-bare state, \ie,
$K^{(0)} = K^{(N)} = 1$.
For ease of notation, I let~$K^{(N+1)} = 0$.
With \emph{\xray~flux}~$J$ I denote a continuous real function in time
which is greater than zero only on a finite time interval and zero otherwise.
With $\sigma_{j \leftarrow i}^{(n)}, \sigma_i^{(n)} \geq 0$
for~$i \in \{1, \ldots, K^{(n)}\}$, $j \in \{1, \ldots, K^{(n-1)}\}$,
and $n \in \{0, \ldots, N\}$, I specify one-\xray-photon absorption \emph{cross
sections}, partial and total cross sections, respectively.~\footnote{%
A fixed photon energy is assumed which needs to be high enough such that
all electrons can be ionized;
the dependence of the quantities on the photon energy is omitted throughout.
If the \xray~photon energy is lower than the energy that is necessary to
strip the system of all electrons, then the total absorption cross
section of this last configuration---I assume that there is a single last
configuration---is zero, although there are still
electrons in lower lying shells.
Increasing the \xray~photon energy renders the total absorption
cross section nonzero.
Please see figure~2 in~\cite{Buth:UA-12} and its explanation.
Only cationic configurations are included
without regarding shake-off or shake-up processes and other
multielectron effects such as double Auger decay.}
Here~$\sigma_{j \leftarrow i}^{(n)}$~is the cross section to transition
from a configuration~$i$ with~$n$~electrons to a configuration~$j$
with~$n-1$~electrons by one-\xray-photon~absorption.
Typically not all configurations~$j$ can be reached
from configuration~$i$ and $\sigma_{j \leftarrow i}^{(n)}$ is set
to zero, if it is not possible.
For notational convenience, I let~$\sigma_{j \leftarrow 1}^{(0)}
= \sigma_{1 \leftarrow i}^{(N+1)} = 0$ for any~$i$, $j$.
The $\sigma_i^{(n)}$~is the cross section to transition from
configuration~$i$ with $n$~electrons to any accessible configuration
with $n-1$~electrons.
As no x~rays can be absorbed by the electron-bare nucleus, I
have for this configuration~$\sigma_1^{(0)} = 0$.

Rate equations can be constructed as follows or be derived from the
master equation in Lindblad form as a Pauli master equation which is discussed
in Refs.~\onlinecite{Pauli:FS-28,Louisell:QS-90,Buth:MC-17}.
Given probabilities~\cite{Krengel:WS-05}~$P_j^{(n)}(t)$ at
time~$t \in \mathbb R$ which are differentiable, to find the system
at time~$t \in \mathbb R$ with~$n$~electrons in configuration~$j$.
The initial condition is specified at the \emph{initial
time}~$\tau \in \mathbb R$ as~$P_j^{(n)}(\tau)$.
Here I assume further that the system is initially in the single neutral
charge state, \ie, $P_j^{(n)}(\tau) = \delta_{1\,j} \, \delta_{N\,n}$
with the Kronecker-$\delta$.~\cite{Kuptsov:KD-16}
Then the first rate equation, for the system in its neutral charge state reads
\begin{equation}
  \label{eq:reqground}
  \dfrac{\differential P_1^{(N)}(t)}{\differential t}
    = - \sigma_1^{(N)} \, P_1^{(N)}(t) \, J(t) \; ;
\end{equation}
it describes the ionization of the neutral system by the x~rays.
The rate equations of the singly-ionized system already have the full
form of rate equations without decay processes
\begin{equation}
  \label{eq:reqsingly}
  \dfrac{\differential P_j^{(N-1)}(t)}{\differential t}
    = \bigl[ \sigma_{j \leftarrow 1}^{(N)} \; P_1^{(N)}(t)
    - \sigma_j^{(N-1)} \, P_j^{(N-1)}(t) \bigr] \, J(t) \; ;
\end{equation}
for~$j \in \{1, \ldots, K^{(N-1)}\}$.
This continues analogously until the last rate equation,
for the electron-bare system, which is
\begin{equation}
  \label{eq:reqbare}
  \dfrac{\differential P_1^{(0)}(t)}{\differential t}
    = \Sum_{i = 1}^{K^{(1)}} \sigma_{1 \leftarrow i}^{(1)} \;
    P_i^{(1)}(t) \, J(t) \; .
\end{equation}
The above said can be compactly summarized as follows.
Inspecting the rate equations~\eref{eq:reqground}, \eref{eq:reqsingly},
\eref{eq:reqbare}, I observe that there are $K$~rate equations which
form the $K \times K$~system of linear first-order ordinary
differential equations
\begin{equation}
  \label{eq:req}
  \dfrac{\differential P_j^{(n)}(t)}{\differential t} = \Bigl[
    \Sum_{i = 1}^{K^{(n+1)}} \sigma_{j \leftarrow i}^{(n+1)} \; P_i^{(n+1)}(t)
    - \sigma_j^{(n)} \, P_j^{(n)}(t) \Bigr] \, J(t) \; .
\end{equation}
Then this constitutes a \emph{system of rate equations}.
If for such a system of rate equations
the \emph{total cross sections} are given by
\begin{equation}
  \label{eq:sigtot}
  \sigma_j^{(n)} = \Sum_{i = 1}^{K^{(n-1)}} \sigma_{i \leftarrow
    j}^{(n)} \; ,
\end{equation}
then it is termed a \emph{closed system of rate equations}.

The system of rate equations~\eref{eq:req} for all~$n$, $j$ can be
rewritten compactly by letting~$\vec p(t) \in \mathbb R^K$~stand
for the vector with the entries~$P_i^{(n)}(t)$, \ie,
\begin{equation}
  \label{eq:probvec}
  \vec p(t) = \bigr( P_1^{(N)}(t), P_1^{(N-1)}(t), \ldots,
                     P_{K^{(N-1)}}^{(N-1)}(t), \ldots, P_1^{(0)}(t) \bigr)
                     \transpose \; .
\end{equation}
The partial cross sections~$\sigma_{j \leftarrow i}^{(n)}$
form a strictly lower triangular matrix and $-\sigma_i^{(n)}$ are on
the matrix diagonal.
Together they are referred to as the \emph{cross section matrix}~$\mat \Sigma$.
With this notation, I express~\eref{eq:req} as
\begin{equation}
  \label{eq:vecreq}
  \dfrac{\differential \vec p(t)}{\differential t} = \mat \Sigma \; \vec p(t)
    \, J(t) \; .
\end{equation}
The structure of~$\mat \Sigma$ can be characterized by
\begin{definition}
\label{st:closure}
A real $K \times K$~matrix~$\mat A$ that has the \emph{closure property}:
\begin{equation}
  \label{eq:closure}
  - A_{ii} = \Sum_{\atopa{\scriptstyle j=1}{\scriptstyle j \neq i}}^K A_{ji}
    \; ,
\end{equation}
for all~$i \in \{1, \ldots, K\}$ is referred to as \emph{closed matrix}.
\end{definition}
The closure property~\eref{eq:sigtot} implies that $\mat \Sigma$
in~\eref{eq:vecreq} is a closed matrix.
\begin{remark}
If the matrix~$\mat A$ from Definition~\ref{st:closure} is also
lower-triangular, then it has at least one zero diagonal element
because~$A_{KK} = 0$~must hold in order to fulfill~\eref{eq:closure}.
\end{remark}
So far I referred to the~$p_i(t)$ as probabilities thereby assuming them
to lie in~$[0 ; 1]$ for all times~$t$.~\cite{Krengel:WS-05}
This fact, however, needs to be established rigorously.
\begin{remark}
\label{st:positivity}
The $\mat \Sigma$~is essentially nonnegative, \ie, $\Sigma_{ij} \geq 0$
for~$i \neq j$, and thus a Metzler matrix.~\cite{Giorgio:ME-15}
Therefore, Theorem~16.VII on page~183 in Ref.~\onlinecite{Walter:GD-00}
ensures the nonnegativity of the~$p_i(t)$.
\end{remark}
The word ``closed'' was chosen above to hint that
\begin{lemma}
\label{st:prob}
Probability in a closed system of rate equations~\eref{eq:req} is conserved.
\end{lemma}
\begin{proof}
As~$\mat \Sigma$ is a closed matrix~\eref{eq:sigtot}, I observe that the
column sums of~$\mat \Sigma$ vanish, \ie, $\Sum_{i=1}^K \Sigma_{ij} = 0$ for
all~$j \in \{1, \ldots, K\}$.
Then it follows by summing the components of~\eref{eq:vecreq}:
\begin{equation}
  \label{eq:probability}
  \dfrac{\differential}{\differential t} \Sum_{i=1}^K p_i(t)
    = \Sum_{j=1}^K \Bigl[ \Sum_{i=1}^K \Sigma_{ij} \Bigr] \; p_j(t)
    \, J(t) = 0 \; ,
\end{equation}
that the probability is conserved for all times~$t$.
\end{proof}
Thus I make the
\begin{remark}
\label{st:probability}
With Remark~\ref{st:positivity} and Lemma~\ref{st:prob}
I can rightfully refer to the~$p_i(t)$ as probabilities~\cite{Krengel:WS-05}
provided that the initial values fulfill~$0 \leq p_i(\tau) \leq 1$
for~$i \in \{1, \ldots, K\}$ and $\Sum_{i = 1}^K p_i(\tau) = 1$ holds.
\end{remark}

\subsection{Linearity theorem of multiphoton absorption}
\label{sec:linmulph}

The rate equations~\eref{eq:vecreq} constitute a matrix initial-value
problem which satisfies the Lappo-Danilevskii condition~\cite{Komlenko:FM-11}:
\begin{equation}
  \label{eq:lappo}
  \mat \Sigma \, J(t) \Int_{\tau}^t \mat \Sigma \, J(t') \differential t'
    = \Int_{\tau}^t \mat \Sigma \, J(t') \differential t' \>
    \mat \Sigma \, J(t) \; .
\end{equation}
Hence the fundamental matrix can be written as an
exponential function~$\euler^{\Int_{\tau}^t \mats \Sigma \, J(t')
\differential t'}$~\cite{Komlenko:FM-11} which leads to the
solution of the rate equations~\eref{eq:vecreq}
on the last line of Eq.~\eref{eq:pertexp}.
But only some of the results in the following paragraphs are
found this way.

To gain deeper insight into the physics described by the
rate equations~\eref{eq:vecreq}, I obtain the solution
of~\eref{eq:vecreq} in an alternative way
using the eigenbasis of~$\mat \Sigma$.
For this I notice that, as~$\mat \Sigma$ is lower triangular, its
eigenvalues are the diagonal elements~$-\sigma_i^{(n)}$ which are
negative or zero such as
the last one~$-\sigma_1^{(0)} = 0$---because the electron-bare
nucleus cannot be ionized further---and pairwise distinct as the
total absorption cross sections~$\sigma_i^{(n)}$ refer to different
configurations~$i$ of charge state~$n$ where, on configuration level,
there are no degeneracies.
Following Lemma~\ref{st:diag}, $\mat \Sigma$~is diagonalizable where
the eigenvalues form the diagonal matrix~$\mat \Lambda
= \diag(\lambda_1, \ldots, \lambda_K)$ and the eigenvectors are
gathered in the matrix~$\mat U$.
As $\mat \Sigma$~is time independent, so are~$\mat \Lambda$ and $\mat U$.
Transforming~\eref{eq:vecreq} to the eigenbasis of~$\mat \Sigma$ yields
\begin{eqnarray}
  \label{eq:diagrates}
  \mat U^{-1} \, \dfrac{\differential \vec p(t)}{\differential t} &=&
    \dfrac{\differential \vec q(t)}{\differential t} = \mat U^{-1} \,
    \mat \Sigma \, \mat U \; \mat U^{-1} \, \vec p(t) \, J(t) \nonumber \\
  &=& \mat \Lambda \, \vec q(t) \, J(t) \; ,
\end{eqnarray}
which is a decoupled system of linear first-order ordinary differential
equations.
The $\ell$~th equation reads
\begin{equation}
  \dfrac{\differential q_{\ell}(t)}{\differential t} = \lambda_{\ell}
  \, q_{\ell}(t) \, J(t) \; ,
\end{equation}
which has the solution for time~$t \geq \tau$ with the integration starting at
time~$\tau$:
\begin{equation}
  \label{eq:decreq}
  q_{\ell}(t) = q_{\ell}(\tau) \; \euler^{\lambda_{\ell} \, \Phi(t)} \; ,
\end{equation}
where the initial value is~$q_{\ell}(\tau) \equiv (\mat U^{-1} \,
\vec p(\tau))_{\ell}$ and the \emph{\xray~fluence} is
\begin{equation}
  \label{eq:fluence}
  \Phi(t) = \Int_{\tau}^t J(t') \differential t' \; .
\end{equation}
As the system is initially in its ground state, I have
$\vec p(\tau) = (1,0,\ldots,0)\transpose \in \mathbb R^K$ and thus
$q_{\ell}(\tau) = (\mat U^{-1})_{\ell 1}$.
From the solution~$\vec q(t)$, the probabilities follow via~$\vec p(t) =
\mat U \, \vec q(t)$, hence
\begin{equation}
  \label{eq:preqfl}
  p_i(t) = \Sum_{\ell=1}^K U_{i\ell} \, q_{\ell}(t) = \Sum_{\ell=1}^K U_{i\ell}
    \, q_{\ell}(\tau) \; \euler^{-|\lambda_{\ell}| \, \Phi(t)} \; ,
\end{equation}
which is the unique solution of the system of rate
equations~\eref{eq:vecreq}.

I observe that $p_i(t)$~depends \emph{only} on the \xray~fluence~$\Phi(t)$
as it is the case for one-photon absorption.~\cite{Buth:UA-12}
The sign of~$\lambda_{\ell}$ in~\eref{eq:decreq} was made explicit
in~\eref{eq:preqfl};
it implies that the exponential functions in~\eref{eq:preqfl}
decay quickly to zero with increasing fluence meaning that the probability
of the system at time~$t$ to be in the configuration~$i$ \emph{with}
electrons decays to zero such that for~$\Phi(t) \to \infty$, eventually,
only the last configuration~$K$ of the system which is stripped bare of
electrons, \ie, $\lambda_K = 0$, has unit probability
[see~\eref{eq:probability}].
Conversely, the limit~$\Phi(t) \to 0$ lets all exponential functions
in~\eref{eq:preqfl} become unity and thus~$\vec p(t) = \vec p(\tau)$,
\ie, the system remains in the initial state for all times~$t$.

The solution~\eref{eq:preqfl} of the rate equations~\eref{eq:vecreq}
is recast, by expanding the exponential function
into a series using~$\vec q(\tau) = \mat U^{-1} \, \vec p(\tau)$, yielding
\begin{eqnarray}
  \label{eq:pertexp}
  p_i(t) &=& \Sum_{m=0}^{\infty} \dfrac{\Phi^m(t)}{m!}
    \Sum_{\ell,j=1}^K (\mat U^{\vphantom{-1}})_{i\ell} \, \lambda_{\ell}^m \,
    (\mat U^{-1})_{\ell j} \> p_j(\tau) \nonumber \\
  &=& \Sum_{j=1}^K p_j(\tau) \Sum_{m=0}^{\infty} \dfrac{\Phi^m(t)}{m!}
    \> \mat \Sigma^m_{ij} \\
  &=& \Sum_{j=1}^K \bigl( \euler^{\mats \Sigma \, \Phi(t)} \bigr)_{ij} \;
    p_j(\tau) \; . \nonumber
\end{eqnarray}
The middle expression in~\eref{eq:pertexp} exhibits the
perturbative nature of the problem.
Namely, the probability~$\vec p(t)$ is composed of the summands with
increasing powers of~$\mat \Sigma \> \Phi(t)$, \ie,
zero-photon absorption gives~$\unitmatrix$, one-photon absorption
adds~$\mat \Sigma \> \Phi(t)$, two-photon absorption
adds~$\tfrac{1}{2} \, \bigl( \mat \Sigma \> \Phi(t) \bigr)^2$, \ldots.
The above said proves
\begin{theorem}[Linearity theorem of multiphoton absorption]
\label{st:linearity}
Given a system of rate equations~\eref{eq:req},
then there is no nonlinearity in the interaction of x~rays with the
system and all multiphoton processes can be decomposed in terms
of subsequent (chained) one-\xray-photon absorptions~\eref{eq:pertexp}.
\end{theorem}
The situation that is described by Theorem~\ref{st:linearity}
is termed \emph{separable multiphoton absorption}.
The dependence of~$\vec p(t)$ on~$t$ is nonlinear due to the
exponential functions in~\eref{eq:preqfl} that become constants
which depend only on the \xray~fluence after the \xray~pulse has
passed for~$t \to \infty$.
In the case of separable multiphoton processes, no nonlinearity
in the interaction is present.
The fact that the exponential series always converges
ought not to mislead the reader: this convergence does not imply that the
physical problem of multiphoton absorption actually is
amenable to a perturbative treatment which leads to the cross sections
that enter~\eref{eq:req}.
This fact needs to be established
independently.~\cite{Sorokin:PE-07,Makris:MM-09}

\section{Rate equations with decay processes}
\label{sec:decayeqn}
\subsection{Closed system of rate equations}
\label{sec:introdecay}

In many cases the formulation of the rate equations~\eref{eq:req}
from Sec.~\ref{sec:rateonly} represent an idealization because
there are decay processes occurring.
First, there is electronic decay in which an electron drops from
a higher shell to a lower shell releasing the excess energy by
ejecting another electron from a higher shell.
Second, there is radiative decay where an electron drops from
a higher shell to a lower shell emitting a photon with an
energy corresponding to the energy difference between the
two shells involved in the transition.
For electronic decay, the final state has one electron less
compared with the initial state whereas in radiative decay
the number of electrons is left unchanged.

With decay widths, the rate equations~\eref{eq:req} become
\begin{eqnarray}
  \label{eq:reqgam}
  \dfrac{\differential P_j^{(n)}(t)}{\differential t} &=& \Bigl[
    \Sum_{i = 1}^{K^{(n+1)}} \sigma_{j \leftarrow i}^{(n+1)} \;
    P_i^{(n+1)}(t) - \sigma_j^{(n)} \, P_j^{(n)}(t)
    \Bigr] \, J(t) \nonumber \\
  &&{} + \Sum_{i = 1}^{K^{(n+1)}} \gamma_{\mathrm{E}, j \leftarrow i}^{(n+1)} \;
    P_i^{(n+1)}(t) \\
  &&{} + \Sum_{i = 1}^{K^{(n)}} \gamma_{\mathrm{R}, j \leftarrow i}^{(n)} \;
    P_i^{(n)}(t) \nonumber \\
  &&{} - \gamma_j^{(n)} \, P_j^{(n)}(t) \; , \nonumber
\end{eqnarray}
where the partial decay widths for electronic decay
are~$\gamma_{\mathrm{E}, j \leftarrow i}^{(n+1)} \geq 0$ and for radiative
decay are~$\gamma_{\mathrm{R}, j \leftarrow i}^{(n)} \geq 0$.
Note that $\gamma_{\mathrm{R}, j \leftarrow i}^{(n)} = 0$ for~$j \geq i$.
The total decay width of a system in configuration~$j$ with~$n$~electrons
to make a transition to any in this way accessible other configuration
is~$\gamma_j^{(n)} = \gamma_{\mathrm{E}, j}^{(n)} + \gamma_{\mathrm{R},
j}^{(n)}$ in analogy to~\eref{eq:sigtot} given by
\begin{equation}
  \label{eq:gamtot}
    \gamma_{\mathrm{E}, j}^{(n)} = \Sum_{i = 1}^{K^{(n-1)}}
    \gamma_{\mathrm{E}, i \leftarrow j}^{(n)}
   \hbox{\quad and\quad} \gamma_{\mathrm{R}, j}^{(n)} = \Sum_{i = 1}^{K^{(n)}}
    \gamma_{\mathrm{R}, i \leftarrow j}^{(n)} \; .
\end{equation}
The decay widths~$\gamma_{\mathrm{R}, j \leftarrow i}^{(n)}$ and
$\gamma_j^{(n)}$ make their first appearance in the rate equations
for~$n = N - 1$, \ie, Eq.~\eref{eq:reqsingly}, and
the~$\gamma_{\mathrm{E}, j \leftarrow i}^{(n)}$
start to appear in the rate equations for~$n = N-2$.
There is no decay in the first rate equation, for the system in its
neutral charge state which has the same form as~\eref{eq:reqground}.
The last rate equation, for the electron-bare system, similar
to~\eref{eq:reqbare}, also does not contain a decay width.

Introducing the notation of~\eref{eq:probvec} and \eref{eq:vecreq},
I express~\eref{eq:reqgam} compactly as
\begin{equation}
  \label{eq:vecreqgam}
  \dfrac{\differential \vec{\mathfrak p}(t)}{\differential t} = \mat \Sigma
    \; \vec {\mathfrak p}(t) \, J(t) + \mat \Gamma
    \; \vec {\mathfrak p}(t) \; .
\end{equation}
I have, additionally, the matrix of decay widths~$\mat \Gamma$.
Because of the radiative decay widths, for this matrix to be lower triangular
with the negative of the total decay widths on the diagonal,
the electron configurations need to be sorted with decreasing energy.
Further the first and last rows and column of~$\mat \Gamma$ are
filled with zeros;
the inner part shall be referred to by~$\mat \Gamma^{\,\prime}$,
\ie, $\mat \Gamma$~can be written as the direct sum~$\mat \Gamma =
(0) \oplus \mat \Gamma^{\,\prime} \oplus (0)$.
Due to the closure property~\eref{eq:gamtot}, $\mat \Gamma$~is a closed
matrix and the probability~\cite{Krengel:WS-05} in the system of rate
equations~\eref{eq:reqgam} is conserved in analogy to
Remarks~\ref{st:positivity} and \ref{st:probability} and Lemma~\ref{st:prob}.

\subsection{Decay equation}
\label{sec:soldecay}

Inspecting~\eref{eq:vecreqgam}, I find that, with decay widths present,
one cannot proceed any longer as for the case without decay widths.
Namely, as~$\mat \Sigma$ and $\mat \Gamma$, in general, do not commute,
the Lappo-Danilevskii condition~\eref{eq:lappo} is not fulfilled
for~$\mat \Sigma \, J(t) + \mat \Gamma$.
For the same reason, the two matrices do not share a common
eigenbasis,~\cite{Fischer:LA-14,SuppData} \ie, the matrix on the
right-hand side of~\eref{eq:vecreqgam} cannot be diagonalized by
a time-independent matrix as in~\eref{eq:diagrates}.
Certainly, \Eref{eq:vecreqgam} can be integrated numerically
as in Ref.~\onlinecite{Buth:UA-12} or solved using the analytical solutions
of~\sref{sec:succapprox} but here I would like to determine
an approximate solution of it.
I obtain the solution~\eref{eq:preqfl} of the rate equations
neglecting the decay widths~\eref{eq:vecreq} first.
Thereby, I assume that the \xray~pulse is so short that, in the course
of the interaction, no noticeable decay of excited states occurs.
Second, I solve~\eref{eq:vecreqgam} without cross sections, \ie,
the \emph{decay equation}:
\begin{equation}
  \label{eq:decayeqn}
  \dfrac{\differential \vec{\mathfrak p}(t)}{\differential t} = \mat \Gamma
    \; \vec {\mathfrak p}(t) \; ,
\end{equation}
by going to the eigenbasis of~$\mat \Gamma$.~\cite{Walter:GD-00}
Namely, the decay widths for different electronic configurations of the system
are pairwise distinct or zero;
then, according to Lemma~\ref{st:diag}, the matrix~$\mat \Gamma$ is
diagonalizable.
Thus I have~$\mat \Gamma = \mat{\mathfrak U} \, \mat M \,
\mat{\mathfrak U}^{-1}$ with the matrix of eigenvectors~$\mat{\mathfrak U}$
and the matrix of eigenvalues~$\mat M = \diag(\mu_1, \ldots, \mu_K)$.
I arrive at the decoupled system of differential equations of the decay
\begin{equation}
  \mat{\mathfrak U}^{-1} \, \dfrac{\differential \vec{\mathfrak p}(t)}
    {\differential t} = \dfrac{\differential \vec{\mathfrak q}(t)}
    {\differential t} = \mat M \, \vec {\mathfrak q}(t) \; ,
\end{equation}
which has the solution with~$1 \leq \ell \leq K$, in analogy
to~\eref{eq:decreq}, given by
\begin{equation}
  \label{eq:decreqsol}
  \mathfrak q_{\ell}(t) = \mathfrak q_{\ell}(\mathfrak T) \; \euler^{\mu_{\ell}
    \, (t - \mathfrak T)} \; .
\end{equation}
Here decay is assumed to begin at time~$\mathfrak T$.
The starting probabilities at time~$T$ are found from the solution of the
system of rate equations without decay~\eref{eq:preqfl} and \eref{eq:pertexp}.
They are $\vec{\mathfrak p}(\mathfrak T) = \vec p(T)
= \mat{\mathfrak U} \, \vec{\mathfrak q}(\mathfrak T) \iff
\vec{\mathfrak q}(\mathfrak T) = \mat{\mathfrak U}^{-1} \,
\vec{\mathfrak p}(\mathfrak T)
= \mat{\mathfrak U}^{-1} \, \mat U \, \vec q(T)$
where $\vec q(T)$~is given by~\eref{eq:decreq}.
\textit{Nota bene}, that not necessarily~$\mathfrak T = T$ holds.
I introduce this flexibility, to allow the decay equations
to be integrated also from a~$\mathfrak T$ as the onset time of decay
processes.

The solution of the decay equations~\eref{eq:decreqsol} is transformed
back from the eigenbasis of~$\mat \Gamma$ to the original basis via
\begin{eqnarray}
  \mathfrak p_i(t) &=& \Sum_{\ell = 1}^K \mathfrak U_{i\ell} \;
    \euler^{\mu_{\ell} \, (t - \mathfrak T)} \> \mathfrak q_{\ell}
    (\mathfrak T) \\
   &=& \Sum_{n = 0}^{\infty} \Sum_{\ell, j = 1}^K \mathfrak U_{i\ell} \,
     \mu_{\ell}^n \, (\mat{\mathfrak U}^{-1})_{\ell j} \,
     \dfrac{(t - \mathfrak T)^n}{n!} \, \mathfrak p_j(\mathfrak T)
     \; . \nonumber
\end{eqnarray}
I arrive at the compact form of the solution into which the solution
of the rate equations without decay~\eref{eq:pertexp} is inserted proving
\begin{theorem}
\label{st:decayeq}
The approximate solution of rate equations with decay
processes~\eref{eq:vecreqgam} in terms of a solution of the rate
equations without decay~\eref{eq:vecreq} in the time
interval~$[\tau ; T]$ and subsequently the decay equation~\eref{eq:decayeqn}
reads
\begin{equation}
  \label{eq:decaysolution}
  \vec{\mathfrak p}(t) = \euler^{\mats \Gamma \, (t - \mathfrak T)} \>
    \vec{\mathfrak p}(\mathfrak T)
  = \euler^{\mats \Gamma \, (t - \mathfrak T)} \; \euler^{\mats \Sigma
    \, \Phi(T)} \> \vec{p}(\tau) \; ,
\end{equation}
for times~$t \geq \mathfrak T$.
Photoionization lasts till time~$T$ and decay processes start at
time~$\mathfrak T$.
\end{theorem}
Theorem~\ref{st:decayeq} does not reduce to Theorem~\ref{st:linearity} upon
letting~$\mat \Gamma = \mat 0$ as the photoionization of the system
is assumed to be completed at time~$T$, \ie, the temporal progression of
photoionization does not enter~\eref{eq:decaysolution}.
This apparent weakness of Theorem~\ref{st:decayeq}, however,
can be used judiciously by replacing the solution of the nondecaying
system~$\euler^{\mats \Sigma \, \Phi(T)} \> \vec{p}(\tau)$ by a
(numerical) solution of~\eref{eq:vecreqgam} in the time
interval~$[\tau ; T]$.
In this case, the solution of Eq.~\eref{eq:decayeqn},
$\euler^{\mats \Gamma \, (t - \mathfrak T)} \> \vec{p}(T)$
for~$\mathfrak T = T$ and $t \geq T$, is the (numerically\nbh)exact solution
of~\eref{eq:vecreqgam}.
This procedure has been used to treat slow radiative decay to
calculate photon yields in Ref.~\onlinecite{Buth:NX-17}.

\subsection{Successive approximation}
\label{sec:succapprox}

A simple analytic solution of the interacting and decaying
system~\eref{eq:vecreqgam} in terms of the Peano-Baker series is well
known~\cite{Komlenko:FM-11,Baake:PB-11};
it is stated in
\begin{theorem}
\label{st:straight}
Given a system of rate equations with decay~\eref{eq:vecreqgam},
then its solution is
\begin{equation}
  \label{eq:succappsimple}
  \vec {\mathfrak p}(t) = \mat X(t) \, \vec {\mathfrak p}(\tau) \; ,
\end{equation}
for times~$t \geq \tau$ with the fundamental matrix~\cite{Komlenko:FM-11}:
\begin{eqnarray}
  \label{eq:timeevolsimple}
  &&{}\mat X(t) = \unitmatrix
    + \Int_{\tau}^t (\mat \Sigma \, J(t') + \mat \Gamma) \differential t' \\
  &&{} + \Int_{\tau}^t (\mat \Sigma \, J(t') + \mat \Gamma) \Int_{\tau}^{t'}
    (\mat \Sigma \, J(t'') + \mat \Gamma) \differential t'' \differential t'
    \nonumber \\
  &&{} + \ldots \; . \nonumber
\end{eqnarray}
\end{theorem}
\begin{proof}
The fundamental matrix~$\mat X(t)$ solves the initial-value
problem~$\tfrac{\differential \mats X(t)}{\differential t}
= (\mat \Sigma \, J(t) + \mat \Gamma) \, \mat X(t)$,
$\mat X(\tau) = \unitmatrix$.~\cite{Komlenko:FM-11}
This equation can be integrated leading to~$\mat X(t)
= \unitmatrix + \Int_{\tau}^t (\mat \Sigma \, J(t') + \mat \Gamma) \,
\mat X(t') \differential t'$;
it is a linear Volterra integral equation of the second
kind.~\cite{Burton:Vi-83,Linz:AN-85,Walter:GD-00,Bakushinskii:VE-11}
Its kernel~$\mat \Sigma \, J(t') + \mat \Gamma$ is continuous because
$J$~is assumed to be continuous.
Therefore, it can be solved uniquely by successive
approximation~\cite{Burton:Vi-83,Linz:AN-85,Walter:GD-00,Khvedelidze:SA-11}
yielding a continuous solution, \ie, by inserting the equation on the
right-hand side for~$\mat X(t')$.
This leads to~\eref{eq:timeevolsimple} and the solution of~\eref{eq:vecreqgam}
is~\eref{eq:succappsimple}.
\end{proof}
Clearly, the analytical solution~\eref{eq:succappsimple} of the coupled system
of differential equations~\eref{eq:vecreqgam}, despite its exactness,
is of limited value.
This is because the exponential functions that emerge in
Theorems~\ref{st:linearity} and \ref{st:decayeq} are, thereby,
expanded into a series.
Thus the solution~\eref{eq:succappsimple} can be expected to converge
slowly with the order taken into account in~\eref{eq:timeevolsimple}.
Yet, despite its limited use, Theorem~\ref{st:straight} reveals
the emergence of nonlinearity due to decay in the sequential
absorption of multiple x~rays.
This case is referred to as \emph{linked multiphoton absorption}.

To incorporate more of the physics of the problem into the analytical expression
for the solution, I change to the eigenbasis of~$\mat \Sigma$ as
for~\eref{eq:diagrates} which produces
\begin{equation}
  \label{eq:reqcoupled}
  \dfrac{\differential \vec{\mathfrak q}(t)}{\differential t} = \mat \Lambda
    \; \vec {\mathfrak q}(t) \, J(t) + \mat \Delta
    \; \vec {\mathfrak q}(t) \; ,
\end{equation}
with $\mat \Delta = \mat U^{-1} \, \mat \Gamma \, \mat U$.
Next I express~\eref{eq:reqcoupled} componentwise by
\begin{equation}
  \dfrac{\differential \mathfrak q_i(t)}{\differential t}
    = \lambda_i \; \mathfrak q_i(t) \, J(t) + \Sum_{j=1}^K
    \Delta_{ij} \; \mathfrak q_j(t) \; ,
\end{equation}
for~$i \in \{1, \ldots, K\}$.
Letting~$\mathfrak q_i(t) = \euler^{\lambda_i \, \Phi(t)} \, \mathfrak r_i(t)$,
I have
\begin{equation}
  \dfrac{\differential \mathfrak r_i(t)}{\differential t}
    = \Sum_{j=1}^K \euler^{-\lambda_i \, \Phi(t)} \, \Delta_{ij} \,
    \euler^{\lambda_j \, \Phi(t)} \; \mathfrak r_j(t)
    = \Sum_{j=1}^K \Xi_{ij} \; \mathfrak r_j(t) \; ,
\end{equation}
or in vector notation
\begin{equation}
  \label{eq:intEOM}
  \dfrac{\differential \vec{\mathfrak r}(t)}{\differential t} = \mat \Xi(t) \>
    \vec{\mathfrak r}(t) \; .
\end{equation}
The time-dependent matrix~$\mat \Xi(t)$ can be expressed succinctly by
\begin{equation}
  \label{eq:timepert}
  \mat \Xi(t) = \euler^{-\mats \Lambda \, \Phi(t)} \> \mat \Delta \>
    \euler^{\mats \Lambda \, \Phi(t)}
    = \mat U^{-1} \> \euler^{-\mats \Sigma \, \Phi(t)} \> \mat \Gamma \>
    \euler^{\mats \Sigma \, \Phi(t)} \> \mat U \; .
\end{equation}
\Eref{eq:intEOM} can be integrated on both sides leading to
\begin{equation}
  \label{eq:intEOMintegral}
  \vec{\mathfrak r}(t) = \vec{\mathfrak r}(\tau) + \Int_{\tau}^t
    \mat \Xi(t') \> \vec{\mathfrak r}(t') \differential t' \; .
\end{equation}
This is a linear Volterra integral equation of the second
kind.~\cite{Burton:Vi-83,Linz:AN-85,Walter:GD-00,Bakushinskii:VE-11}
Its kernel~$\mat \Xi(t')$ is continuous because
the \xray~fluence~\eref{eq:fluence} is a continuous function with time.
Then the unique solution of Eq.~\eref{eq:intEOMintegral} is continuous as well.
Using Eq.~\eref{eq:timepert}, I rewrite Eq.~\eref{eq:intEOMintegral}
in the original basis
\begin{equation}
  \label{eq:intEOMintegralOrig}
  \vec{\mathfrak p}(t) = \euler^{\mats \Sigma \, \Phi(t)} \>
    \vec{\mathfrak p}(\tau) + \euler^{\mats \Sigma \, \Phi(t)} \Int_{\tau}^t
    \euler^{-\mats \Sigma \, \Phi(t')} \> \mat \Gamma \>
    \vec{\mathfrak p}(t') \differential t' \; ,
\end{equation}
where I realize that
\begin{equation}
  \vec {\mathfrak q}(t) = \euler^{\mats \Lambda \, \Phi(t)} \>
    \vec{\mathfrak r}(t) = \mat U^{-1} \, \vec {\mathfrak p}(t) \; ,
\end{equation}
holds with the initial condition~$\vec {\mathfrak q}(\tau)
= \vec{\mathfrak r}(\tau) = \mat U^{-1} \, \vec {\mathfrak p}(\tau)$.

Volterra equations such as Eq.~\eref{eq:intEOMintegral} with a continuous
kernel can be solved uniquely by successive approximation~\cite{Burton:Vi-83,%
Linz:AN-85,Walter:GD-00,Khvedelidze:SA-11}---\ie,
inserting~\eref{eq:intEOMintegral} for~$\vec{\mathfrak r}(t')$
in~\eref{eq:intEOMintegral}---reading
\begin{equation}
  \label{eq:intEOMtimeev}
  \vec{\mathfrak r}(t) = \mat{\Cal U}(t,\tau) \> \vec{\mathfrak r}(\tau)
    \; ,
\end{equation}
introducing the time-evolution matrix via
\begin{eqnarray}
  \label{eq:timeevmat}
  \mat{\Cal U}(t,\tau) &=& \unitmatrix + \Int_{\tau}^t \mat \Xi(t') \>
    \differential t' \\
  &&{} + \Int_{\tau}^t \mat \Xi(t') \> \Int_{\tau}^{t'} \mat \Xi(t'') \>
    \differential t'' \differential t' + \ldots \; , \nonumber
\end{eqnarray}
which transforms the initial condition~$\vec{\mathfrak r}(\tau)$
at time~$\tau$ into the solution~$\vec{\mathfrak r}(t)$ at time~$t \geq \tau$.
This is the matrix that needs to be evaluated in order to
solve~\eref{eq:intEOMtimeev}.
Expressing~\eref{eq:intEOMtimeev} and \eref{eq:timeevmat} in
the original basis, I prove that the solution of the Volterra
integral equation~\eref{eq:intEOMintegralOrig} can be written
concisely as stated in
\begin{theorem}
\label{st:volterra}
Given a system of rate equations with decay~\eref{eq:vecreqgam},
then its solution is
\begin{equation}
  \label{eq:succapp}
  \vec {\mathfrak p}(t) = \euler^{\mats \Sigma \, \Phi(t)} \>
    \mat U \, \mat{\Cal U}(t,\tau) \, \mat U^{-1} \,
    \vec {\mathfrak p}(\tau) \; ,
\end{equation}
for times~$t \geq \tau$ with the time-evolution matrix~\eref{eq:timeevmat} which
can be expressed via~\eref{eq:timepert} as
\begin{eqnarray}
  \label{eq:timeevol}
  &&{}\mat U \, \mat{\Cal U}(t,\tau) \, \mat U^{-1} = \unitmatrix
    + \Int_{\tau}^t \euler^{-\mats \Sigma \, \Phi(t')} \> \mat \Gamma \>
    \euler^{\mats \Sigma \, \Phi(t')} \differential t' \\
  &&{} + \Int_{\tau}^t \euler^{-\mats \Sigma \, \Phi(t')} \> \mat \Gamma \>
    \euler^{\mats \Sigma \, \Phi(t')} \Int_{\tau}^{t'}
    \euler^{-\mats \Sigma \, \Phi(t'')} \> \mat \Gamma \>
    \euler^{\mats \Sigma \, \Phi(t'')}
    \differential t'' \differential t' \nonumber \\
  &&{} + \ldots \; . \nonumber
\end{eqnarray}
\end{theorem}
The final result~\eref{eq:timeevol} clearly exhibits the role of
decay processes for turning the interaction with x~rays nonlinear
in contrast to the case without decay in Theorem~\ref{st:linearity}.
Theorem~\ref{st:volterra} becomes Theorem~\ref{st:linearity} upon
letting~$\mat \Gamma = \mat 0$.

The line of arguments that constitutes the proof of Theorem~\ref{st:volterra}
can be applied also to the complementary case that one
transforms~\eref{eq:vecreqgam} to the eigenbasis of~$\mat \Gamma$
[see Eq.~\eref{eq:decayeqn} and the ensuing discussion] instead of
using Eq.~\eref{eq:reqcoupled}.
Then a Volterra equation analogous to Eq.~\eref{eq:intEOMintegralOrig}
is derived reading
\begin{equation}
  \label{eq:intEOMintegralOrigdec}
  \vec{\mathfrak p}(t) = \euler^{\mats \Gamma \, (t - \tau)} \>
    \vec{\mathfrak p}(\tau) + \euler^{\mats \Gamma \, (t - \tau)} \Int_{\tau}^t
    \euler^{-\mats \Gamma \, (t' - \tau)} \> \mat \Sigma \, J(t') \>
    \vec{\mathfrak p}(t') \differential t' \; ,
\end{equation}
which is solved by successive
approximation~\cite{Burton:Vi-83,Linz:AN-85,Walter:GD-00,Khvedelidze:SA-11}
proving
\begin{theorem}
\label{st:volterradec}
Given a system of rate equations with decay~\eref{eq:vecreqgam},
then its solution is
\begin{equation}
  \vec {\mathfrak p}(t) = \euler^{\mats \Gamma \, (t - \tau)} \>
    \mat{\mathfrak U} \, \mat{\Cal U}'(t,\tau) \, \mat{\mathfrak U}^{-1} \,
    \vec {\mathfrak p}(\tau) \; ,
\end{equation}
for times~$t \geq \tau$ with the time-evolution matrix~$\mat{\Cal U}'(t,\tau)$
[analogous to Eq.~\eref{eq:timeevmat}] which can be expressed as
\begin{eqnarray}
  \label{eq:timeevoldec}
  \mat{\mathfrak U} \, \mat{\Cal U}'(t,\tau)&& \, \mat{\mathfrak U}^{-1}
    = \unitmatrix  + \Int_{\tau}^t \euler^{-\mats \Gamma \, (t' - \tau)}
    \> \mat \Sigma \, J(t') \> \euler^{\mats \Gamma \, (t' - \tau)}
    \differential t' \nonumber \\
  &&{} + \Int_{\tau}^t \euler^{-\mats \Gamma \, (t' - \tau)}
    \> \mat \Sigma \, J(t') \> \euler^{\mats \Gamma \, (t' - \tau)} \nonumber \\
  &&{} \times \Int_{\tau}^{t'} \euler^{-\mats \Gamma \, (t'' - \tau)}
    \> \mat \Sigma \, J(t'') \> \euler^{\mats \Gamma \, (t'' - \tau)}
    \differential t'' \differential t' \nonumber \\
  &&{} + \ldots \; .
\end{eqnarray}
\end{theorem}
As before, the analytical solution of the rate
equations~\eref{eq:timeevoldec} showcases the impact of
decay for rendering the absorption of x~rays nonlinear.
Theorem~\ref{st:volterradec} becomes Theorem~\ref{st:decayeq} by
approximatively neglecting decay in the course of the interaction
with the x~rays, \ie, by letting~$\mat \Gamma = \mat 0$
in~$\mat{\Cal U}'(T,\tau)$ for a fixed time~$T$ and
letting~$\mathfrak T = \tau$.

The derivation that has lead to Theorems~\ref{st:volterra} and
\ref{st:volterradec} of this subsection resembles in many aspects the derivation
of the time-dependent perturbation series in quantum field theory on
pages~56--58 of Ref.~\onlinecite{Fetter:MP-71}.
Thereby, Eq.~\eref{eq:reqcoupled} is structurally similar to the time-dependent
Schr\"odinger equation for an interacting quantum system, however,
with the crucial difference that in our case the time derivative
is not multiplied by the imaginary unit.
Then $\mat \Xi(t)$ [\eref{eq:timepert}]~stands for the time-dependent
perturbation in the interaction picture and
$\mat{\Cal U}(t,\tau)$~corresponds to the time-evolution operator.
Note that for Theorem~\ref{st:volterra}, in contrast to time-dependent
perturbation theory, the exactly solvable part [first term on the
right-hand side of Eq.~\eref{eq:reqcoupled}] is time dependent via~$J(t)$.
In time-dependent perturbation theory, one proceeds by
introducing time-ordered products and expressing the
series~\eref{eq:timeevol} in terms of an exponential
function.~\cite{Fetter:MP-71}

\section{Results and discussion}
\label{sec:results}

I illustrate the theory of the previous Secs.~\ref{sec:ratewodecay} and
\ref{sec:decayeqn} by applying it to the rate equations of a nitrogen
atom, which has $N=7$~electrons, in LCLS radiation.
All details on the atomic electronic structure and the rate equations
can be found in Ref.~\onlinecite{Buth:UA-12}.
I assume an \xray~photon energy of~$840 \eV$ and a fluence
of~$8 \E{11} \U{\tfrac{photons}{\mu m^2}}$.
The temporal pulse shape of the \xray~flux is Gaussian (See Eq.~(7) in
Ref.~\onlinecite{Buth:UA-12}) and the spatial dependence is assumed
to be constant.
Unless stated otherwise, the \xray~pulse has a full width at half
maximum~(FWHM) duration of~$2 \U{fs}$.
I calculate the probability to find the atomic cationic state with
charge~$j \in \{ 0, \ldots, N \}$ via
\begin{equation}
  W_j(t) = \Sum_{i=1}^{K^{(N-j)}} P_i^{(N-j)}(t) \; .
\end{equation}
This quantity allows one to calculate the experimentally mensurable
ion yields from which follows the average charge state that is an
indicator of the amount of charge found on the
ions.~\cite{Buth:UA-12,Liu:RE-16,Buth:NX-17}

\begin{figure}
  \includegraphics[clip,width=\hsize]{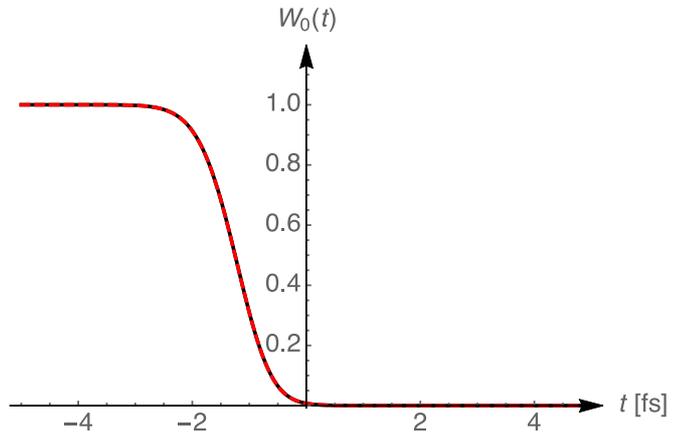}
  \caption{(Color) Ground-state depletion~$W_0(t)$ at time~$t$ of a nitrogen
           atom with the
           numerically-exact solution (solid-black lines) and
           the linearity theorem~\eref{eq:pertexp} (dashed-red lines).}
  \label{fig:ground_linth}
\end{figure}

In \fref{fig:ground_linth}, I plot the numerically-exact solution of the
rate equations for a nitrogen atom together with the result from
the linearity theorem~\eref{eq:pertexp}.
Both curves are indistinguishable which is not surprising as the
rate-equation for ground-state depletion~\eref{eq:reqground} does
not contain any decay terms.

\begin{figure}
  \includegraphics[clip,width=\hsize]{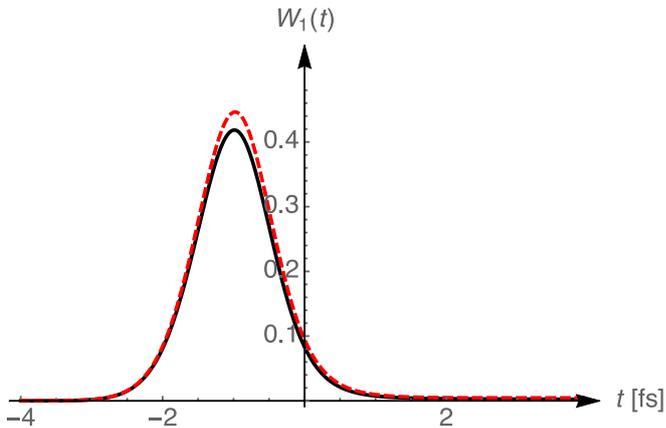}
  \caption{(Color) Single ionization~$W_1(t)$ at time~$t$ of a nitrogen
           atom.
           Line styles as in \fref{fig:ground_linth}.}
  \label{fig:singly_linth}
\end{figure}

Single ionization is investigated in \fref{fig:singly_linth} and, again,
a good agreement with the reference curve is found.
However, this time, there is a decay term in the rate
equations~\eref{eq:reqgam} which leads to a easily discernible deviation
between both curves at the peak.
The fact that the long-term behavior of both curves agrees so well
can only be ascribed to the fact that further ionization of
the atom leads to the observed decay and not Auger or radiative
decay as they are not described by the linearity theorem~\eref{eq:pertexp}.

\begin{figure}
  \includegraphics[clip,width=\hsize]{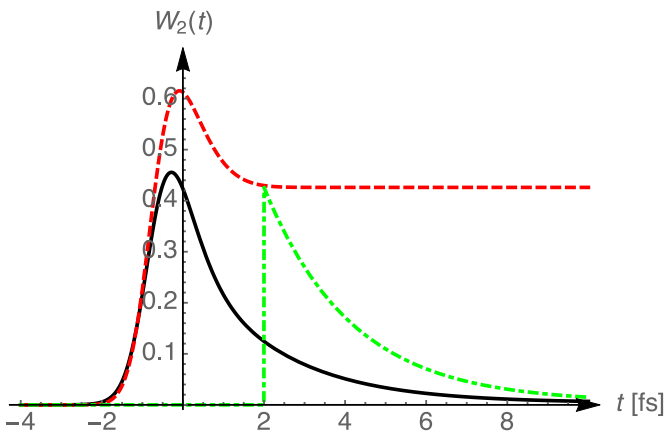}
  \caption{(Color) Double ionization~$W_2(t)$ at time~$t$ of a nitrogen
           atom with the numerically-exact solution (solid-black lines),
           the linearity theorem~\eref{eq:pertexp} (dashed-red lines),
           and the decay equation~\eref{eq:decayeqn} (dotdashed-green lines).}
  \label{fig:doubly_dec2fs}
\end{figure}

\begin{figure}
  \includegraphics[clip,width=\hsize]{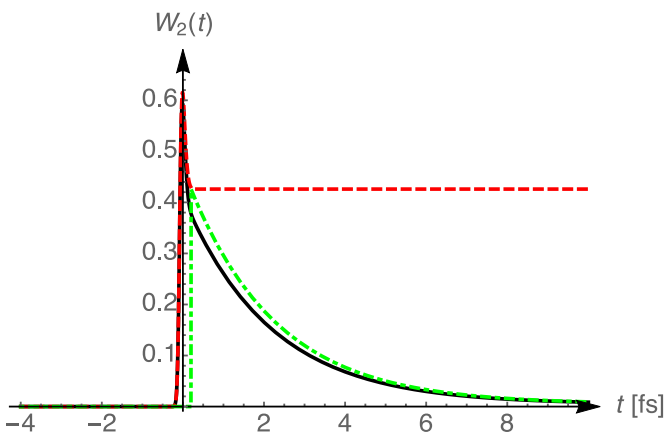}
  \caption{(Color) Double ionization~$W_2(t)$ as in \fref{fig:doubly_dec2fs}
           but for an \xray~pulse with a FWHM duration of~$200 \U{as}$.}
  \label{fig:doubly_dec200as}
\end{figure}

The decay equations~\eref{eq:decayeqn} are used in \fref{fig:doubly_dec2fs}
and \fref{fig:doubly_dec200as} for \xray~pulses with a FWHM duration
of~$2 \U{fs}$ and $200 \U{as}$, respectively, to examine double ionization.
Thereby, I choose the time when the decay processes start~$\mathfrak T$ to be
the FWHM duration but the time~$T$ at which the linearity
theorem~\eref{eq:pertexp} is evaluated is chosen larger than~$\mathfrak T$.
Clearly, the curve from the decay equations in \fref{fig:doubly_dec2fs}
shows a significant improvement over the result from the linearity
theorem but does not approximate the numerically-exact solution acceptably.
This changes dramatically if the FWHM pulse duration is decreased
by an order of magnitude to~$200 \U{as}$ in \fref{fig:doubly_dec200as}
where a good agreement is achieved.
After the pulse is over the double ionization probability rises on
the time scale of Auger decay of core holes of~$6.7 \U{fs}$~\cite{Buth:UA-12}
which turns singly ionized nitrogen atoms into doubly ionized ones.
Noteworthy is that this good agreement seen for double ionization
does not remain so for charge states of~N$^{4+}$ and higher~\cite{SuppData}.
For them, to be accurately described by the decay equation, one needs
to shorten the pulse duration even further in order to reduce
the interplay between photoionization and decay processes.
However, I need to stress that even assuming a $200 \U{as}$~pulse
is stretching the validity of the rate-equation approximation too much.
For such short pulses, coherence effects become important.~\cite{Li:CR-16}
The upshot of the analysis of the decay equations~\eref{eq:decayeqn} in
this paragraph is that they are of limited use in the \xray~regime.

\begin{figure}
  \includegraphics[clip,width=\hsize]{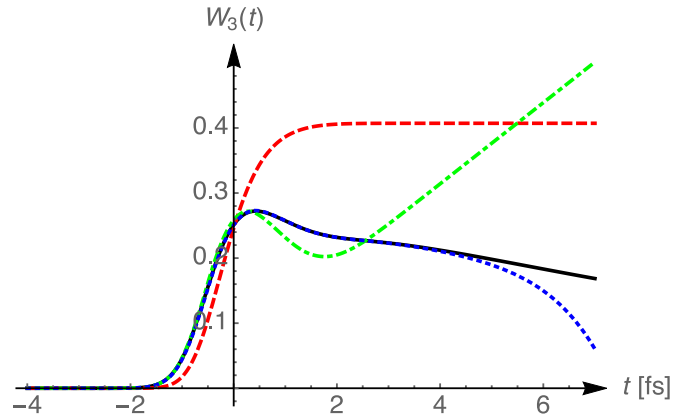}
  \caption{(Color) Triple ionization~$W_3(t)$ at time~$t$
           of a nitrogen atom with the numerically-exact solution
           (solid-black lines), the linearity theorem~\eref{eq:pertexp}
           (dashed-red lines), and the successive
           approximation~\eref{eq:succapp} in first (dash-dotted-green)
           and eighth order (dotted-blue).}
  \label{fig:third_succ}
\end{figure}

Triple ionization of a nitrogen atom is displayed in~\fref{fig:third_succ}
for the results from the linearity theorem, the successive approximation,
and the numerically-exact solution.
The linearity theorem produces a curve which only very crudely
follows the behavior of the numerically-exact solution.
Clearly, the absence of decay processes in the approximation
leads to the overpopulation of the triple-ionization channels
for~$t > 0$.
Furthermore, after the pulse is over, the probability remains constant
whereas the numerically-exact solution slopes downward due to decay
processes.
Taking decay into account in terms of the successive
approximation~\eref{eq:succapp} in first order leads to a good
agreement up to~$\approx 0.5 \U{fs}$.
Afterwards the approximation quickly deteriorates.
However, going to eighth order in the approximation of the
Volterra integral equation improves the result dramatically
such that it agrees with the numerically-exact solution
up to~$\approx 4.5 \U{fs}$.
Still higher orders in the successive approximation are
required to approximate the numerically-exact result beyond this time.

One needs to consider several terms in the successive approximation
to reach converged solutions of the Volterra integral equation.
As can be seen in~\fref{fig:third_succ}, the first-order produces
already good agreement for short times.
However, to approximate the numerically-exact solution for
longer times, higher orders are required.
This can be understood by considering very short pulses for
which the solution of the decay equation~\eref{eq:decayeqn}
is acceptable.
Then the iteration of the Volterra integral equation needs to reproduce
the terms of the series of the exponential function
with the decay matrix in~\eref{eq:decayeqn}
up to sufficiently high order to obtain good agreement.
The larger the time, the higher orders in the series need to be included.

\section{Conclusion}
\label{sec:conclusion}

I carry out a formal analysis of multiphoton absorption which is
called simultaneous, if it cannot be split into individual one- or
few-photon absorptions but has to be described by the full
expressions for the few-photon cross section.
If a rate equation approximation turns out to be satisfactory, I
can also distinguish linked multiphoton absorption---which are rendered
nonlinear by the decay of intermediate states---and separable
multiphoton absorption which depends only on the fluence the atom or molecule
was subjected to---just like one-photon absorption.
I turn to decay processes and analyze the situation under the assumption
that the \xray~pulse is so short that one needs not consider decay
in the course of it.
Then the ensuing decay can be treated independently and an analytic
expression is obtained for the probabilities to find the atom in
specific states after the pulse is over.
Finally, I solve the coupled rate equations that describe the joint
process of \xray~absorption and decay.
This is achieved by recasting the problem in terms of a Volterra integral
equation of the second kind which is amenable to a solution by
successive approximation.
I apply the equations to a nitrogen atom in LCLS x~rays which reveals
that decay processes are crucial for an accurate description and the decay
equation is acceptable only for very short pulses which makes it
unattractive for applications in the \xray~regime.
Yet for situations where the decay widths are considerably smaller
than the case studied here, \eg, in the ultraviolet
regime, the short-pulse
approximation [Theorem~\ref{st:decayeq}] shall be much better and
even Theorem~\ref{st:linearity} may find its use for approximating
the quantum dynamics on short time scales.
Clearly, if decay is absent, \eg, for the sequential absorption of
outer valence electrons in ultraviolet light,
Theorem~\ref{st:linearity} is the solution of the rate equations.
In any case, the successive approximation provides a reliable
and robust method to solve the rate equations.
Already the first-order approximation to the solution of the Volterra equation
is useful for small decay widths and limited duration, \eg, for computing
\xray~diffraction of ultrashort pulses.~\cite{Son:HA-11,*Son:EH-11}
However, higher orders need to be taken into account, if a longer time
evolution is desired.

The presented research opens up rich perspectives for future work.
The formulation of multiphoton absorption that has been devised here
can be used in formal developments and practical applications.
It provides a different viewpoint on the solution of the rate equations
because the expression for the solution incorporates the
underlying physics already in contrast to
the conventional methods for numerically-solving systems
of ordinary differential equations.

There are a two major challenges, however, involved in applying
Theorems~\ref{st:linearity}, \ref{st:decayeq}, \ref{st:volterra},
and \ref{st:volterradec} in practical computations.
First, the matrix exponentials need to be evaluated which requires
careful numerical analysis.~\cite{Golub:MC-96,Moler:DW-03}
Second, depending on the number of terms taken into account
in Eq.~\eref{eq:timeevol}, repeated numerical integrations are necessary.
This is, of course, only the case, if Eq.~\eref{eq:timeevol} is applied;
the mathematics of Volterra integral equations~\eref{eq:intEOMintegralOrig}
and \eref{eq:intEOMintegralOrigdec}
is well researched and there are alternative methods
available which may prove suitable for the problem at
hand.~\cite{Burton:Vi-83,Linz:AN-85,Bakushinskii:VE-11}
Both points need to be investigated further in order to turn
this approach into a viable method in practice.
Equation~\eref{eq:intEOMintegralOrigdec} has the advantage over
Eq.~\eref{eq:intEOMintegralOrig} that the integral needs to be evaluated
only for the duration of the \xray~pulse.
Specifically, the systems of rate equations for heavy atoms become
huge~\cite{Fukuzawa:DI-13,Ho:TT-14} such that a complete numerical
solution is rendered impracticable.
Therefore, a Monte Carlo method is used in
Refs.~\onlinecite{Fukuzawa:DI-13,Ho:TT-14} to solve the rate equations for which
all atomic quantities, \ie, cross sections and decay widths are computed
beforehand.~\cite{Rudek:UE–12,Son:MC-12,*Son:EM-15}
For very heavy atoms and ionization of deep inner shells, even this is not
enough as the computation of all involved atomic quantities becomes intractable.
Hence the Monte Carlo approach is extended to also decide
which quantities are to be calculated.~\cite{Fukuzawa:DI-13}
Theorems~\ref{st:linearity}, \ref{st:decayeq}, \ref{st:straight},
\ref{st:volterra}, and \ref{st:volterradec} provide a different
perspective on the problem.
Namely, the concept of Monte Carlo methods is replaced by
matrix methods, \ie, the repeated solution of the rate equations
involving random numbers is changed to an examination of matrix times
vector products for sparse matrices.~\cite{Golub:MC-96}
Decisions whether to compute certain atomic quantities need to
be made based on their relevance to these products where suitable
criteria need to be devised.

The rate equations in this work are restricted to non-resonant
one-\xray-absorption for simplicity and because it is the most
important case for practical applications which rely frequently
on \xray~diffraction.
This limitation, however, can be lifted straightforwardly.
One can include also REXMI terms,~\cite{Rudek:UE–12,Rudek:RE-13}
if the rate-equation approximation remains valid for these
cases.~\cite{Delone:MP-00}
The solution of the rate equations with decay, Theorems~\ref{st:volterra}
and \ref{st:volterradec}, can be extended, \eg, to include simultaneous
two-photon absorption~\cite{Goppert:EZ-31,Rohringer:XR-07,Doumy:NA-11}
by augmenting the decay part in~\eref{eq:vecreqgam} by appropriate terms.
The other way round is the case if there is no one-photon
absorption, \eg, due to a too low photon energy.
Then the analysis of Theorems~\ref{st:linearity}, \ref{st:decayeq},
\ref{st:straight}, \ref{st:volterra}, and \ref{st:volterradec}
can be done analogously starting with the lowest order of simultaneous
multiphoton absorption instead.
So far mostly configurations have been used in the literature
to write down the rate equations.
There is, however, no restriction that prevents one to use
fine-structure-resolved states instead.
This becomes relevant when heavier atoms are considered where relativistic
effects are prominent.~\cite{Rudek:UE–12,Fukuzawa:DI-13}

\begin{acknowledgments}
Thanks to Lorenz S.~Cederbaum, Simona Scheit, and Jochen Schirmer
for helpful discussions and a critical reading of the manuscript.
\end{acknowledgments}

\appendix
\section{Diagonalizability of a closed, lower triangular matrix}

\begin{lemma}
\label{st:diag}
Given a real, lower triangular $K \times K$~matrix~$\mat A$ whose nonzero
eigenvalues---the nonzero diagonal elements---are pairwise distinct;
the eigenvalue~0 may occur multiple times.
Furthermore, $\mat A$ is a closed essentially nonnegative matrix.
Then $\mat A$ is diagonalizable and the eigenvectors corresponding
to eigenvalues~0 are given by Cartesian eigenvectors~$\vec e_i$
where $i \in \{1, \ldots, K\}$ and the value~1 in these vectors is at
the row index~$i$ of the eigenvalue~0 of~$\mat A$.
\end{lemma}

\begin{proof}
As $\mat A$ is closed~\eref{eq:closure} and essentially nonnegative,
for $A_{ii} = 0$ follows that $A_{ji} = 0$ for all~$j$, \ie, the
entire column of~$\mat A$ is filled with zeros.
Consequently, I have
\begin{equation}
  \mat A \, \vec e_i = A_{ii} \, \vec e_i = \vec 0 \; .
\end{equation}
Thus for all eigenvalues~0 of~$\mat A$ I have found eigenvectors, \ie,
the algebraic multiplicity in the characteristic polynom due to eigenvalue~0
corresponds to its geometric multiplicity.
All other eigenvalues are pairwise distinct.
Hence the characteristic polynom contains only linear factors with respect to
these eigenvalues.
Thus $\mat A$~is diagonalizable.~\cite{Fischer:LA-14}
\end{proof}


\begin{thebibliography}{69}%
\makeatletter
\providecommand \@ifxundefined [1]{%
 \@ifx{#1\undefined}
}%
\providecommand \@ifnum [1]{%
 \ifnum #1\expandafter \@firstoftwo
 \else \expandafter \@secondoftwo
 \fi
}%
\providecommand \@ifx [1]{%
 \ifx #1\expandafter \@firstoftwo
 \else \expandafter \@secondoftwo
 \fi
}%
\providecommand \natexlab [1]{#1}%
\providecommand \enquote  [1]{``#1''}%
\providecommand \bibnamefont  [1]{#1}%
\providecommand \bibfnamefont [1]{#1}%
\providecommand \citenamefont [1]{#1}%
\providecommand \href@noop [0]{\@secondoftwo}%
\providecommand \href [0]{\begingroup \@sanitize@url \@href}%
\providecommand \@href[1]{\@@startlink{#1}\@@href}%
\providecommand \@@href[1]{\endgroup#1\@@endlink}%
\providecommand \@sanitize@url [0]{\catcode `\\12\catcode `\$12\catcode
  `\&12\catcode `\#12\catcode `\^12\catcode `\_12\catcode `\%12\relax}%
\providecommand \@@startlink[1]{}%
\providecommand \@@endlink[0]{}%
\providecommand \url  [0]{\begingroup\@sanitize@url \@url }%
\providecommand \@url [1]{\endgroup\@href {#1}{\urlprefix }}%
\providecommand \urlprefix  [0]{URL }%
\providecommand \Eprint [0]{\href }%
\providecommand \doibase [0]{http://dx.doi.org/}%
\providecommand \selectlanguage [0]{\@gobble}%
\providecommand \bibinfo  [0]{\@secondoftwo}%
\providecommand \bibfield  [0]{\@secondoftwo}%
\providecommand \translation [1]{[#1]}%
\providecommand \BibitemOpen [0]{}%
\providecommand \bibitemStop [0]{}%
\providecommand \bibitemNoStop [0]{.\EOS\space}%
\providecommand \EOS [0]{\spacefactor3000\relax}%
\providecommand \BibitemShut  [1]{\csname bibitem#1\endcsname}%
\let\auto@bib@innerbib\@empty
\bibitem [{\citenamefont {Walter}(2000)}]{Walter:GD-00}%
  \BibitemOpen
  \bibfield  {author} {\bibinfo {author} {\bibfnamefont {W.}~\bibnamefont
  {Walter}},\ }\href {\doibase 10.1007/978-3-642-57240-1} {\emph {\bibinfo
  {title} {Gew{\"o}hnliche {Differentialgleichungen}}}},\ \bibinfo {edition}
  {7th}\ ed.\ (\bibinfo  {publisher} {Springer-Verlag},\ \bibinfo {address}
  {Berlin, Heidelberg},\ \bibinfo {year} {2000})\BibitemShut {NoStop}%
\bibitem [{\citenamefont {Rohringer}\ and\ \citenamefont
  {Santra}(2007)}]{Rohringer:XR-07}%
  \BibitemOpen
  \bibfield  {author} {\bibinfo {author} {\bibfnamefont {N.}~\bibnamefont
  {Rohringer}}\ and\ \bibinfo {author} {\bibfnamefont {R.}~\bibnamefont
  {Santra}},\ }\bibfield  {title} {\enquote {\bibinfo {title} {X-ray nonlinear
  optical processes using a self-amplified spontaneous emission free-electron
  laser},}\ }\href {\doibase 10.1103/PhysRevA.76.033416} {\bibfield  {journal}
  {\bibinfo  {journal} {Phys. Rev. A}\ }\textbf {\bibinfo {volume} {76}},\
  \bibinfo {pages} {033416} (\bibinfo {year} {2007})}\BibitemShut {NoStop}%
\bibitem [{\citenamefont {Son}, \citenamefont {Young},\ and\ \citenamefont
  {Santra}(2011{\natexlab{a}})}]{Son:HA-11}%
  \BibitemOpen
  \bibfield  {author} {\bibinfo {author} {\bibfnamefont {S.-K.}\ \bibnamefont
  {Son}}, \bibinfo {author} {\bibfnamefont {L.}~\bibnamefont {Young}}, \ and\
  \bibinfo {author} {\bibfnamefont {R.}~\bibnamefont {Santra}},\ }\bibfield
  {title} {\enquote {\bibinfo {title} {Impact of hollow-atom formation on
  coherent x-ray scattering at high intensity},}\ }\href {\doibase
  10.1103/PhysRevA.83.033402} {\bibfield  {journal} {\bibinfo  {journal} {Phys.
  Rev. A}\ }\textbf {\bibinfo {volume} {83}},\ \bibinfo {pages} {033402}
  (\bibinfo {year} {2011}{\natexlab{a}})}\BibitemShut {NoStop}%
\bibitem [{\citenamefont {Son}, \citenamefont {Young},\ and\ \citenamefont
  {Santra}(2011{\natexlab{b}})}]{Son:EH-11}%
  \BibitemOpen
  \bibfield  {author} {\bibinfo {author} {\bibfnamefont {S.-K.}\ \bibnamefont
  {Son}}, \bibinfo {author} {\bibfnamefont {L.}~\bibnamefont {Young}}, \ and\
  \bibinfo {author} {\bibfnamefont {R.}~\bibnamefont {Santra}},\ }\bibfield
  {title} {\enquote {\bibinfo {title} {Erratum: {Impact} of hollow-atom
  formation on coherent x-ray scattering at high intensity [{Phys. Rev. A}
  \textbf{83}, 033402 (2011)]},}\ }\href {\doibase 10.1103/PhysRevA.83.069906}
  {\bibfield  {journal} {\bibinfo  {journal} {Phys. Rev. A}\ }\textbf {\bibinfo
  {volume} {83}},\ \bibinfo {pages} {069906(E)} (\bibinfo {year}
  {2011}{\natexlab{b}})}\BibitemShut {NoStop}%
\bibitem [{\citenamefont {Buth}\ \emph {et~al.}(2012)\citenamefont {Buth},
  \citenamefont {Liu}, \citenamefont {Chen}, \citenamefont {Cryan},
  \citenamefont {Fang}, \citenamefont {Glownia}, \citenamefont {Hoener},
  \citenamefont {Coffee},\ and\ \citenamefont {Berrah}}]{Buth:UA-12}%
  \BibitemOpen
  \bibfield  {author} {\bibinfo {author} {\bibfnamefont {C.}~\bibnamefont
  {Buth}}, \bibinfo {author} {\bibfnamefont {J.-C.}\ \bibnamefont {Liu}},
  \bibinfo {author} {\bibfnamefont {M.~H.}\ \bibnamefont {Chen}}, \bibinfo
  {author} {\bibfnamefont {J.~P.}\ \bibnamefont {Cryan}}, \bibinfo {author}
  {\bibfnamefont {L.}~\bibnamefont {Fang}}, \bibinfo {author} {\bibfnamefont
  {J.~M.}\ \bibnamefont {Glownia}}, \bibinfo {author} {\bibfnamefont
  {M.}~\bibnamefont {Hoener}}, \bibinfo {author} {\bibfnamefont {R.~N.}\
  \bibnamefont {Coffee}}, \ and\ \bibinfo {author} {\bibfnamefont
  {N.}~\bibnamefont {Berrah}},\ }\bibfield  {title} {\enquote {\bibinfo {title}
  {Ultrafast absorption of intense x~rays by nitrogen molecules},}\ }\href
  {\doibase 10.1063/1.4722756} {\bibfield  {journal} {\bibinfo  {journal} {J.
  Chem. Phys.}\ }\textbf {\bibinfo {volume} {136}},\ \bibinfo {pages} {214310}
  (\bibinfo {year} {2012})},\ \Eprint {http://arxiv.org/abs/1201.1896}
  {arXiv:1201.1896} \BibitemShut {NoStop}%
\bibitem [{\citenamefont {Son}\ and\ \citenamefont {Santra}(2012)}]{Son:MC-12}%
  \BibitemOpen
  \bibfield  {author} {\bibinfo {author} {\bibfnamefont {S.-K.}\ \bibnamefont
  {Son}}\ and\ \bibinfo {author} {\bibfnamefont {R.}~\bibnamefont {Santra}},\
  }\bibfield  {title} {\enquote {\bibinfo {title} {{Monte Carlo} calculation of
  ion, electron, and photon spectra of xenon atoms in x-ray free-electron laser
  pulses},}\ }\href {\doibase 10.1103/PhysRevA.85.063415} {\bibfield  {journal}
  {\bibinfo  {journal} {Phys. Rev. A}\ }\textbf {\bibinfo {volume} {85}},\
  \bibinfo {pages} {063415} (\bibinfo {year} {2012})}\BibitemShut {NoStop}%
\bibitem [{\citenamefont {Son}\ and\ \citenamefont {Santra}(2015)}]{Son:EM-15}%
  \BibitemOpen
  \bibfield  {author} {\bibinfo {author} {\bibfnamefont {S.-K.}\ \bibnamefont
  {Son}}\ and\ \bibinfo {author} {\bibfnamefont {R.}~\bibnamefont {Santra}},\
  }\bibfield  {title} {\enquote {\bibinfo {title} {Erratum: {Monte Carlo}
  calculation of ion, electron, and photon spectra of xenon atoms in x-ray
  free-electron laser pulses [{Phys. Rev. A} \textbf{85}, 063415 (2012)]},}\
  }\href {\doibase 10.1103/PhysRevA.92.039906} {\bibfield  {journal} {\bibinfo
  {journal} {Phys. Rev. A}\ }\textbf {\bibinfo {volume} {92}},\ \bibinfo
  {pages} {039906(E)} (\bibinfo {year} {2015})}\BibitemShut {NoStop}%
\bibitem [{\citenamefont {Rudek}\ \emph {et~al.}(2012)\citenamefont {Rudek},
  \citenamefont {Son}, \citenamefont {Foucar}, \citenamefont {Epp},
  \citenamefont {Erk}, \citenamefont {Hartmann}, \citenamefont {Adolph},
  \citenamefont {Andritschke}, \citenamefont {Aquila}, \citenamefont {Berrah},
  \citenamefont {Bostedt}, \citenamefont {Bozek}, \citenamefont {Coppola},
  \citenamefont {Filsinger}, \citenamefont {Gorke}, \citenamefont {Gorkhover},
  \citenamefont {Graafsma}, \citenamefont {Gumprecht}, \citenamefont
  {Hartmann}, \citenamefont {Hauser}, \citenamefont {Herrmann}, \citenamefont
  {Hirsemann}, \citenamefont {Holl}, \citenamefont {Homke}, \citenamefont
  {Journel}, \citenamefont {Kaiser}, \citenamefont {Kimmel}, \citenamefont
  {Krasniqi}, \citenamefont {Kuhnel}, \citenamefont {Matysek}, \citenamefont
  {Messerschmidt}, \citenamefont {Miesner}, \citenamefont {Moller},
  \citenamefont {Moshammer}, \citenamefont {Nagaya}, \citenamefont {Nilsson},
  \citenamefont {Potdevin}, \citenamefont {Pietschner}, \citenamefont {Reich},
  \citenamefont {Rupp}, \citenamefont {Schaller}, \citenamefont {Schlichting},
  \citenamefont {Schmidt}, \citenamefont {Schopper}, \citenamefont {Schorb},
  \citenamefont {Schr{\"o}ter}, \citenamefont {Schulz}, \citenamefont {Simon},
  \citenamefont {Soltau}, \citenamefont {Struder}, \citenamefont {Ueda},
  \citenamefont {Weidenspointner}, \citenamefont {Santra}, \citenamefont
  {Ullrich}, \citenamefont {Rudenko},\ and\ \citenamefont
  {Rolles}}]{Rudek:UE–12}%
  \BibitemOpen
  \bibfield  {author} {\bibinfo {author} {\bibfnamefont {B.}~\bibnamefont
  {Rudek}}, \bibinfo {author} {\bibfnamefont {S.-K.}\ \bibnamefont {Son}},
  \bibinfo {author} {\bibfnamefont {L.}~\bibnamefont {Foucar}}, \bibinfo
  {author} {\bibfnamefont {S.~W.}\ \bibnamefont {Epp}}, \bibinfo {author}
  {\bibfnamefont {B.}~\bibnamefont {Erk}}, \bibinfo {author} {\bibfnamefont
  {R.}~\bibnamefont {Hartmann}}, \bibinfo {author} {\bibfnamefont
  {M.}~\bibnamefont {Adolph}}, \bibinfo {author} {\bibfnamefont
  {R.}~\bibnamefont {Andritschke}}, \bibinfo {author} {\bibfnamefont
  {A.}~\bibnamefont {Aquila}}, \bibinfo {author} {\bibfnamefont
  {N.}~\bibnamefont {Berrah}}, \bibinfo {author} {\bibfnamefont
  {C.}~\bibnamefont {Bostedt}}, \bibinfo {author} {\bibfnamefont
  {J.}~\bibnamefont {Bozek}}, \bibinfo {author} {\bibfnamefont
  {N.}~\bibnamefont {Coppola}}, \bibinfo {author} {\bibfnamefont
  {F.}~\bibnamefont {Filsinger}}, \bibinfo {author} {\bibfnamefont
  {H.}~\bibnamefont {Gorke}}, \bibinfo {author} {\bibfnamefont
  {T.}~\bibnamefont {Gorkhover}}, \bibinfo {author} {\bibfnamefont
  {H.}~\bibnamefont {Graafsma}}, \bibinfo {author} {\bibfnamefont
  {L.}~\bibnamefont {Gumprecht}}, \bibinfo {author} {\bibfnamefont
  {A.}~\bibnamefont {Hartmann}}, \bibinfo {author} {\bibfnamefont
  {G.}~\bibnamefont {Hauser}}, \bibinfo {author} {\bibfnamefont
  {S.}~\bibnamefont {Herrmann}}, \bibinfo {author} {\bibfnamefont
  {H.}~\bibnamefont {Hirsemann}}, \bibinfo {author} {\bibfnamefont
  {P.}~\bibnamefont {Holl}}, \bibinfo {author} {\bibfnamefont {A.}~\bibnamefont
  {Homke}}, \bibinfo {author} {\bibfnamefont {L.}~\bibnamefont {Journel}},
  \bibinfo {author} {\bibfnamefont {C.}~\bibnamefont {Kaiser}}, \bibinfo
  {author} {\bibfnamefont {N.}~\bibnamefont {Kimmel}}, \bibinfo {author}
  {\bibfnamefont {F.}~\bibnamefont {Krasniqi}}, \bibinfo {author}
  {\bibfnamefont {K.-U.}\ \bibnamefont {Kuhnel}}, \bibinfo {author}
  {\bibfnamefont {M.}~\bibnamefont {Matysek}}, \bibinfo {author} {\bibfnamefont
  {M.}~\bibnamefont {Messerschmidt}}, \bibinfo {author} {\bibfnamefont
  {D.}~\bibnamefont {Miesner}}, \bibinfo {author} {\bibfnamefont
  {T.}~\bibnamefont {Moller}}, \bibinfo {author} {\bibfnamefont
  {R.}~\bibnamefont {Moshammer}}, \bibinfo {author} {\bibfnamefont
  {K.}~\bibnamefont {Nagaya}}, \bibinfo {author} {\bibfnamefont
  {B.}~\bibnamefont {Nilsson}}, \bibinfo {author} {\bibfnamefont
  {G.}~\bibnamefont {Potdevin}}, \bibinfo {author} {\bibfnamefont
  {D.}~\bibnamefont {Pietschner}}, \bibinfo {author} {\bibfnamefont
  {C.}~\bibnamefont {Reich}}, \bibinfo {author} {\bibfnamefont
  {D.}~\bibnamefont {Rupp}}, \bibinfo {author} {\bibfnamefont {G.}~\bibnamefont
  {Schaller}}, \bibinfo {author} {\bibfnamefont {I.}~\bibnamefont
  {Schlichting}}, \bibinfo {author} {\bibfnamefont {C.}~\bibnamefont
  {Schmidt}}, \bibinfo {author} {\bibfnamefont {F.}~\bibnamefont {Schopper}},
  \bibinfo {author} {\bibfnamefont {S.}~\bibnamefont {Schorb}}, \bibinfo
  {author} {\bibfnamefont {C.-D.}\ \bibnamefont {Schr{\"o}ter}}, \bibinfo
  {author} {\bibfnamefont {J.}~\bibnamefont {Schulz}}, \bibinfo {author}
  {\bibfnamefont {M.}~\bibnamefont {Simon}}, \bibinfo {author} {\bibfnamefont
  {H.}~\bibnamefont {Soltau}}, \bibinfo {author} {\bibfnamefont
  {L.}~\bibnamefont {Struder}}, \bibinfo {author} {\bibfnamefont
  {K.}~\bibnamefont {Ueda}}, \bibinfo {author} {\bibfnamefont {G.}~\bibnamefont
  {Weidenspointner}}, \bibinfo {author} {\bibfnamefont {R.}~\bibnamefont
  {Santra}}, \bibinfo {author} {\bibfnamefont {J.}~\bibnamefont {Ullrich}},
  \bibinfo {author} {\bibfnamefont {A.}~\bibnamefont {Rudenko}}, \ and\
  \bibinfo {author} {\bibfnamefont {D.}~\bibnamefont {Rolles}},\ }\bibfield
  {title} {\enquote {\bibinfo {title} {Ultra-efficient ionization of heavy
  atoms by intense {X}-ray free-electron laser pulses},}\ }\href {\doibase
  10.1038/nphoton.2012.261} {\bibfield  {journal} {\bibinfo  {journal} {Nat.
  Photon.}\ }\textbf {\bibinfo {volume} {6}},\ \bibinfo {pages} {858--865}
  (\bibinfo {year} {2012})}\BibitemShut {NoStop}%
\bibitem [{\citenamefont {Rudek}\ \emph {et~al.}(2013)\citenamefont {Rudek},
  \citenamefont {Rolles}, \citenamefont {Son}, \citenamefont {Foucar},
  \citenamefont {Erk}, \citenamefont {Epp}, \citenamefont {Boll}, \citenamefont
  {Anielski}, \citenamefont {Bostedt}, \citenamefont {Schorb}, \citenamefont
  {Coffee}, \citenamefont {Bozek}, \citenamefont {Trippel}, \citenamefont
  {Marchenko}, \citenamefont {Simon}, \citenamefont {Christensen},
  \citenamefont {De}, \citenamefont {Wada}, \citenamefont {Ueda}, \citenamefont
  {Schlichting}, \citenamefont {Santra}, \citenamefont {Ullrich},\ and\
  \citenamefont {Rudenko}}]{Rudek:RE-13}%
  \BibitemOpen
  \bibfield  {author} {\bibinfo {author} {\bibfnamefont {B.}~\bibnamefont
  {Rudek}}, \bibinfo {author} {\bibfnamefont {D.}~\bibnamefont {Rolles}},
  \bibinfo {author} {\bibfnamefont {S.-K.}\ \bibnamefont {Son}}, \bibinfo
  {author} {\bibfnamefont {L.}~\bibnamefont {Foucar}}, \bibinfo {author}
  {\bibfnamefont {B.}~\bibnamefont {Erk}}, \bibinfo {author} {\bibfnamefont
  {S.}~\bibnamefont {Epp}}, \bibinfo {author} {\bibfnamefont {R.}~\bibnamefont
  {Boll}}, \bibinfo {author} {\bibfnamefont {D.}~\bibnamefont {Anielski}},
  \bibinfo {author} {\bibfnamefont {C.}~\bibnamefont {Bostedt}}, \bibinfo
  {author} {\bibfnamefont {S.}~\bibnamefont {Schorb}}, \bibinfo {author}
  {\bibfnamefont {R.}~\bibnamefont {Coffee}}, \bibinfo {author} {\bibfnamefont
  {J.}~\bibnamefont {Bozek}}, \bibinfo {author} {\bibfnamefont
  {S.}~\bibnamefont {Trippel}}, \bibinfo {author} {\bibfnamefont
  {T.}~\bibnamefont {Marchenko}}, \bibinfo {author} {\bibfnamefont
  {M.}~\bibnamefont {Simon}}, \bibinfo {author} {\bibfnamefont
  {L.}~\bibnamefont {Christensen}}, \bibinfo {author} {\bibfnamefont
  {S.}~\bibnamefont {De}}, \bibinfo {author} {\bibfnamefont {S.-i.}\
  \bibnamefont {Wada}}, \bibinfo {author} {\bibfnamefont {K.}~\bibnamefont
  {Ueda}}, \bibinfo {author} {\bibfnamefont {I.}~\bibnamefont {Schlichting}},
  \bibinfo {author} {\bibfnamefont {R.}~\bibnamefont {Santra}}, \bibinfo
  {author} {\bibfnamefont {J.}~\bibnamefont {Ullrich}}, \ and\ \bibinfo
  {author} {\bibfnamefont {A.}~\bibnamefont {Rudenko}},\ }\bibfield  {title}
  {\enquote {\bibinfo {title} {Resonance-enhanced multiple ionization of
  krypton at an x-ray free-electron laser},}\ }\href {\doibase
  10.1103/PhysRevA.87.023413} {\bibfield  {journal} {\bibinfo  {journal} {Phys.
  Rev. A}\ }\textbf {\bibinfo {volume} {87}},\ \bibinfo {pages} {023413}
  (\bibinfo {year} {2013})}\BibitemShut {NoStop}%
\bibitem [{\citenamefont {Fukuzawa}\ \emph {et~al.}(2013)\citenamefont
  {Fukuzawa}, \citenamefont {Son}, \citenamefont {Motomura}, \citenamefont
  {Mondal}, \citenamefont {Nagaya}, \citenamefont {Wada}, \citenamefont {Liu},
  \citenamefont {Feifel}, \citenamefont {Tachibana}, \citenamefont {Ito},
  \citenamefont {Kimura}, \citenamefont {Sakai}, \citenamefont {Matsunami},
  \citenamefont {Hayashita}, \citenamefont {Kajikawa}, \citenamefont
  {Johnsson}, \citenamefont {Siano}, \citenamefont {Kukk}, \citenamefont
  {Rudek}, \citenamefont {Erk}, \citenamefont {Foucar}, \citenamefont {Robert},
  \citenamefont {Miron}, \citenamefont {Tono}, \citenamefont {Inubushi},
  \citenamefont {Hatsui}, \citenamefont {Yabashi}, \citenamefont {Yao},
  \citenamefont {Santra},\ and\ \citenamefont {Ueda}}]{Fukuzawa:DI-13}%
  \BibitemOpen
  \bibfield  {author} {\bibinfo {author} {\bibfnamefont {H.}~\bibnamefont
  {Fukuzawa}}, \bibinfo {author} {\bibfnamefont {S.-K.}\ \bibnamefont {Son}},
  \bibinfo {author} {\bibfnamefont {K.}~\bibnamefont {Motomura}}, \bibinfo
  {author} {\bibfnamefont {S.}~\bibnamefont {Mondal}}, \bibinfo {author}
  {\bibfnamefont {K.}~\bibnamefont {Nagaya}}, \bibinfo {author} {\bibfnamefont
  {S.}~\bibnamefont {Wada}}, \bibinfo {author} {\bibfnamefont {X.-J.}\
  \bibnamefont {Liu}}, \bibinfo {author} {\bibfnamefont {R.}~\bibnamefont
  {Feifel}}, \bibinfo {author} {\bibfnamefont {T.}~\bibnamefont {Tachibana}},
  \bibinfo {author} {\bibfnamefont {Y.}~\bibnamefont {Ito}}, \bibinfo {author}
  {\bibfnamefont {M.}~\bibnamefont {Kimura}}, \bibinfo {author} {\bibfnamefont
  {T.}~\bibnamefont {Sakai}}, \bibinfo {author} {\bibfnamefont
  {K.}~\bibnamefont {Matsunami}}, \bibinfo {author} {\bibfnamefont
  {H.}~\bibnamefont {Hayashita}}, \bibinfo {author} {\bibfnamefont
  {J.}~\bibnamefont {Kajikawa}}, \bibinfo {author} {\bibfnamefont
  {P.}~\bibnamefont {Johnsson}}, \bibinfo {author} {\bibfnamefont
  {M.}~\bibnamefont {Siano}}, \bibinfo {author} {\bibfnamefont
  {E.}~\bibnamefont {Kukk}}, \bibinfo {author} {\bibfnamefont {B.}~\bibnamefont
  {Rudek}}, \bibinfo {author} {\bibfnamefont {B.}~\bibnamefont {Erk}}, \bibinfo
  {author} {\bibfnamefont {L.}~\bibnamefont {Foucar}}, \bibinfo {author}
  {\bibfnamefont {E.}~\bibnamefont {Robert}}, \bibinfo {author} {\bibfnamefont
  {C.}~\bibnamefont {Miron}}, \bibinfo {author} {\bibfnamefont
  {K.}~\bibnamefont {Tono}}, \bibinfo {author} {\bibfnamefont {Y.}~\bibnamefont
  {Inubushi}}, \bibinfo {author} {\bibfnamefont {T.}~\bibnamefont {Hatsui}},
  \bibinfo {author} {\bibfnamefont {M.}~\bibnamefont {Yabashi}}, \bibinfo
  {author} {\bibfnamefont {M.}~\bibnamefont {Yao}}, \bibinfo {author}
  {\bibfnamefont {R.}~\bibnamefont {Santra}}, \ and\ \bibinfo {author}
  {\bibfnamefont {K.}~\bibnamefont {Ueda}},\ }\bibfield  {title} {\enquote
  {\bibinfo {title} {Deep inner-shell multiphoton ionization by intense x-ray
  free-electron laser pulses},}\ }\href {\doibase
  10.1103/PhysRevLett.110.173005} {\bibfield  {journal} {\bibinfo  {journal}
  {Phys. Rev. Lett.}\ }\textbf {\bibinfo {volume} {110}},\ \bibinfo {pages}
  {173005} (\bibinfo {year} {2013})}\BibitemShut {NoStop}%
\bibitem [{\citenamefont {Ho}\ \emph {et~al.}(2014)\citenamefont {Ho},
  \citenamefont {Bostedt}, \citenamefont {Schorb},\ and\ \citenamefont
  {Young}}]{Ho:TT-14}%
  \BibitemOpen
  \bibfield  {author} {\bibinfo {author} {\bibfnamefont {P.~J.}\ \bibnamefont
  {Ho}}, \bibinfo {author} {\bibfnamefont {C.}~\bibnamefont {Bostedt}},
  \bibinfo {author} {\bibfnamefont {S.}~\bibnamefont {Schorb}}, \ and\ \bibinfo
  {author} {\bibfnamefont {L.}~\bibnamefont {Young}},\ }\bibfield  {title}
  {\enquote {\bibinfo {title} {Theoretical tracking of resonance-enhanced
  multiple ionization pathways in x-ray free-electron laser pulses},}\ }\href
  {\doibase 10.1103/PhysRevLett.113.253001} {\bibfield  {journal} {\bibinfo
  {journal} {Phys. Rev. Lett.}\ }\textbf {\bibinfo {volume} {113}},\ \bibinfo
  {pages} {253001} (\bibinfo {year} {2014})}\BibitemShut {NoStop}%
\bibitem [{\citenamefont {Liu}\ \emph {et~al.}(2016)\citenamefont {Liu},
  \citenamefont {Berrah}, \citenamefont {Cederbaum}, \citenamefont {Cryan},
  \citenamefont {Glownia}, \citenamefont {Schafer},\ and\ \citenamefont
  {Buth}}]{Liu:RE-16}%
  \BibitemOpen
  \bibfield  {author} {\bibinfo {author} {\bibfnamefont {J.-C.}\ \bibnamefont
  {Liu}}, \bibinfo {author} {\bibfnamefont {N.}~\bibnamefont {Berrah}},
  \bibinfo {author} {\bibfnamefont {L.~S.}\ \bibnamefont {Cederbaum}}, \bibinfo
  {author} {\bibfnamefont {J.~P.}\ \bibnamefont {Cryan}}, \bibinfo {author}
  {\bibfnamefont {J.~M.}\ \bibnamefont {Glownia}}, \bibinfo {author}
  {\bibfnamefont {K.~J.}\ \bibnamefont {Schafer}}, \ and\ \bibinfo {author}
  {\bibfnamefont {C.}~\bibnamefont {Buth}},\ }\bibfield  {title} {\enquote
  {\bibinfo {title} {Rate equations for nitrogen molecules in ultrashort and
  intense x-ray pulses},}\ }\href {\doibase 10.1088/0953-4075/49/7/075602}
  {\bibfield  {journal} {\bibinfo  {journal} {J. Phys. B}\ }\textbf {\bibinfo
  {volume} {49}},\ \bibinfo {pages} {075602} (\bibinfo {year} {2016})},\
  \Eprint {http://arxiv.org/abs/1508.05223} {arXiv:1508.05223} \BibitemShut
  {NoStop}%
\bibitem [{\citenamefont {Buth}\ \emph {et~al.}(2017)\citenamefont {Buth},
  \citenamefont {Beerwerth}, \citenamefont {Obaid}, \citenamefont {Berrah},
  \citenamefont {Cederbaum},\ and\ \citenamefont {Fritzsche}}]{Buth:NX-17}%
  \BibitemOpen
  \bibfield  {author} {\bibinfo {author} {\bibfnamefont {C.}~\bibnamefont
  {Buth}}, \bibinfo {author} {\bibfnamefont {R.}~\bibnamefont {Beerwerth}},
  \bibinfo {author} {\bibfnamefont {R.}~\bibnamefont {Obaid}}, \bibinfo
  {author} {\bibfnamefont {N.}~\bibnamefont {Berrah}}, \bibinfo {author}
  {\bibfnamefont {L.~S.}\ \bibnamefont {Cederbaum}}, \ and\ \bibinfo {author}
  {\bibfnamefont {S.}~\bibnamefont {Fritzsche}},\ }\bibfield  {title} {\enquote
  {\bibinfo {title} {Neon in ultrashort and intense x~rays from free-electron
  lasers},}\ }\href@noop {} {\bibfield  {journal} {\bibinfo  {journal}
  {submitted}\ } (\bibinfo {year} {2017})},\ \Eprint
  {http://arxiv.org/abs/1705.07521} {arXiv:1705.07521} \BibitemShut {NoStop}%
\bibitem [{\citenamefont {Arthur}\ \emph {et~al.}(2002)\citenamefont {Arthur},
  \citenamefont {Anfinrud}, \citenamefont {Audebert}, \citenamefont {Bane},
  \citenamefont {Ben-Zvi}, \citenamefont {Bharadwaj}, \citenamefont {Bionta},
  \citenamefont {Bolton}, \citenamefont {Borland}, \citenamefont {Bucksbaum},
  \citenamefont {Cauble}, \citenamefont {Clendenin}, \citenamefont
  {Cornacchia}, \citenamefont {Decker}, \citenamefont {Den~Hartog},
  \citenamefont {Dierker}, \citenamefont {Dowell}, \citenamefont {Dungan},
  \citenamefont {Emma}, \citenamefont {Evans}, \citenamefont {Faigel},
  \citenamefont {Falcone}, \citenamefont {Fawley}, \citenamefont {Ferrario},
  \citenamefont {Fisher}, \citenamefont {Freeman}, \citenamefont {Frisch},
  \citenamefont {Galayda}, \citenamefont {Gauthier}, \citenamefont {Gierman},
  \citenamefont {Gluskin}, \citenamefont {Graves}, \citenamefont {Hajdu},
  \citenamefont {Hastings}, \citenamefont {Hodgson}, \citenamefont {Huang},
  \citenamefont {Humphry}, \citenamefont {Ilinski}, \citenamefont {Imre},
  \citenamefont {Jacobsen}, \citenamefont {Kao}, \citenamefont {Kase},
  \citenamefont {Kim}, \citenamefont {Kirby}, \citenamefont {Kirz},
  \citenamefont {Klaisner}, \citenamefont {Krejcik}, \citenamefont {Kulander},
  \citenamefont {Landen}, \citenamefont {Lee}, \citenamefont {Lewis},
  \citenamefont {Limborg}, \citenamefont {Lindau}, \citenamefont {Lumpkin},
  \citenamefont {Materlik}, \citenamefont {Mao}, \citenamefont {Miao},
  \citenamefont {Mochrie}, \citenamefont {Moog}, \citenamefont {Milton},
  \citenamefont {Mulhollan}, \citenamefont {Nelson}, \citenamefont {Nelson},
  \citenamefont {Neutze}, \citenamefont {Ng}, \citenamefont {Nguyen},
  \citenamefont {Nuhn}, \citenamefont {Palmer}, \citenamefont {Paterson},
  \citenamefont {Pellegrini}, \citenamefont {Reiche}, \citenamefont {Renner},
  \citenamefont {Riley}, \citenamefont {Robinson}, \citenamefont {Rokni},
  \citenamefont {Rose}, \citenamefont {Rosenzweig}, \citenamefont {Ruland},
  \citenamefont {Ruocco}, \citenamefont {Saenz}, \citenamefont {Sasaki},
  \citenamefont {Sayre}, \citenamefont {Schmerge}, \citenamefont {Schneider},
  \citenamefont {Schroeder}, \citenamefont {Serafini}, \citenamefont {Sette},
  \citenamefont {Sinha}, \citenamefont {van~der Spoel}, \citenamefont
  {Stephenson}, \citenamefont {Stupakov}, \citenamefont {Sutton}, \citenamefont
  {Sz{\"o}ke}, \citenamefont {Tatchyn}, \citenamefont {Toor}, \citenamefont
  {Trakhtenberg}, \citenamefont {Vasserman}, \citenamefont {Vinokurov},
  \citenamefont {Wang}, \citenamefont {Waltz}, \citenamefont {Wark},
  \citenamefont {Weckert}, \citenamefont {{Wilson-Squire Group}}, \citenamefont
  {Winick}, \citenamefont {Woodley}, \citenamefont {Wootton}, \citenamefont
  {Wulff}, \citenamefont {Xie}, \citenamefont {Yotam}, \citenamefont {Young},\
  and\ \citenamefont {Zewail}}]{LCLS:CDR-02}%
  \BibitemOpen
  \bibfield  {author} {\bibinfo {author} {\bibfnamefont {J.}~\bibnamefont
  {Arthur}}, \bibinfo {author} {\bibfnamefont {P.}~\bibnamefont {Anfinrud}},
  \bibinfo {author} {\bibfnamefont {P.}~\bibnamefont {Audebert}}, \bibinfo
  {author} {\bibfnamefont {K.}~\bibnamefont {Bane}}, \bibinfo {author}
  {\bibfnamefont {I.}~\bibnamefont {Ben-Zvi}}, \bibinfo {author} {\bibfnamefont
  {V.}~\bibnamefont {Bharadwaj}}, \bibinfo {author} {\bibfnamefont
  {R.}~\bibnamefont {Bionta}}, \bibinfo {author} {\bibfnamefont
  {P.}~\bibnamefont {Bolton}}, \bibinfo {author} {\bibfnamefont
  {M.}~\bibnamefont {Borland}}, \bibinfo {author} {\bibfnamefont {P.~H.}\
  \bibnamefont {Bucksbaum}}, \bibinfo {author} {\bibfnamefont {R.~C.}\
  \bibnamefont {Cauble}}, \bibinfo {author} {\bibfnamefont {J.}~\bibnamefont
  {Clendenin}}, \bibinfo {author} {\bibfnamefont {M.}~\bibnamefont
  {Cornacchia}}, \bibinfo {author} {\bibfnamefont {G.}~\bibnamefont {Decker}},
  \bibinfo {author} {\bibfnamefont {P.}~\bibnamefont {Den~Hartog}}, \bibinfo
  {author} {\bibfnamefont {S.}~\bibnamefont {Dierker}}, \bibinfo {author}
  {\bibfnamefont {D.}~\bibnamefont {Dowell}}, \bibinfo {author} {\bibfnamefont
  {D.}~\bibnamefont {Dungan}}, \bibinfo {author} {\bibfnamefont
  {P.}~\bibnamefont {Emma}}, \bibinfo {author} {\bibfnamefont {I.}~\bibnamefont
  {Evans}}, \bibinfo {author} {\bibfnamefont {G.}~\bibnamefont {Faigel}},
  \bibinfo {author} {\bibfnamefont {R.}~\bibnamefont {Falcone}}, \bibinfo
  {author} {\bibfnamefont {W.~M.}\ \bibnamefont {Fawley}}, \bibinfo {author}
  {\bibfnamefont {M.}~\bibnamefont {Ferrario}}, \bibinfo {author}
  {\bibfnamefont {A.~S.}\ \bibnamefont {Fisher}}, \bibinfo {author}
  {\bibfnamefont {R.~R.}\ \bibnamefont {Freeman}}, \bibinfo {author}
  {\bibfnamefont {J.}~\bibnamefont {Frisch}}, \bibinfo {author} {\bibfnamefont
  {J.}~\bibnamefont {Galayda}}, \bibinfo {author} {\bibfnamefont {J.-C.}\
  \bibnamefont {Gauthier}}, \bibinfo {author} {\bibfnamefont {S.}~\bibnamefont
  {Gierman}}, \bibinfo {author} {\bibfnamefont {E.}~\bibnamefont {Gluskin}},
  \bibinfo {author} {\bibfnamefont {W.}~\bibnamefont {Graves}}, \bibinfo
  {author} {\bibfnamefont {J.}~\bibnamefont {Hajdu}}, \bibinfo {author}
  {\bibfnamefont {J.}~\bibnamefont {Hastings}}, \bibinfo {author}
  {\bibfnamefont {K.}~\bibnamefont {Hodgson}}, \bibinfo {author} {\bibfnamefont
  {Z.}~\bibnamefont {Huang}}, \bibinfo {author} {\bibfnamefont
  {R.}~\bibnamefont {Humphry}}, \bibinfo {author} {\bibfnamefont
  {P.}~\bibnamefont {Ilinski}}, \bibinfo {author} {\bibfnamefont
  {D.}~\bibnamefont {Imre}}, \bibinfo {author} {\bibfnamefont {C.}~\bibnamefont
  {Jacobsen}}, \bibinfo {author} {\bibfnamefont {C.-C.}\ \bibnamefont {Kao}},
  \bibinfo {author} {\bibfnamefont {K.~R.}\ \bibnamefont {Kase}}, \bibinfo
  {author} {\bibfnamefont {K.-J.}\ \bibnamefont {Kim}}, \bibinfo {author}
  {\bibfnamefont {R.}~\bibnamefont {Kirby}}, \bibinfo {author} {\bibfnamefont
  {J.}~\bibnamefont {Kirz}}, \bibinfo {author} {\bibfnamefont {L.}~\bibnamefont
  {Klaisner}}, \bibinfo {author} {\bibfnamefont {P.}~\bibnamefont {Krejcik}},
  \bibinfo {author} {\bibfnamefont {K.}~\bibnamefont {Kulander}}, \bibinfo
  {author} {\bibfnamefont {O.~L.}\ \bibnamefont {Landen}}, \bibinfo {author}
  {\bibfnamefont {R.~W.}\ \bibnamefont {Lee}}, \bibinfo {author} {\bibfnamefont
  {C.}~\bibnamefont {Lewis}}, \bibinfo {author} {\bibfnamefont
  {C.}~\bibnamefont {Limborg}}, \bibinfo {author} {\bibfnamefont {E.~I.}\
  \bibnamefont {Lindau}}, \bibinfo {author} {\bibfnamefont {A.}~\bibnamefont
  {Lumpkin}}, \bibinfo {author} {\bibfnamefont {G.}~\bibnamefont {Materlik}},
  \bibinfo {author} {\bibfnamefont {S.}~\bibnamefont {Mao}}, \bibinfo {author}
  {\bibfnamefont {J.}~\bibnamefont {Miao}}, \bibinfo {author} {\bibfnamefont
  {S.}~\bibnamefont {Mochrie}}, \bibinfo {author} {\bibfnamefont
  {E.}~\bibnamefont {Moog}}, \bibinfo {author} {\bibfnamefont {S.}~\bibnamefont
  {Milton}}, \bibinfo {author} {\bibfnamefont {G.}~\bibnamefont {Mulhollan}},
  \bibinfo {author} {\bibfnamefont {K.}~\bibnamefont {Nelson}}, \bibinfo
  {author} {\bibfnamefont {W.~R.}\ \bibnamefont {Nelson}}, \bibinfo {author}
  {\bibfnamefont {R.}~\bibnamefont {Neutze}}, \bibinfo {author} {\bibfnamefont
  {A.}~\bibnamefont {Ng}}, \bibinfo {author} {\bibfnamefont {D.}~\bibnamefont
  {Nguyen}}, \bibinfo {author} {\bibfnamefont {H.-D.}\ \bibnamefont {Nuhn}},
  \bibinfo {author} {\bibfnamefont {D.~T.}\ \bibnamefont {Palmer}}, \bibinfo
  {author} {\bibfnamefont {J.~M.}\ \bibnamefont {Paterson}}, \bibinfo {author}
  {\bibfnamefont {C.}~\bibnamefont {Pellegrini}}, \bibinfo {author}
  {\bibfnamefont {S.}~\bibnamefont {Reiche}}, \bibinfo {author} {\bibfnamefont
  {M.}~\bibnamefont {Renner}}, \bibinfo {author} {\bibfnamefont
  {D.}~\bibnamefont {Riley}}, \bibinfo {author} {\bibfnamefont {C.~V.}\
  \bibnamefont {Robinson}}, \bibinfo {author} {\bibfnamefont {S.~H.}\
  \bibnamefont {Rokni}}, \bibinfo {author} {\bibfnamefont {S.~J.}\ \bibnamefont
  {Rose}}, \bibinfo {author} {\bibfnamefont {J.}~\bibnamefont {Rosenzweig}},
  \bibinfo {author} {\bibfnamefont {R.}~\bibnamefont {Ruland}}, \bibinfo
  {author} {\bibfnamefont {G.}~\bibnamefont {Ruocco}}, \bibinfo {author}
  {\bibfnamefont {D.}~\bibnamefont {Saenz}}, \bibinfo {author} {\bibfnamefont
  {S.}~\bibnamefont {Sasaki}}, \bibinfo {author} {\bibfnamefont
  {D.}~\bibnamefont {Sayre}}, \bibinfo {author} {\bibfnamefont
  {J.}~\bibnamefont {Schmerge}}, \bibinfo {author} {\bibfnamefont
  {D.}~\bibnamefont {Schneider}}, \bibinfo {author} {\bibfnamefont
  {C.}~\bibnamefont {Schroeder}}, \bibinfo {author} {\bibfnamefont
  {L.}~\bibnamefont {Serafini}}, \bibinfo {author} {\bibfnamefont
  {F.}~\bibnamefont {Sette}}, \bibinfo {author} {\bibfnamefont
  {S.}~\bibnamefont {Sinha}}, \bibinfo {author} {\bibfnamefont
  {D.}~\bibnamefont {van~der Spoel}}, \bibinfo {author} {\bibfnamefont
  {B.}~\bibnamefont {Stephenson}}, \bibinfo {author} {\bibfnamefont
  {G.}~\bibnamefont {Stupakov}}, \bibinfo {author} {\bibfnamefont
  {M.}~\bibnamefont {Sutton}}, \bibinfo {author} {\bibfnamefont
  {A.}~\bibnamefont {Sz{\"o}ke}}, \bibinfo {author} {\bibfnamefont
  {R.}~\bibnamefont {Tatchyn}}, \bibinfo {author} {\bibfnamefont
  {A.}~\bibnamefont {Toor}}, \bibinfo {author} {\bibfnamefont {E.}~\bibnamefont
  {Trakhtenberg}}, \bibinfo {author} {\bibfnamefont {I.}~\bibnamefont
  {Vasserman}}, \bibinfo {author} {\bibfnamefont {N.}~\bibnamefont
  {Vinokurov}}, \bibinfo {author} {\bibfnamefont {X.~J.}\ \bibnamefont {Wang}},
  \bibinfo {author} {\bibfnamefont {D.}~\bibnamefont {Waltz}}, \bibinfo
  {author} {\bibfnamefont {J.~S.}\ \bibnamefont {Wark}}, \bibinfo {author}
  {\bibfnamefont {E.}~\bibnamefont {Weckert}}, \bibinfo {author} {\bibnamefont
  {{Wilson-Squire Group}}}, \bibinfo {author} {\bibfnamefont {H.}~\bibnamefont
  {Winick}}, \bibinfo {author} {\bibfnamefont {M.}~\bibnamefont {Woodley}},
  \bibinfo {author} {\bibfnamefont {A.}~\bibnamefont {Wootton}}, \bibinfo
  {author} {\bibfnamefont {M.}~\bibnamefont {Wulff}}, \bibinfo {author}
  {\bibfnamefont {M.}~\bibnamefont {Xie}}, \bibinfo {author} {\bibfnamefont
  {R.}~\bibnamefont {Yotam}}, \bibinfo {author} {\bibfnamefont
  {L.}~\bibnamefont {Young}}, \ and\ \bibinfo {author} {\bibfnamefont
  {A.}~\bibnamefont {Zewail}},\ }\href
  {http://www-ssrl.slac.stanford.edu/lcls/cdr} {\enquote {\bibinfo {title}
  {{Linac Coherent Light Source (LCLS): Conceptual Design Report}},}\ }\bibinfo
  {type} {Tech. Rep.}\ \bibinfo {number} {SLAC-R-593, UC-414}\ (\bibinfo
  {institution} {Stanford Linear Accelerator Center (SLAC)},\ \bibinfo
  {address} {Menlo Park, California, USA},\ \bibinfo {year} {2002})\BibitemShut
  {NoStop}%
\bibitem [{\citenamefont {Emma}\ \emph {et~al.}(2010)\citenamefont {Emma},
  \citenamefont {Akre}, \citenamefont {Arthur}, \citenamefont {Bionta},
  \citenamefont {Bostedt}, \citenamefont {Bozek}, \citenamefont {Brachmann},
  \citenamefont {Bucksbaum}, \citenamefont {Coffee}, \citenamefont {Decker},
  \citenamefont {Ding}, \citenamefont {Dowell}, \citenamefont {Edstrom},
  \citenamefont {Fisher}, \citenamefont {Gilevich}, \citenamefont {Hastings},
  \citenamefont {Hays}, \citenamefont {Hering}, \citenamefont {Huang},
  \citenamefont {Iverson}, \citenamefont {Loos}, \citenamefont {Messerschmidt},
  \citenamefont {Miahnahri}, \citenamefont {Moeller}, \citenamefont {Nuhn},
  \citenamefont {Pile}, \citenamefont {Ratner}, \citenamefont {Rzepiela},
  \citenamefont {Schultz}, \citenamefont {Smith}, \citenamefont {Stefan},
  \citenamefont {Tompkins}, \citenamefont {Turner}, \citenamefont {Welch},
  \citenamefont {White}, \citenamefont {Wu}, \citenamefont {Yocky},\ and\
  \citenamefont {Galayda}}]{Emma:FL-10}%
  \BibitemOpen
  \bibfield  {author} {\bibinfo {author} {\bibfnamefont {P.}~\bibnamefont
  {Emma}}, \bibinfo {author} {\bibfnamefont {R.}~\bibnamefont {Akre}}, \bibinfo
  {author} {\bibfnamefont {J.}~\bibnamefont {Arthur}}, \bibinfo {author}
  {\bibfnamefont {R.}~\bibnamefont {Bionta}}, \bibinfo {author} {\bibfnamefont
  {C.}~\bibnamefont {Bostedt}}, \bibinfo {author} {\bibfnamefont
  {J.}~\bibnamefont {Bozek}}, \bibinfo {author} {\bibfnamefont
  {A.}~\bibnamefont {Brachmann}}, \bibinfo {author} {\bibfnamefont
  {P.}~\bibnamefont {Bucksbaum}}, \bibinfo {author} {\bibfnamefont
  {R.}~\bibnamefont {Coffee}}, \bibinfo {author} {\bibfnamefont {F.-J.}\
  \bibnamefont {Decker}}, \bibinfo {author} {\bibfnamefont {Y.}~\bibnamefont
  {Ding}}, \bibinfo {author} {\bibfnamefont {D.}~\bibnamefont {Dowell}},
  \bibinfo {author} {\bibfnamefont {S.}~\bibnamefont {Edstrom}}, \bibinfo
  {author} {\bibfnamefont {J.}~\bibnamefont {Fisher}, \bibfnamefont
  {A.~Frisch}}, \bibinfo {author} {\bibfnamefont {S.}~\bibnamefont {Gilevich}},
  \bibinfo {author} {\bibfnamefont {J.}~\bibnamefont {Hastings}}, \bibinfo
  {author} {\bibfnamefont {G.}~\bibnamefont {Hays}}, \bibinfo {author}
  {\bibfnamefont {P.}~\bibnamefont {Hering}}, \bibinfo {author} {\bibfnamefont
  {Z.}~\bibnamefont {Huang}}, \bibinfo {author} {\bibfnamefont
  {R.}~\bibnamefont {Iverson}}, \bibinfo {author} {\bibfnamefont
  {H.}~\bibnamefont {Loos}}, \bibinfo {author} {\bibfnamefont {M.}~\bibnamefont
  {Messerschmidt}}, \bibinfo {author} {\bibfnamefont {A.}~\bibnamefont
  {Miahnahri}}, \bibinfo {author} {\bibfnamefont {S.}~\bibnamefont {Moeller}},
  \bibinfo {author} {\bibfnamefont {H.-D.}\ \bibnamefont {Nuhn}}, \bibinfo
  {author} {\bibfnamefont {G.}~\bibnamefont {Pile}}, \bibinfo {author}
  {\bibfnamefont {D.}~\bibnamefont {Ratner}}, \bibinfo {author} {\bibfnamefont
  {J.}~\bibnamefont {Rzepiela}}, \bibinfo {author} {\bibfnamefont
  {D.}~\bibnamefont {Schultz}}, \bibinfo {author} {\bibfnamefont
  {T.}~\bibnamefont {Smith}}, \bibinfo {author} {\bibfnamefont
  {P.}~\bibnamefont {Stefan}}, \bibinfo {author} {\bibfnamefont
  {H.}~\bibnamefont {Tompkins}}, \bibinfo {author} {\bibfnamefont
  {J.}~\bibnamefont {Turner}}, \bibinfo {author} {\bibfnamefont
  {J.}~\bibnamefont {Welch}}, \bibinfo {author} {\bibfnamefont
  {W.}~\bibnamefont {White}}, \bibinfo {author} {\bibfnamefont
  {J.}~\bibnamefont {Wu}}, \bibinfo {author} {\bibfnamefont {G.}~\bibnamefont
  {Yocky}}, \ and\ \bibinfo {author} {\bibfnamefont {J.}~\bibnamefont
  {Galayda}},\ }\bibfield  {title} {\enquote {\bibinfo {title} {First lasing
  and operation of an {\aa}ngstrom-wavelength free-electron laser},}\ }\href
  {\doibase 10.1038/nphoton.2010.176} {\bibfield  {journal} {\bibinfo
  {journal} {Nature Photon.}\ }\textbf {\bibinfo {volume} {4}},\ \bibinfo
  {pages} {641--647} (\bibinfo {year} {2010})}\BibitemShut {NoStop}%
\bibitem [{\citenamefont {Ishikawa}\ \emph {et~al.}(2012)\citenamefont
  {Ishikawa}, \citenamefont {Aoyagi}, \citenamefont {Asaka}, \citenamefont
  {Asano}, \citenamefont {Azumi}, \citenamefont {Bizen}, \citenamefont {Ego},
  \citenamefont {Fukami}, \citenamefont {Fukui}, \citenamefont {Furukawa},
  \citenamefont {Goto}, \citenamefont {Hanaki}, \citenamefont {Hara},
  \citenamefont {Hasegawa}, \citenamefont {Hatsui}, \citenamefont {Higashiya},
  \citenamefont {Hirono}, \citenamefont {Hosoda}, \citenamefont {Ishii},
  \citenamefont {Inagaki}, \citenamefont {Inubushi}, \citenamefont {Itoga},
  \citenamefont {Joti}, \citenamefont {Kago}, \citenamefont {Kameshima},
  \citenamefont {Kimura}, \citenamefont {Kirihara}, \citenamefont {Kiyomichi},
  \citenamefont {Kobayashi}, \citenamefont {Kondo}, \citenamefont {Kudo},
  \citenamefont {Maesaka}, \citenamefont {Marechal}, \citenamefont {Masuda},
  \citenamefont {Matsubara}, \citenamefont {Matsumoto}, \citenamefont
  {Matsushita}, \citenamefont {Matsui}, \citenamefont {Nagasono}, \citenamefont
  {Nariyama}, \citenamefont {Ohashi}, \citenamefont {Ohata}, \citenamefont
  {Ohshima}, \citenamefont {Ono}, \citenamefont {Otake}, \citenamefont {Saji},
  \citenamefont {Sakurai}, \citenamefont {Sato}, \citenamefont {Sawada},
  \citenamefont {Seike}, \citenamefont {Shirasawa}, \citenamefont {Sugimoto},
  \citenamefont {Suzuki}, \citenamefont {Takahashi}, \citenamefont {Takebe},
  \citenamefont {Takeshita}, \citenamefont {Tamasaku}, \citenamefont {Tanaka},
  \citenamefont {Tanaka}, \citenamefont {Tanaka}, \citenamefont {Togashi},
  \citenamefont {Togawa}, \citenamefont {Tokuhisa}, \citenamefont {Tomizawa},
  \citenamefont {Tono}, \citenamefont {Wu}, \citenamefont {Yabashi},
  \citenamefont {Yamaga}, \citenamefont {Yamashita}, \citenamefont {Yanagida},
  \citenamefont {Zhang}, \citenamefont {Shintake}, \citenamefont {Kitamura},\
  and\ \citenamefont {Kumagai}}]{Ishikawa:CX-12}%
  \BibitemOpen
  \bibfield  {author} {\bibinfo {author} {\bibfnamefont {T.}~\bibnamefont
  {Ishikawa}}, \bibinfo {author} {\bibfnamefont {H.}~\bibnamefont {Aoyagi}},
  \bibinfo {author} {\bibfnamefont {T.}~\bibnamefont {Asaka}}, \bibinfo
  {author} {\bibfnamefont {Y.}~\bibnamefont {Asano}}, \bibinfo {author}
  {\bibfnamefont {N.}~\bibnamefont {Azumi}}, \bibinfo {author} {\bibfnamefont
  {T.}~\bibnamefont {Bizen}}, \bibinfo {author} {\bibfnamefont
  {H.}~\bibnamefont {Ego}}, \bibinfo {author} {\bibfnamefont {K.}~\bibnamefont
  {Fukami}}, \bibinfo {author} {\bibfnamefont {T.}~\bibnamefont {Fukui}},
  \bibinfo {author} {\bibfnamefont {Y.}~\bibnamefont {Furukawa}}, \bibinfo
  {author} {\bibfnamefont {S.}~\bibnamefont {Goto}}, \bibinfo {author}
  {\bibfnamefont {H.}~\bibnamefont {Hanaki}}, \bibinfo {author} {\bibfnamefont
  {T.}~\bibnamefont {Hara}}, \bibinfo {author} {\bibfnamefont {T.}~\bibnamefont
  {Hasegawa}}, \bibinfo {author} {\bibfnamefont {T.}~\bibnamefont {Hatsui}},
  \bibinfo {author} {\bibfnamefont {A.}~\bibnamefont {Higashiya}}, \bibinfo
  {author} {\bibfnamefont {T.}~\bibnamefont {Hirono}}, \bibinfo {author}
  {\bibfnamefont {N.}~\bibnamefont {Hosoda}}, \bibinfo {author} {\bibfnamefont
  {M.}~\bibnamefont {Ishii}}, \bibinfo {author} {\bibfnamefont
  {T.}~\bibnamefont {Inagaki}}, \bibinfo {author} {\bibfnamefont
  {Y.}~\bibnamefont {Inubushi}}, \bibinfo {author} {\bibfnamefont
  {T.}~\bibnamefont {Itoga}}, \bibinfo {author} {\bibfnamefont
  {Y.}~\bibnamefont {Joti}}, \bibinfo {author} {\bibfnamefont {M.}~\bibnamefont
  {Kago}}, \bibinfo {author} {\bibfnamefont {T.}~\bibnamefont {Kameshima}},
  \bibinfo {author} {\bibfnamefont {H.}~\bibnamefont {Kimura}}, \bibinfo
  {author} {\bibfnamefont {Y.}~\bibnamefont {Kirihara}}, \bibinfo {author}
  {\bibfnamefont {A.}~\bibnamefont {Kiyomichi}}, \bibinfo {author}
  {\bibfnamefont {T.}~\bibnamefont {Kobayashi}}, \bibinfo {author}
  {\bibfnamefont {C.}~\bibnamefont {Kondo}}, \bibinfo {author} {\bibfnamefont
  {T.}~\bibnamefont {Kudo}}, \bibinfo {author} {\bibfnamefont {H.}~\bibnamefont
  {Maesaka}}, \bibinfo {author} {\bibfnamefont {X.~M.}\ \bibnamefont
  {Marechal}}, \bibinfo {author} {\bibfnamefont {T.}~\bibnamefont {Masuda}},
  \bibinfo {author} {\bibfnamefont {S.}~\bibnamefont {Matsubara}}, \bibinfo
  {author} {\bibfnamefont {T.}~\bibnamefont {Matsumoto}}, \bibinfo {author}
  {\bibfnamefont {T.}~\bibnamefont {Matsushita}}, \bibinfo {author}
  {\bibfnamefont {S.}~\bibnamefont {Matsui}}, \bibinfo {author} {\bibfnamefont
  {M.}~\bibnamefont {Nagasono}}, \bibinfo {author} {\bibfnamefont
  {N.}~\bibnamefont {Nariyama}}, \bibinfo {author} {\bibfnamefont
  {H.}~\bibnamefont {Ohashi}}, \bibinfo {author} {\bibfnamefont
  {T.}~\bibnamefont {Ohata}}, \bibinfo {author} {\bibfnamefont
  {T.}~\bibnamefont {Ohshima}}, \bibinfo {author} {\bibfnamefont
  {S.}~\bibnamefont {Ono}}, \bibinfo {author} {\bibfnamefont {Y.}~\bibnamefont
  {Otake}}, \bibinfo {author} {\bibfnamefont {C.}~\bibnamefont {Saji}},
  \bibinfo {author} {\bibfnamefont {T.}~\bibnamefont {Sakurai}}, \bibinfo
  {author} {\bibfnamefont {T.}~\bibnamefont {Sato}}, \bibinfo {author}
  {\bibfnamefont {K.}~\bibnamefont {Sawada}}, \bibinfo {author} {\bibfnamefont
  {T.}~\bibnamefont {Seike}}, \bibinfo {author} {\bibfnamefont
  {K.}~\bibnamefont {Shirasawa}}, \bibinfo {author} {\bibfnamefont
  {T.}~\bibnamefont {Sugimoto}}, \bibinfo {author} {\bibfnamefont
  {S.}~\bibnamefont {Suzuki}}, \bibinfo {author} {\bibfnamefont
  {S.}~\bibnamefont {Takahashi}}, \bibinfo {author} {\bibfnamefont
  {H.}~\bibnamefont {Takebe}}, \bibinfo {author} {\bibfnamefont
  {K.}~\bibnamefont {Takeshita}}, \bibinfo {author} {\bibfnamefont
  {K.}~\bibnamefont {Tamasaku}}, \bibinfo {author} {\bibfnamefont
  {H.}~\bibnamefont {Tanaka}}, \bibinfo {author} {\bibfnamefont
  {R.}~\bibnamefont {Tanaka}}, \bibinfo {author} {\bibfnamefont
  {T.}~\bibnamefont {Tanaka}}, \bibinfo {author} {\bibfnamefont
  {T.}~\bibnamefont {Togashi}}, \bibinfo {author} {\bibfnamefont
  {K.}~\bibnamefont {Togawa}}, \bibinfo {author} {\bibfnamefont
  {A.}~\bibnamefont {Tokuhisa}}, \bibinfo {author} {\bibfnamefont
  {H.}~\bibnamefont {Tomizawa}}, \bibinfo {author} {\bibfnamefont
  {K.}~\bibnamefont {Tono}}, \bibinfo {author} {\bibfnamefont {S.}~\bibnamefont
  {Wu}}, \bibinfo {author} {\bibfnamefont {M.}~\bibnamefont {Yabashi}},
  \bibinfo {author} {\bibfnamefont {M.}~\bibnamefont {Yamaga}}, \bibinfo
  {author} {\bibfnamefont {A.}~\bibnamefont {Yamashita}}, \bibinfo {author}
  {\bibfnamefont {K.}~\bibnamefont {Yanagida}}, \bibinfo {author}
  {\bibfnamefont {C.}~\bibnamefont {Zhang}}, \bibinfo {author} {\bibfnamefont
  {T.}~\bibnamefont {Shintake}}, \bibinfo {author} {\bibfnamefont
  {H.}~\bibnamefont {Kitamura}}, \ and\ \bibinfo {author} {\bibfnamefont
  {N.}~\bibnamefont {Kumagai}},\ }\bibfield  {title} {\enquote {\bibinfo
  {title} {A compact {X}-ray free-electron laser emitting in the
  sub-\aa{}ngstrom region},}\ }\href {\doibase 10.1038/nphoton.2012.141}
  {\bibfield  {journal} {\bibinfo  {journal} {Nat. Photon.}\ }\textbf {\bibinfo
  {volume} {6}},\ \bibinfo {pages} {540--544} (\bibinfo {year}
  {2012})}\BibitemShut {NoStop}%
\bibitem [{\citenamefont {Ganter}(2012)}]{SwissFEL:CDR-12}%
  \BibitemOpen
  \bibinfo {editor} {\bibfnamefont {R.}~\bibnamefont {Ganter}},\ ed.,\ \href
  {http://www.psi.ch/swissfel/CurrentSwissFELPublicationsEN/SwissFEL_CDR_V20_23.04.12_small.pdf}
  {\emph {\bibinfo {title} {SwissFEL Conceptual Design Report}}},\ \bibinfo
  {number} {PSI Bericht Nr.~10-04}\ (\bibinfo  {publisher} {Paul Scherrer
  Institut~(PSI)},\ \bibinfo {address} {Villigen, Switzerland},\ \bibinfo
  {year} {2012})\ \bibinfo {note}
  {\url{http://www.psi.ch/swissfel/CurrentSwissFELPublicationsEN/SwissFEL_CDR_V20_23.04.12_small.pdf}}\BibitemShut
  {NoStop}%
\bibitem [{\citenamefont {Altarelli}\ \emph {et~al.}(2006)\citenamefont
  {Altarelli}, \citenamefont {Brinkmann}, \citenamefont {Chergui},
  \citenamefont {Decking}, \citenamefont {Dobson}, \citenamefont
  {D{\"u}sterer}, \citenamefont {Gr{\"u}bel}, \citenamefont {Graeff},
  \citenamefont {Graafsma}, \citenamefont {Hajdu}, \citenamefont {Marangos},
  \citenamefont {Pfl{\"u}ger}, \citenamefont {Redlin}, \citenamefont {Riley},
  \citenamefont {Robinson}, \citenamefont {Rossbach}, \citenamefont {Schwarz},
  \citenamefont {Tiedtke}, \citenamefont {Tschentscher}, \citenamefont
  {Vartaniants}, \citenamefont {Wabnitz}, \citenamefont {Weise}, \citenamefont
  {Wichmann}, \citenamefont {Witte}, \citenamefont {Wolf}, \citenamefont
  {Wulff},\ and\ \citenamefont {Yurkov}}]{Altarelli:TDR-06}%
  \BibitemOpen
  \bibinfo {editor} {\bibfnamefont {M.}~\bibnamefont {Altarelli}}, \bibinfo
  {editor} {\bibfnamefont {R.}~\bibnamefont {Brinkmann}}, \bibinfo {editor}
  {\bibfnamefont {M.}~\bibnamefont {Chergui}}, \bibinfo {editor} {\bibfnamefont
  {W.}~\bibnamefont {Decking}}, \bibinfo {editor} {\bibfnamefont
  {B.}~\bibnamefont {Dobson}}, \bibinfo {editor} {\bibfnamefont
  {S.}~\bibnamefont {D{\"u}sterer}}, \bibinfo {editor} {\bibfnamefont
  {G.}~\bibnamefont {Gr{\"u}bel}}, \bibinfo {editor} {\bibfnamefont
  {W.}~\bibnamefont {Graeff}}, \bibinfo {editor} {\bibfnamefont
  {H.}~\bibnamefont {Graafsma}}, \bibinfo {editor} {\bibfnamefont
  {J.}~\bibnamefont {Hajdu}}, \bibinfo {editor} {\bibfnamefont
  {J.}~\bibnamefont {Marangos}}, \bibinfo {editor} {\bibfnamefont
  {J.}~\bibnamefont {Pfl{\"u}ger}}, \bibinfo {editor} {\bibfnamefont
  {H.}~\bibnamefont {Redlin}}, \bibinfo {editor} {\bibfnamefont
  {D.}~\bibnamefont {Riley}}, \bibinfo {editor} {\bibfnamefont
  {I.}~\bibnamefont {Robinson}}, \bibinfo {editor} {\bibfnamefont
  {J.}~\bibnamefont {Rossbach}}, \bibinfo {editor} {\bibfnamefont
  {A.}~\bibnamefont {Schwarz}}, \bibinfo {editor} {\bibfnamefont
  {K.}~\bibnamefont {Tiedtke}}, \bibinfo {editor} {\bibfnamefont
  {T.}~\bibnamefont {Tschentscher}}, \bibinfo {editor} {\bibfnamefont
  {I.}~\bibnamefont {Vartaniants}}, \bibinfo {editor} {\bibfnamefont
  {H.}~\bibnamefont {Wabnitz}}, \bibinfo {editor} {\bibfnamefont
  {H.}~\bibnamefont {Weise}}, \bibinfo {editor} {\bibfnamefont
  {R.}~\bibnamefont {Wichmann}}, \bibinfo {editor} {\bibfnamefont
  {K.}~\bibnamefont {Witte}}, \bibinfo {editor} {\bibfnamefont
  {A.}~\bibnamefont {Wolf}}, \bibinfo {editor} {\bibfnamefont {M.}~\bibnamefont
  {Wulff}}, \ and\ \bibinfo {editor} {\bibfnamefont {M.}~\bibnamefont
  {Yurkov}},\ eds.,\ \href {\doibase doi:10.3204/DESY_06-097} {\emph {\bibinfo
  {title} {{XFEL}: {T}he {E}uropean {X}-{R}ay {F}ree-{E}lectron {L}aser -
  {T}echnical {D}esign {R}eport}}},\ \bibinfo {number} {DESY 2006-097}\
  (\bibinfo  {publisher} {DESY XFEL Project Group, Deutsches
  Elektronen-Synchrotron~(DESY)},\ \bibinfo {address} {Notkestra{\ss}e 85,
  22607~Hamburg, Germany},\ \bibinfo {year} {2006})\BibitemShut {NoStop}%
\bibitem [{\citenamefont {{L'Huillier}}\ \emph
  {et~al.}(1983{\natexlab{a}})\citenamefont {{L'Huillier}}, \citenamefont
  {Lompre}, \citenamefont {Mainfray},\ and\ \citenamefont
  {Manus}}]{LHuillier:MC-83}%
  \BibitemOpen
  \bibfield  {author} {\bibinfo {author} {\bibfnamefont {A.}~\bibnamefont
  {{L'Huillier}}}, \bibinfo {author} {\bibfnamefont {L.~A.}\ \bibnamefont
  {Lompre}}, \bibinfo {author} {\bibfnamefont {G.}~\bibnamefont {Mainfray}}, \
  and\ \bibinfo {author} {\bibfnamefont {C.}~\bibnamefont {Manus}},\ }\bibfield
   {title} {\enquote {\bibinfo {title} {Multiply charged ions induced by
  multiphoton absorption processes in rare-gas atoms at~$1.064 \, \mu \mathrm
  m$},}\ }\href {\doibase 10.1088/0022-3700/16/8/012} {\bibfield  {journal}
  {\bibinfo  {journal} {J. Phys. B}\ }\textbf {\bibinfo {volume} {16}},\
  \bibinfo {pages} {1363--1381} (\bibinfo {year}
  {1983}{\natexlab{a}})}\BibitemShut {NoStop}%
\bibitem [{\citenamefont {{L'Huillier}}\ \emph
  {et~al.}(1983{\natexlab{b}})\citenamefont {{L'Huillier}}, \citenamefont
  {Lompre}, \citenamefont {Mainfray},\ and\ \citenamefont
  {Manus}}]{LHuillier:RG-83}%
  \BibitemOpen
  \bibfield  {author} {\bibinfo {author} {\bibfnamefont {A.}~\bibnamefont
  {{L'Huillier}}}, \bibinfo {author} {\bibfnamefont {L.~A.}\ \bibnamefont
  {Lompre}}, \bibinfo {author} {\bibfnamefont {G.}~\bibnamefont {Mainfray}}, \
  and\ \bibinfo {author} {\bibfnamefont {C.}~\bibnamefont {Manus}},\ }\bibfield
   {title} {\enquote {\bibinfo {title} {Multiply charged ions induced by
  multiphoton absorption in rare gases at~$0.53 \, \mu \mathrm m$},}\ }\href
  {\doibase 10.1103/PhysRevA.27.2503} {\bibfield  {journal} {\bibinfo
  {journal} {Phys. Rev. A}\ }\textbf {\bibinfo {volume} {27}},\ \bibinfo
  {pages} {2503--2512} (\bibinfo {year} {1983}{\natexlab{b}})}\BibitemShut
  {NoStop}%
\bibitem [{\citenamefont {{L'Huillier}}\ \emph
  {et~al.}(1983{\natexlab{c}})\citenamefont {{L'Huillier}}, \citenamefont
  {Lompre}, \citenamefont {Mainfray},\ and\ \citenamefont
  {Manus}}]{LHuillier:LP-83}%
  \BibitemOpen
  \bibfield  {author} {\bibinfo {author} {\bibfnamefont {A.}~\bibnamefont
  {{L'Huillier}}}, \bibinfo {author} {\bibfnamefont {L.~A.}\ \bibnamefont
  {Lompre}}, \bibinfo {author} {\bibfnamefont {G.}~\bibnamefont {Mainfray}}, \
  and\ \bibinfo {author} {\bibfnamefont {C.}~\bibnamefont {Manus}},\ }\bibfield
   {title} {\enquote {\bibinfo {title} {Laser pulse duration effects in
  {Xe}$^{2+}$ ions induced by multiphoton absorption at~$0.53 \, \mu \mathrm
  m$},}\ }\href {\doibase 10.1051/jphys:0198300440110124700} {\bibfield
  {journal} {\bibinfo  {journal} {J. Phys. France}\ }\textbf {\bibinfo {volume}
  {44}},\ \bibinfo {pages} {1247--1255} (\bibinfo {year}
  {1983}{\natexlab{c}})}\BibitemShut {NoStop}%
\bibitem [{\citenamefont {Crance}\ and\ \citenamefont
  {Aymar}(1985)}]{Crance:DS-85}%
  \BibitemOpen
  \bibfield  {author} {\bibinfo {author} {\bibfnamefont {M.}~\bibnamefont
  {Crance}}\ and\ \bibinfo {author} {\bibfnamefont {M.}~\bibnamefont {Aymar}},\
  }\bibfield  {title} {\enquote {\bibinfo {title} {Competition between direct
  process and stepwise process in multiphoton stripping of atoms. {A} model
  calculation in helium},}\ }\href {\doibase 10.1051/jphys:0198500460110188700}
  {\bibfield  {journal} {\bibinfo  {journal} {J. Phys. France}\ }\textbf
  {\bibinfo {volume} {46}},\ \bibinfo {pages} {1887--1896} (\bibinfo {year}
  {1985})}\BibitemShut {NoStop}%
\bibitem [{\citenamefont {Lambropoulos}\ and\ \citenamefont
  {Tang}(1987)}]{Lambropoulos:ME-87}%
  \BibitemOpen
  \bibfield  {author} {\bibinfo {author} {\bibfnamefont {P.}~\bibnamefont
  {Lambropoulos}}\ and\ \bibinfo {author} {\bibfnamefont {X.}~\bibnamefont
  {Tang}},\ }\bibfield  {title} {\enquote {\bibinfo {title} {Multiple
  excitation and ionization of atoms by strong lasers},}\ }\href {\doibase
  10.1364/JOSAB.4.000821} {\bibfield  {journal} {\bibinfo  {journal} {J. Opt.
  Soc. Am. B}\ }\textbf {\bibinfo {volume} {4}},\ \bibinfo {pages} {821--832}
  (\bibinfo {year} {1987})}\BibitemShut {NoStop}%
\bibitem [{\citenamefont {Delone}\ and\ \citenamefont
  {Krainov}(2000)}]{Delone:MP-00}%
  \BibitemOpen
  \bibfield  {author} {\bibinfo {author} {\bibfnamefont {N.~B.}\ \bibnamefont
  {Delone}}\ and\ \bibinfo {author} {\bibfnamefont {V.~P.}\ \bibnamefont
  {Krainov}},\ }\href {\doibase 10.1007/978-3-642-57208-1} {\emph {\bibinfo
  {title} {Multiphoton Processes in Atoms}}},\ \bibinfo {edition} {2nd}\ ed.,\
  edited by\ \bibinfo {editor} {\bibfnamefont {P.}~\bibnamefont
  {Lambropoulos}},\ \bibinfo {series} {Atoms and Plasmas}, Vol.~\bibinfo
  {volume} {13}\ (\bibinfo  {publisher} {Springer},\ \bibinfo {address}
  {Berlin},\ \bibinfo {year} {2000})\BibitemShut {NoStop}%
\bibitem [{\citenamefont {Hoener}\ \emph {et~al.}(2010)\citenamefont {Hoener},
  \citenamefont {Fang}, \citenamefont {Kornilov}, \citenamefont {Gessner},
  \citenamefont {Pratt}, \citenamefont {G{\"u}hr}, \citenamefont {Kanter},
  \citenamefont {Blaga}, \citenamefont {Bostedt}, \citenamefont {Bozek},
  \citenamefont {Bucksbaum}, \citenamefont {Buth}, \citenamefont {Chen},
  \citenamefont {Coffee}, \citenamefont {Cryan}, \citenamefont {DiMauro},
  \citenamefont {Glownia}, \citenamefont {Hosler}, \citenamefont {Kukk},
  \citenamefont {Leone}, \citenamefont {McFarland}, \citenamefont
  {Messerschmidt}, \citenamefont {Murphy}, \citenamefont {Petrovic},
  \citenamefont {Rolles},\ and\ \citenamefont {Berrah}}]{Hoener:FA-10}%
  \BibitemOpen
  \bibfield  {author} {\bibinfo {author} {\bibfnamefont {M.}~\bibnamefont
  {Hoener}}, \bibinfo {author} {\bibfnamefont {L.}~\bibnamefont {Fang}},
  \bibinfo {author} {\bibfnamefont {O.}~\bibnamefont {Kornilov}}, \bibinfo
  {author} {\bibfnamefont {O.}~\bibnamefont {Gessner}}, \bibinfo {author}
  {\bibfnamefont {S.~T.}\ \bibnamefont {Pratt}}, \bibinfo {author}
  {\bibfnamefont {M.}~\bibnamefont {G{\"u}hr}}, \bibinfo {author}
  {\bibfnamefont {E.~P.}\ \bibnamefont {Kanter}}, \bibinfo {author}
  {\bibfnamefont {C.}~\bibnamefont {Blaga}}, \bibinfo {author} {\bibfnamefont
  {C.}~\bibnamefont {Bostedt}}, \bibinfo {author} {\bibfnamefont {J.~D.}\
  \bibnamefont {Bozek}}, \bibinfo {author} {\bibfnamefont {P.~H.}\ \bibnamefont
  {Bucksbaum}}, \bibinfo {author} {\bibfnamefont {C.}~\bibnamefont {Buth}},
  \bibinfo {author} {\bibfnamefont {M.}~\bibnamefont {Chen}}, \bibinfo {author}
  {\bibfnamefont {R.}~\bibnamefont {Coffee}}, \bibinfo {author} {\bibfnamefont
  {J.}~\bibnamefont {Cryan}}, \bibinfo {author} {\bibfnamefont
  {L.}~\bibnamefont {DiMauro}}, \bibinfo {author} {\bibfnamefont
  {M.}~\bibnamefont {Glownia}}, \bibinfo {author} {\bibfnamefont
  {E.}~\bibnamefont {Hosler}}, \bibinfo {author} {\bibfnamefont
  {E.}~\bibnamefont {Kukk}}, \bibinfo {author} {\bibfnamefont {S.~R.}\
  \bibnamefont {Leone}}, \bibinfo {author} {\bibfnamefont {B.}~\bibnamefont
  {McFarland}}, \bibinfo {author} {\bibfnamefont {M.}~\bibnamefont
  {Messerschmidt}}, \bibinfo {author} {\bibfnamefont {B.}~\bibnamefont
  {Murphy}}, \bibinfo {author} {\bibfnamefont {V.}~\bibnamefont {Petrovic}},
  \bibinfo {author} {\bibfnamefont {D.}~\bibnamefont {Rolles}}, \ and\ \bibinfo
  {author} {\bibfnamefont {N.}~\bibnamefont {Berrah}},\ }\bibfield  {title}
  {\enquote {\bibinfo {title} {Ultraintense x-ray induced ionization,
  dissociation, and frustrated absorption in molecular nitrogen},}\ }\href
  {\doibase 10.1103/PhysRevLett.104.253002} {\bibfield  {journal} {\bibinfo
  {journal} {Phys. Rev. Lett.}\ }\textbf {\bibinfo {volume} {104}},\ \bibinfo
  {pages} {253002} (\bibinfo {year} {2010})}\BibitemShut {NoStop}%
\bibitem [{\citenamefont {Young}\ \emph {et~al.}(2010)\citenamefont {Young},
  \citenamefont {Kanter}, \citenamefont {Kr{\"a}ssig}, \citenamefont {Li},
  \citenamefont {March}, \citenamefont {Pratt}, \citenamefont {Santra},
  \citenamefont {Southworth}, \citenamefont {Rohringer}, \citenamefont
  {DiMauro}, \citenamefont {Doumy}, \citenamefont {Roedig}, \citenamefont
  {Berrah}, \citenamefont {Fang}, \citenamefont {Hoener}, \citenamefont
  {Bucksbaum}, \citenamefont {Cryan}, \citenamefont {Ghimire}, \citenamefont
  {Glownia}, \citenamefont {Reis}, \citenamefont {Bozek}, \citenamefont
  {Bostedt},\ and\ \citenamefont {Messerschmidt}}]{Young:FE-10}%
  \BibitemOpen
  \bibfield  {author} {\bibinfo {author} {\bibfnamefont {L.}~\bibnamefont
  {Young}}, \bibinfo {author} {\bibfnamefont {E.~P.}\ \bibnamefont {Kanter}},
  \bibinfo {author} {\bibfnamefont {B.}~\bibnamefont {Kr{\"a}ssig}}, \bibinfo
  {author} {\bibfnamefont {Y.}~\bibnamefont {Li}}, \bibinfo {author}
  {\bibfnamefont {A.~M.}\ \bibnamefont {March}}, \bibinfo {author}
  {\bibfnamefont {S.~T.}\ \bibnamefont {Pratt}}, \bibinfo {author}
  {\bibfnamefont {R.}~\bibnamefont {Santra}}, \bibinfo {author} {\bibfnamefont
  {S.~H.}\ \bibnamefont {Southworth}}, \bibinfo {author} {\bibfnamefont
  {N.}~\bibnamefont {Rohringer}}, \bibinfo {author} {\bibfnamefont {L.~F.}\
  \bibnamefont {DiMauro}}, \bibinfo {author} {\bibfnamefont {G.}~\bibnamefont
  {Doumy}}, \bibinfo {author} {\bibfnamefont {C.~A.}\ \bibnamefont {Roedig}},
  \bibinfo {author} {\bibfnamefont {N.}~\bibnamefont {Berrah}}, \bibinfo
  {author} {\bibfnamefont {L.}~\bibnamefont {Fang}}, \bibinfo {author}
  {\bibfnamefont {M.}~\bibnamefont {Hoener}}, \bibinfo {author} {\bibfnamefont
  {P.~H.}\ \bibnamefont {Bucksbaum}}, \bibinfo {author} {\bibfnamefont {J.~P.}\
  \bibnamefont {Cryan}}, \bibinfo {author} {\bibfnamefont {S.}~\bibnamefont
  {Ghimire}}, \bibinfo {author} {\bibfnamefont {J.~M.}\ \bibnamefont
  {Glownia}}, \bibinfo {author} {\bibfnamefont {D.~A.}\ \bibnamefont {Reis}},
  \bibinfo {author} {\bibfnamefont {J.~D.}\ \bibnamefont {Bozek}}, \bibinfo
  {author} {\bibfnamefont {C.}~\bibnamefont {Bostedt}}, \ and\ \bibinfo
  {author} {\bibfnamefont {M.}~\bibnamefont {Messerschmidt}},\ }\bibfield
  {title} {\enquote {\bibinfo {title} {Femtosecond electronic response of atoms
  to ultra-intense x-rays},}\ }\href {\doibase 10.1038/nature09177} {\bibfield
  {journal} {\bibinfo  {journal} {Nature}\ }\textbf {\bibinfo {volume} {466}},\
  \bibinfo {pages} {56--61} (\bibinfo {year} {2010})}\BibitemShut {NoStop}%
\bibitem [{\citenamefont {Obaid}\ \emph {et~al.}(2017)\citenamefont {Obaid},
  \citenamefont {Buth}, \citenamefont {Dakovski}, \citenamefont {Beerwerth},
  \citenamefont {Holmes}, \citenamefont {Aldrich}, \citenamefont {Lin},
  \citenamefont {Minitti}, \citenamefont {Osipov}, \citenamefont {Schlotter},
  \citenamefont {Cederbaum}, \citenamefont {Fritzsche},\ and\ \citenamefont
  {Berrah}}]{Obaid:FL-17}%
  \BibitemOpen
  \bibfield  {author} {\bibinfo {author} {\bibfnamefont {R.}~\bibnamefont
  {Obaid}}, \bibinfo {author} {\bibfnamefont {C.}~\bibnamefont {Buth}},
  \bibinfo {author} {\bibfnamefont {G.}~\bibnamefont {Dakovski}}, \bibinfo
  {author} {\bibfnamefont {R.}~\bibnamefont {Beerwerth}}, \bibinfo {author}
  {\bibfnamefont {M.}~\bibnamefont {Holmes}}, \bibinfo {author} {\bibfnamefont
  {J.}~\bibnamefont {Aldrich}}, \bibinfo {author} {\bibfnamefont {M.-F.}\
  \bibnamefont {Lin}}, \bibinfo {author} {\bibfnamefont {M.}~\bibnamefont
  {Minitti}}, \bibinfo {author} {\bibfnamefont {T.}~\bibnamefont {Osipov}},
  \bibinfo {author} {\bibfnamefont {W.}~\bibnamefont {Schlotter}}, \bibinfo
  {author} {\bibfnamefont {L.~S.}\ \bibnamefont {Cederbaum}}, \bibinfo {author}
  {\bibfnamefont {S.}~\bibnamefont {Fritzsche}}, \ and\ \bibinfo {author}
  {\bibfnamefont {N.}~\bibnamefont {Berrah}},\ }\bibfield  {title} {\enquote
  {\bibinfo {title} {{LCLS} in -- photon out: fluorescence measurement of neon
  using soft x-rays},}\ }\href@noop {} {\bibfield  {journal} {\bibinfo
  {journal} {submitted}\ } (\bibinfo {year} {2017})},\ \Eprint
  {http://arxiv.org/abs/1708.01283} {arXiv:1708.01283} \BibitemShut {NoStop}%
\bibitem [{\citenamefont {Xiang}\ \emph {et~al.}(2012)\citenamefont {Xiang},
  \citenamefont {Gao}, \citenamefont {Fu}, \citenamefont {Zeng},\ and\
  \citenamefont {Yuan}}]{Xiang:RA-12}%
  \BibitemOpen
  \bibfield  {author} {\bibinfo {author} {\bibfnamefont {W.}~\bibnamefont
  {Xiang}}, \bibinfo {author} {\bibfnamefont {C.}~\bibnamefont {Gao}}, \bibinfo
  {author} {\bibfnamefont {Y.}~\bibnamefont {Fu}}, \bibinfo {author}
  {\bibfnamefont {J.}~\bibnamefont {Zeng}}, \ and\ \bibinfo {author}
  {\bibfnamefont {J.}~\bibnamefont {Yuan}},\ }\bibfield  {title} {\enquote
  {\bibinfo {title} {Inner-shell resonant absorption effects on evolution
  dynamics of the charge state distribution in a neon atom interacting with
  ultraintense x-ray pulses},}\ }\href {\doibase 10.1103/PhysRevA.86.061401}
  {\bibfield  {journal} {\bibinfo  {journal} {Phys. Rev. A}\ }\textbf {\bibinfo
  {volume} {86}},\ \bibinfo {pages} {061401(R)} (\bibinfo {year}
  {2012})}\BibitemShut {NoStop}%
\bibitem [{\citenamefont {Rohringer}\ and\ \citenamefont
  {Santra}(2008{\natexlab{a}})}]{Rohringer:RA-08}%
  \BibitemOpen
  \bibfield  {author} {\bibinfo {author} {\bibfnamefont {N.}~\bibnamefont
  {Rohringer}}\ and\ \bibinfo {author} {\bibfnamefont {R.}~\bibnamefont
  {Santra}},\ }\bibfield  {title} {\enquote {\bibinfo {title} {Resonant {Auger}
  effect at high x-ray intensity},}\ }\href {\doibase
  10.1103/PhysRevA.77.053404} {\bibfield  {journal} {\bibinfo  {journal} {Phys.
  Rev. A}\ }\textbf {\bibinfo {volume} {77}},\ \bibinfo {pages} {053404}
  (\bibinfo {year} {2008}{\natexlab{a}})}\BibitemShut {NoStop}%
\bibitem [{\citenamefont {Rohringer}\ and\ \citenamefont
  {Santra}(2008{\natexlab{b}})}]{Rohringer:PN-08}%
  \BibitemOpen
  \bibfield  {author} {\bibinfo {author} {\bibfnamefont {N.}~\bibnamefont
  {Rohringer}}\ and\ \bibinfo {author} {\bibfnamefont {R.}~\bibnamefont
  {Santra}},\ }\bibfield  {title} {\enquote {\bibinfo {title} {Publisher's
  {Note: Resonant Auger} effect at high x-ray intensity [{Phys. Rev. A}
  \textbf{77}, 053404 (2008)]},}\ }\href {\doibase 10.1103/PhysRevA.77.059903}
  {\bibfield  {journal} {\bibinfo  {journal} {Phys. Rev. A}\ }\textbf {\bibinfo
  {volume} {77}},\ \bibinfo {pages} {059903(E)} (\bibinfo {year}
  {2008}{\natexlab{b}})}\BibitemShut {NoStop}%
\bibitem [{\citenamefont {Kanter}\ \emph {et~al.}(2011)\citenamefont {Kanter},
  \citenamefont {Kr{\"a}ssig}, \citenamefont {Li}, \citenamefont {March},
  \citenamefont {Ho}, \citenamefont {Rohringer}, \citenamefont {Santra},
  \citenamefont {Southworth}, \citenamefont {DiMauro}, \citenamefont {Doumy},
  \citenamefont {Roedig}, \citenamefont {Berrah}, \citenamefont {Fang},
  \citenamefont {Hoener}, \citenamefont {Bucksbaum}, \citenamefont {Ghimire},
  \citenamefont {Reis}, \citenamefont {Bozek}, \citenamefont {Bostedt},
  \citenamefont {Messerschmidt},\ and\ \citenamefont {Young}}]{Kanter:MA-11}%
  \BibitemOpen
  \bibfield  {author} {\bibinfo {author} {\bibfnamefont {E.~P.}\ \bibnamefont
  {Kanter}}, \bibinfo {author} {\bibfnamefont {B.}~\bibnamefont {Kr{\"a}ssig}},
  \bibinfo {author} {\bibfnamefont {Y.}~\bibnamefont {Li}}, \bibinfo {author}
  {\bibfnamefont {A.~M.}\ \bibnamefont {March}}, \bibinfo {author}
  {\bibfnamefont {P.}~\bibnamefont {Ho}}, \bibinfo {author} {\bibfnamefont
  {N.}~\bibnamefont {Rohringer}}, \bibinfo {author} {\bibfnamefont
  {R.}~\bibnamefont {Santra}}, \bibinfo {author} {\bibfnamefont {S.~H.}\
  \bibnamefont {Southworth}}, \bibinfo {author} {\bibfnamefont {L.~F.}\
  \bibnamefont {DiMauro}}, \bibinfo {author} {\bibfnamefont {G.}~\bibnamefont
  {Doumy}}, \bibinfo {author} {\bibfnamefont {C.~A.}\ \bibnamefont {Roedig}},
  \bibinfo {author} {\bibfnamefont {N.}~\bibnamefont {Berrah}}, \bibinfo
  {author} {\bibfnamefont {L.}~\bibnamefont {Fang}}, \bibinfo {author}
  {\bibfnamefont {M.}~\bibnamefont {Hoener}}, \bibinfo {author} {\bibfnamefont
  {P.~H.}\ \bibnamefont {Bucksbaum}}, \bibinfo {author} {\bibfnamefont
  {S.}~\bibnamefont {Ghimire}}, \bibinfo {author} {\bibfnamefont {D.~A.}\
  \bibnamefont {Reis}}, \bibinfo {author} {\bibfnamefont {J.~D.}\ \bibnamefont
  {Bozek}}, \bibinfo {author} {\bibfnamefont {C.}~\bibnamefont {Bostedt}},
  \bibinfo {author} {\bibfnamefont {M.}~\bibnamefont {Messerschmidt}}, \ and\
  \bibinfo {author} {\bibfnamefont {L.}~\bibnamefont {Young}},\ }\bibfield
  {title} {\enquote {\bibinfo {title} {Unveiling and driving hidden resonances
  with high-fluence, high-intensity x-ray pulses},}\ }\href {\doibase
  10.1103/PhysRevLett.107.233001} {\bibfield  {journal} {\bibinfo  {journal}
  {Phys. Rev. Lett.}\ }\textbf {\bibinfo {volume} {107}},\ \bibinfo {pages}
  {233001} (\bibinfo {year} {2011})}\BibitemShut {NoStop}%
\bibitem [{\citenamefont {Cavaletto}\ \emph {et~al.}(2012)\citenamefont
  {Cavaletto}, \citenamefont {Buth}, \citenamefont {Harman}, \citenamefont
  {Kanter}, \citenamefont {Southworth}, \citenamefont {Young},\ and\
  \citenamefont {Keitel}}]{Cavaletto:RF-12}%
  \BibitemOpen
  \bibfield  {author} {\bibinfo {author} {\bibfnamefont {S.~M.}\ \bibnamefont
  {Cavaletto}}, \bibinfo {author} {\bibfnamefont {C.}~\bibnamefont {Buth}},
  \bibinfo {author} {\bibfnamefont {Z.}~\bibnamefont {Harman}}, \bibinfo
  {author} {\bibfnamefont {E.~P.}\ \bibnamefont {Kanter}}, \bibinfo {author}
  {\bibfnamefont {S.~H.}\ \bibnamefont {Southworth}}, \bibinfo {author}
  {\bibfnamefont {L.}~\bibnamefont {Young}}, \ and\ \bibinfo {author}
  {\bibfnamefont {C.~H.}\ \bibnamefont {Keitel}},\ }\bibfield  {title}
  {\enquote {\bibinfo {title} {Resonance fluorescence in ultrafast and intense
  x-ray free electron laser pulses},}\ }\href {\doibase
  10.1103/PhysRevA.86.033402} {\bibfield  {journal} {\bibinfo  {journal} {Phys.
  Rev. A}\ }\textbf {\bibinfo {volume} {86}},\ \bibinfo {pages} {033402}
  (\bibinfo {year} {2012})},\ \Eprint {http://arxiv.org/abs/1205.4918}
  {arXiv:1205.4918} \BibitemShut {NoStop}%
\bibitem [{\citenamefont {Cavaletto}\ \emph {et~al.}(2013)\citenamefont
  {Cavaletto}, \citenamefont {Harman}, \citenamefont {Buth},\ and\
  \citenamefont {Keitel}}]{Cavaletto:FC-13}%
  \BibitemOpen
  \bibfield  {author} {\bibinfo {author} {\bibfnamefont {S.~M.}\ \bibnamefont
  {Cavaletto}}, \bibinfo {author} {\bibfnamefont {Z.}~\bibnamefont {Harman}},
  \bibinfo {author} {\bibfnamefont {C.}~\bibnamefont {Buth}}, \ and\ \bibinfo
  {author} {\bibfnamefont {C.~H.}\ \bibnamefont {Keitel}},\ }\bibfield  {title}
  {\enquote {\bibinfo {title} {X-ray frequency combs from optically controlled
  resonance fluorescence},}\ }\href {\doibase 10.1103/PhysRevA.88.063402}
  {\bibfield  {journal} {\bibinfo  {journal} {Phys. Rev. A}\ }\textbf {\bibinfo
  {volume} {88}},\ \bibinfo {pages} {063402} (\bibinfo {year} {2013})},\
  \Eprint {http://arxiv.org/abs/1302.3141} {arXiv:1302.3141} \BibitemShut
  {NoStop}%
\bibitem [{\citenamefont {Adams}\ \emph {et~al.}(2013)\citenamefont {Adams},
  \citenamefont {Buth}, \citenamefont {Cavaletto}, \citenamefont {Evers},
  \citenamefont {Harman}, \citenamefont {Keitel}, \citenamefont {P{\'a}lffy},
  \citenamefont {Pic{\'o}n}, \citenamefont {R{\"o}hlsberger}, \citenamefont
  {Rostovtsev},\ and\ \citenamefont {Tamasaku}}]{Adams:QO-13}%
  \BibitemOpen
  \bibfield  {author} {\bibinfo {author} {\bibfnamefont {B.~W.}\ \bibnamefont
  {Adams}}, \bibinfo {author} {\bibfnamefont {C.}~\bibnamefont {Buth}},
  \bibinfo {author} {\bibfnamefont {S.~M.}\ \bibnamefont {Cavaletto}}, \bibinfo
  {author} {\bibfnamefont {J.}~\bibnamefont {Evers}}, \bibinfo {author}
  {\bibfnamefont {Z.}~\bibnamefont {Harman}}, \bibinfo {author} {\bibfnamefont
  {C.~H.}\ \bibnamefont {Keitel}}, \bibinfo {author} {\bibfnamefont
  {A.}~\bibnamefont {P{\'a}lffy}}, \bibinfo {author} {\bibfnamefont
  {A.}~\bibnamefont {Pic{\'o}n}}, \bibinfo {author} {\bibfnamefont
  {R.}~\bibnamefont {R{\"o}hlsberger}}, \bibinfo {author} {\bibfnamefont
  {Y.}~\bibnamefont {Rostovtsev}}, \ and\ \bibinfo {author} {\bibfnamefont
  {K.}~\bibnamefont {Tamasaku}},\ }\bibfield  {title} {\enquote {\bibinfo
  {title} {X-ray quantum optics},}\ }\href {\doibase
  10.1080/09500340.2012.752113} {\bibfield  {journal} {\bibinfo  {journal} {J.
  Mod. Opt.}\ }\textbf {\bibinfo {volume} {60}},\ \bibinfo {pages} {2--21}
  (\bibinfo {year} {2013})}\BibitemShut {NoStop}%
\bibitem [{\citenamefont {Cavaletto}\ \emph {et~al.}(2014)\citenamefont
  {Cavaletto}, \citenamefont {Harman}, \citenamefont {Ott}, \citenamefont
  {Buth}, \citenamefont {Pfeifer},\ and\ \citenamefont
  {Keitel}}]{Cavaletto:HF-14}%
  \BibitemOpen
  \bibfield  {author} {\bibinfo {author} {\bibfnamefont {S.~M.}\ \bibnamefont
  {Cavaletto}}, \bibinfo {author} {\bibfnamefont {Z.}~\bibnamefont {Harman}},
  \bibinfo {author} {\bibfnamefont {C.}~\bibnamefont {Ott}}, \bibinfo {author}
  {\bibfnamefont {C.}~\bibnamefont {Buth}}, \bibinfo {author} {\bibfnamefont
  {T.}~\bibnamefont {Pfeifer}}, \ and\ \bibinfo {author} {\bibfnamefont
  {C.~H.}\ \bibnamefont {Keitel}},\ }\bibfield  {title} {\enquote {\bibinfo
  {title} {Broadband high-resolution x-ray frequency combs},}\ }\href {\doibase
  10.1038/nphoton.2014.113} {\bibfield  {journal} {\bibinfo  {journal} {Nature
  Photon.}\ }\textbf {\bibinfo {volume} {8}},\ \bibinfo {pages} {520--523}
  (\bibinfo {year} {2014})},\ \Eprint {http://arxiv.org/abs/1402.6652}
  {arXiv:1402.6652} \BibitemShut {NoStop}%
\bibitem [{\citenamefont {Li}\ \emph {et~al.}(2016)\citenamefont {Li},
  \citenamefont {Gao}, \citenamefont {Dong}, \citenamefont {Zeng},
  \citenamefont {Zhao},\ and\ \citenamefont {Yuan}}]{Li:CR-16}%
  \BibitemOpen
  \bibfield  {author} {\bibinfo {author} {\bibfnamefont {Y.}~\bibnamefont
  {Li}}, \bibinfo {author} {\bibfnamefont {C.}~\bibnamefont {Gao}}, \bibinfo
  {author} {\bibfnamefont {W.}~\bibnamefont {Dong}}, \bibinfo {author}
  {\bibfnamefont {J.}~\bibnamefont {Zeng}}, \bibinfo {author} {\bibfnamefont
  {Z.}~\bibnamefont {Zhao}}, \ and\ \bibinfo {author} {\bibfnamefont
  {J.}~\bibnamefont {Yuan}},\ }\bibfield  {title} {\enquote {\bibinfo {title}
  {Coherence and resonance effects in the ultra-intense laser-induced ultrafast
  response of complex atoms},}\ }\href {\doibase 10.1038/srep18529} {\bibfield
  {journal} {\bibinfo  {journal} {Sci. Rep.}\ }\textbf {\bibinfo {volume}
  {6}},\ \bibinfo {pages} {18529} (\bibinfo {year} {2016})}\BibitemShut
  {NoStop}%
\bibitem [{\citenamefont {Als-Nielsen}\ and\ \citenamefont
  {McMorrow}(2001)}]{Als-Nielsen:EM-01}%
  \BibitemOpen
  \bibfield  {author} {\bibinfo {author} {\bibfnamefont {J.}~\bibnamefont
  {Als-Nielsen}}\ and\ \bibinfo {author} {\bibfnamefont {D.}~\bibnamefont
  {McMorrow}},\ }\href@noop {} {\emph {\bibinfo {title} {Elements of Modern
  X-Ray Physics}}}\ (\bibinfo  {publisher} {John Wiley {\&} Sons},\ \bibinfo
  {address} {New York},\ \bibinfo {year} {2001})\BibitemShut {NoStop}%
\bibitem [{\citenamefont {Makris}, \citenamefont {Lambropoulos},\ and\
  \citenamefont {Miheli{\v{c}}}(2009)}]{Makris:MM-09}%
  \BibitemOpen
  \bibfield  {author} {\bibinfo {author} {\bibfnamefont {M.~G.}\ \bibnamefont
  {Makris}}, \bibinfo {author} {\bibfnamefont {P.}~\bibnamefont
  {Lambropoulos}}, \ and\ \bibinfo {author} {\bibfnamefont {A.}~\bibnamefont
  {Miheli{\v{c}}}},\ }\bibfield  {title} {\enquote {\bibinfo {title} {Theory of
  multiphoton multielectron ionization of xenon under strong {93-eV}
  radiation},}\ }\href {\doibase 10.1103/PhysRevLett.102.033002} {\bibfield
  {journal} {\bibinfo  {journal} {Phys. Rev. Lett.}\ }\textbf {\bibinfo
  {volume} {102}},\ \bibinfo {pages} {033002} (\bibinfo {year}
  {2009})}\BibitemShut {NoStop}%
\bibitem [{\citenamefont {Doumy}\ \emph {et~al.}(2011)\citenamefont {Doumy},
  \citenamefont {Roedig}, \citenamefont {Son}, \citenamefont {Blaga},
  \citenamefont {DiChiara}, \citenamefont {Santra}, \citenamefont {Berrah},
  \citenamefont {Bostedt}, \citenamefont {Bozek}, \citenamefont {Bucksbaum},
  \citenamefont {Cryan}, \citenamefont {Fang}, \citenamefont {Ghimire},
  \citenamefont {Glownia}, \citenamefont {Hoener}, \citenamefont {Kanter},
  \citenamefont {Kr{\"a}ssig}, \citenamefont {Kuebel}, \citenamefont
  {Messerschmidt}, \citenamefont {Paulus}, \citenamefont {Reis}, \citenamefont
  {Rohringer}, \citenamefont {Young}, \citenamefont {Agostini},\ and\
  \citenamefont {DiMauro}}]{Doumy:NA-11}%
  \BibitemOpen
  \bibfield  {author} {\bibinfo {author} {\bibfnamefont {G.}~\bibnamefont
  {Doumy}}, \bibinfo {author} {\bibfnamefont {C.}~\bibnamefont {Roedig}},
  \bibinfo {author} {\bibfnamefont {S.-K.}\ \bibnamefont {Son}}, \bibinfo
  {author} {\bibfnamefont {C.~I.}\ \bibnamefont {Blaga}}, \bibinfo {author}
  {\bibfnamefont {A.~D.}\ \bibnamefont {DiChiara}}, \bibinfo {author}
  {\bibfnamefont {R.}~\bibnamefont {Santra}}, \bibinfo {author} {\bibfnamefont
  {N.}~\bibnamefont {Berrah}}, \bibinfo {author} {\bibfnamefont
  {C.}~\bibnamefont {Bostedt}}, \bibinfo {author} {\bibfnamefont {J.~D.}\
  \bibnamefont {Bozek}}, \bibinfo {author} {\bibfnamefont {P.~H.}\ \bibnamefont
  {Bucksbaum}}, \bibinfo {author} {\bibfnamefont {J.~P.}\ \bibnamefont
  {Cryan}}, \bibinfo {author} {\bibfnamefont {L.}~\bibnamefont {Fang}},
  \bibinfo {author} {\bibfnamefont {S.}~\bibnamefont {Ghimire}}, \bibinfo
  {author} {\bibfnamefont {J.~M.}\ \bibnamefont {Glownia}}, \bibinfo {author}
  {\bibfnamefont {M.}~\bibnamefont {Hoener}}, \bibinfo {author} {\bibfnamefont
  {E.~P.}\ \bibnamefont {Kanter}}, \bibinfo {author} {\bibfnamefont
  {B.}~\bibnamefont {Kr{\"a}ssig}}, \bibinfo {author} {\bibfnamefont
  {M.}~\bibnamefont {Kuebel}}, \bibinfo {author} {\bibfnamefont
  {M.}~\bibnamefont {Messerschmidt}}, \bibinfo {author} {\bibfnamefont {G.~G.}\
  \bibnamefont {Paulus}}, \bibinfo {author} {\bibfnamefont {D.~A.}\
  \bibnamefont {Reis}}, \bibinfo {author} {\bibfnamefont {N.}~\bibnamefont
  {Rohringer}}, \bibinfo {author} {\bibfnamefont {L.}~\bibnamefont {Young}},
  \bibinfo {author} {\bibfnamefont {P.}~\bibnamefont {Agostini}}, \ and\
  \bibinfo {author} {\bibfnamefont {L.~F.}\ \bibnamefont {DiMauro}},\
  }\bibfield  {title} {\enquote {\bibinfo {title} {Nonlinear atomic response to
  intense ultrashort x~rays},}\ }\href {\doibase
  10.1103/PhysRevLett.106.083002} {\bibfield  {journal} {\bibinfo  {journal}
  {Phys. Rev. Lett.}\ }\textbf {\bibinfo {volume} {106}},\ \bibinfo {pages}
  {083002} (\bibinfo {year} {2011})}\BibitemShut {NoStop}%
\bibitem [{\citenamefont {Tamasaku}\ \emph {et~al.}(2014)\citenamefont
  {Tamasaku}, \citenamefont {Shigemasa}, \citenamefont {Inubushi},
  \citenamefont {Katayama}, \citenamefont {Sawada}, \citenamefont {Yumoto},
  \citenamefont {Ohashi}, \citenamefont {Mimura}, \citenamefont {Yabashi},
  \citenamefont {Yamauchi},\ and\ \citenamefont {Ishikawa}}]{Tamasaku:TP-14}%
  \BibitemOpen
  \bibfield  {author} {\bibinfo {author} {\bibfnamefont {K.}~\bibnamefont
  {Tamasaku}}, \bibinfo {author} {\bibfnamefont {E.}~\bibnamefont {Shigemasa}},
  \bibinfo {author} {\bibfnamefont {Y.}~\bibnamefont {Inubushi}}, \bibinfo
  {author} {\bibfnamefont {T.}~\bibnamefont {Katayama}}, \bibinfo {author}
  {\bibfnamefont {K.}~\bibnamefont {Sawada}}, \bibinfo {author} {\bibfnamefont
  {H.}~\bibnamefont {Yumoto}}, \bibinfo {author} {\bibfnamefont
  {H.}~\bibnamefont {Ohashi}}, \bibinfo {author} {\bibfnamefont
  {H.}~\bibnamefont {Mimura}}, \bibinfo {author} {\bibfnamefont
  {M.}~\bibnamefont {Yabashi}}, \bibinfo {author} {\bibfnamefont
  {K.}~\bibnamefont {Yamauchi}}, \ and\ \bibinfo {author} {\bibfnamefont
  {T.}~\bibnamefont {Ishikawa}},\ }\bibfield  {title} {\enquote {\bibinfo
  {title} {X-ray two-photon absorption competing against single and sequential
  multiphoton processes},}\ }\href {\doibase 10.1038/nphoton.2014.10}
  {\bibfield  {journal} {\bibinfo  {journal} {Nat. Photon.}\ }\textbf {\bibinfo
  {volume} {8}},\ \bibinfo {pages} {313--316} (\bibinfo {year}
  {2014})}\BibitemShut {NoStop}%
\bibitem [{\citenamefont {Neutze}\ \emph {et~al.}(2000)\citenamefont {Neutze},
  \citenamefont {Wouts}, \citenamefont {van~der Spoel}, \citenamefont
  {Weckert},\ and\ \citenamefont {Hajdu}}]{Neutze:BI-00}%
  \BibitemOpen
  \bibfield  {author} {\bibinfo {author} {\bibfnamefont {R.}~\bibnamefont
  {Neutze}}, \bibinfo {author} {\bibfnamefont {R.}~\bibnamefont {Wouts}},
  \bibinfo {author} {\bibfnamefont {D.}~\bibnamefont {van~der Spoel}}, \bibinfo
  {author} {\bibfnamefont {E.}~\bibnamefont {Weckert}}, \ and\ \bibinfo
  {author} {\bibfnamefont {J.}~\bibnamefont {Hajdu}},\ }\bibfield  {title}
  {\enquote {\bibinfo {title} {Potential for biomolecular imaging with
  femtosecond x-ray pulses},}\ }\href@noop {} {\bibfield  {journal} {\bibinfo
  {journal} {Nature}\ }\textbf {\bibinfo {volume} {406}},\ \bibinfo {pages}
  {752--757} (\bibinfo {year} {2000})}\BibitemShut {NoStop}%
\bibitem [{\citenamefont {Cryan}\ \emph {et~al.}(2010)\citenamefont {Cryan},
  \citenamefont {Glownia}, \citenamefont {Andreasson}, \citenamefont
  {Belkacem}, \citenamefont {Berrah}, \citenamefont {Blaga}, \citenamefont
  {Bostedt}, \citenamefont {Bozek}, \citenamefont {Buth}, \citenamefont
  {DiMauro}, \citenamefont {Fang}, \citenamefont {Gessner}, \citenamefont
  {Guehr}, \citenamefont {Hajdu}, \citenamefont {Hertlein}, \citenamefont
  {Hoener}, \citenamefont {Kornilov}, \citenamefont {Marangos}, \citenamefont
  {March}, \citenamefont {McFarland}, \citenamefont {Merdji}, \citenamefont
  {Petrovi{\ifmmode \acute{c}\else \'{c}\fi{}}}, \citenamefont {Raman},
  \citenamefont {Ray}, \citenamefont {Reis}, \citenamefont {Tarantelli},
  \citenamefont {Trigo}, \citenamefont {White}, \citenamefont {White},
  \citenamefont {Young}, \citenamefont {Bucksbaum},\ and\ \citenamefont
  {Coffee}}]{Cryan:AE-10}%
  \BibitemOpen
  \bibfield  {author} {\bibinfo {author} {\bibfnamefont {J.~P.}\ \bibnamefont
  {Cryan}}, \bibinfo {author} {\bibfnamefont {J.~M.}\ \bibnamefont {Glownia}},
  \bibinfo {author} {\bibfnamefont {J.}~\bibnamefont {Andreasson}}, \bibinfo
  {author} {\bibfnamefont {A.}~\bibnamefont {Belkacem}}, \bibinfo {author}
  {\bibfnamefont {N.}~\bibnamefont {Berrah}}, \bibinfo {author} {\bibfnamefont
  {C.~I.}\ \bibnamefont {Blaga}}, \bibinfo {author} {\bibfnamefont
  {C.}~\bibnamefont {Bostedt}}, \bibinfo {author} {\bibfnamefont
  {J.}~\bibnamefont {Bozek}}, \bibinfo {author} {\bibfnamefont
  {C.}~\bibnamefont {Buth}}, \bibinfo {author} {\bibfnamefont {L.~F.}\
  \bibnamefont {DiMauro}}, \bibinfo {author} {\bibfnamefont {L.}~\bibnamefont
  {Fang}}, \bibinfo {author} {\bibfnamefont {O.}~\bibnamefont {Gessner}},
  \bibinfo {author} {\bibfnamefont {M.}~\bibnamefont {Guehr}}, \bibinfo
  {author} {\bibfnamefont {J.}~\bibnamefont {Hajdu}}, \bibinfo {author}
  {\bibfnamefont {M.~P.}\ \bibnamefont {Hertlein}}, \bibinfo {author}
  {\bibfnamefont {M.}~\bibnamefont {Hoener}}, \bibinfo {author} {\bibfnamefont
  {O.}~\bibnamefont {Kornilov}}, \bibinfo {author} {\bibfnamefont {J.~P.}\
  \bibnamefont {Marangos}}, \bibinfo {author} {\bibfnamefont {A.~M.}\
  \bibnamefont {March}}, \bibinfo {author} {\bibfnamefont {B.~K.}\ \bibnamefont
  {McFarland}}, \bibinfo {author} {\bibfnamefont {H.}~\bibnamefont {Merdji}},
  \bibinfo {author} {\bibfnamefont {V.~S.}\ \bibnamefont {Petrovi{\ifmmode
  \acute{c}\else \'{c}\fi{}}}}, \bibinfo {author} {\bibfnamefont
  {C.}~\bibnamefont {Raman}}, \bibinfo {author} {\bibfnamefont
  {D.}~\bibnamefont {Ray}}, \bibinfo {author} {\bibfnamefont {D.}~\bibnamefont
  {Reis}}, \bibinfo {author} {\bibfnamefont {F.}~\bibnamefont {Tarantelli}},
  \bibinfo {author} {\bibfnamefont {M.}~\bibnamefont {Trigo}}, \bibinfo
  {author} {\bibfnamefont {J.~L.}\ \bibnamefont {White}}, \bibinfo {author}
  {\bibfnamefont {W.}~\bibnamefont {White}}, \bibinfo {author} {\bibfnamefont
  {L.}~\bibnamefont {Young}}, \bibinfo {author} {\bibfnamefont {P.~H.}\
  \bibnamefont {Bucksbaum}}, \ and\ \bibinfo {author} {\bibfnamefont {R.~N.}\
  \bibnamefont {Coffee}},\ }\bibfield  {title} {\enquote {\bibinfo {title}
  {{Auger} electron angular distribution of double core hole states in the
  molecular reference frame},}\ }\href {\doibase
  10.1103/PhysRevLett.105.083004} {\bibfield  {journal} {\bibinfo  {journal}
  {Phys. Rev. Lett.}\ }\textbf {\bibinfo {volume} {105}},\ \bibinfo {pages}
  {083004} (\bibinfo {year} {2010})}\BibitemShut {NoStop}%
\bibitem [{\citenamefont {Fang}\ \emph {et~al.}(2010)\citenamefont {Fang},
  \citenamefont {Hoener}, \citenamefont {Gessner}, \citenamefont {Tarantelli},
  \citenamefont {Pratt}, \citenamefont {Kornilov}, \citenamefont {Buth},
  \citenamefont {G{\"u}hr}, \citenamefont {Kanter}, \citenamefont {Bostedt},
  \citenamefont {Bozek}, \citenamefont {Bucksbaum}, \citenamefont {Chen},
  \citenamefont {Coffee}, \citenamefont {Cryan}, \citenamefont {Glownia},
  \citenamefont {Kukk}, \citenamefont {Leone},\ and\ \citenamefont
  {Berrah}}]{Fang:DC-10}%
  \BibitemOpen
  \bibfield  {author} {\bibinfo {author} {\bibfnamefont {L.}~\bibnamefont
  {Fang}}, \bibinfo {author} {\bibfnamefont {M.}~\bibnamefont {Hoener}},
  \bibinfo {author} {\bibfnamefont {O.}~\bibnamefont {Gessner}}, \bibinfo
  {author} {\bibfnamefont {F.}~\bibnamefont {Tarantelli}}, \bibinfo {author}
  {\bibfnamefont {S.~T.}\ \bibnamefont {Pratt}}, \bibinfo {author}
  {\bibfnamefont {O.}~\bibnamefont {Kornilov}}, \bibinfo {author}
  {\bibfnamefont {C.}~\bibnamefont {Buth}}, \bibinfo {author} {\bibfnamefont
  {M.}~\bibnamefont {G{\"u}hr}}, \bibinfo {author} {\bibfnamefont {E.~P.}\
  \bibnamefont {Kanter}}, \bibinfo {author} {\bibfnamefont {C.}~\bibnamefont
  {Bostedt}}, \bibinfo {author} {\bibfnamefont {J.~D.}\ \bibnamefont {Bozek}},
  \bibinfo {author} {\bibfnamefont {P.~H.}\ \bibnamefont {Bucksbaum}}, \bibinfo
  {author} {\bibfnamefont {M.}~\bibnamefont {Chen}}, \bibinfo {author}
  {\bibfnamefont {R.}~\bibnamefont {Coffee}}, \bibinfo {author} {\bibfnamefont
  {J.}~\bibnamefont {Cryan}}, \bibinfo {author} {\bibfnamefont
  {M.}~\bibnamefont {Glownia}}, \bibinfo {author} {\bibfnamefont
  {E.}~\bibnamefont {Kukk}}, \bibinfo {author} {\bibfnamefont {S.~R.}\
  \bibnamefont {Leone}}, \ and\ \bibinfo {author} {\bibfnamefont
  {N.}~\bibnamefont {Berrah}},\ }\bibfield  {title} {\enquote {\bibinfo {title}
  {Double core hole production in~{N}$_2$: {Beating} the {Auger} clock},}\
  }\href {\doibase 10.1103/PhysRevLett.105.083005} {\bibfield  {journal}
  {\bibinfo  {journal} {Phys. Rev. Lett.}\ }\textbf {\bibinfo {volume} {105}},\
  \bibinfo {pages} {083005} (\bibinfo {year} {2010})},\ \Eprint
  {http://arxiv.org/abs/1303.1429} {arXiv:1303.1429} \BibitemShut {NoStop}%
\bibitem [{\citenamefont {Fang}\ \emph {et~al.}(2012)\citenamefont {Fang},
  \citenamefont {Osipov}, \citenamefont {Murphy}, \citenamefont {Tarantelli},
  \citenamefont {Kukk}, \citenamefont {Cryan}, \citenamefont {Bucksbaum},
  \citenamefont {Coffee}, \citenamefont {Chen}, \citenamefont {Buth},\ and\
  \citenamefont {Berrah}}]{Fang:MI-12}%
  \BibitemOpen
  \bibfield  {author} {\bibinfo {author} {\bibfnamefont {L.}~\bibnamefont
  {Fang}}, \bibinfo {author} {\bibfnamefont {T.}~\bibnamefont {Osipov}},
  \bibinfo {author} {\bibfnamefont {B.}~\bibnamefont {Murphy}}, \bibinfo
  {author} {\bibfnamefont {F.}~\bibnamefont {Tarantelli}}, \bibinfo {author}
  {\bibfnamefont {E.}~\bibnamefont {Kukk}}, \bibinfo {author} {\bibfnamefont
  {J.}~\bibnamefont {Cryan}}, \bibinfo {author} {\bibfnamefont {P.~H.}\
  \bibnamefont {Bucksbaum}}, \bibinfo {author} {\bibfnamefont {R.~N.}\
  \bibnamefont {Coffee}}, \bibinfo {author} {\bibfnamefont {M.}~\bibnamefont
  {Chen}}, \bibinfo {author} {\bibfnamefont {C.}~\bibnamefont {Buth}}, \ and\
  \bibinfo {author} {\bibfnamefont {N.}~\bibnamefont {Berrah}},\ }\bibfield
  {title} {\enquote {\bibinfo {title} {Multiphoton ionization as a clock to
  reveal molecular dynamics with intense short x-ray free electron laser
  pulses},}\ }\href {\doibase 10.1103/PhysRevLett.109.263001} {\bibfield
  {journal} {\bibinfo  {journal} {Phys. Rev. Lett.}\ }\textbf {\bibinfo
  {volume} {109}},\ \bibinfo {pages} {263001} (\bibinfo {year} {2012})},\
  \Eprint {http://arxiv.org/abs/1301.6459} {arXiv:1301.6459} \BibitemShut
  {NoStop}%
\bibitem [{Note1()}]{Note1}%
  \BibitemOpen
  \bibinfo {note} {The essence of the linearity theorem of multiphoton
  absorption has already been stated in the third paragraph in the left column
  on page~2 of Ref.~\protect \rev@citealpnum {Buth:UA-12}.}\BibitemShut {Stop}%
\bibitem [{\citenamefont {G{\"o}ppert-Mayer}(1931)}]{Goppert:EZ-31}%
  \BibitemOpen
  \bibfield  {author} {\bibinfo {author} {\bibfnamefont {M.}~\bibnamefont
  {G{\"o}ppert-Mayer}},\ }\bibfield  {title} {\enquote {\bibinfo {title}
  {{{\"U}ber Elementarakte mit zwei Quantenspr{\"u}ngen}},}\ }\href {\doibase
  10.1002/andp.19314010303} {\bibfield  {journal} {\bibinfo  {journal} {Ann.
  Phys. (Leipzig)}\ }\textbf {\bibinfo {volume} {17}},\ \bibinfo {pages}
  {273--294} (\bibinfo {year} {1931})}\BibitemShut {NoStop}%
\bibitem [{\citenamefont {Hartree}(1928)}]{Hartree:WM-28}%
  \BibitemOpen
  \bibfield  {author} {\bibinfo {author} {\bibfnamefont {D.~R.}\ \bibnamefont
  {Hartree}},\ }\bibfield  {title} {\enquote {\bibinfo {title} {The wave
  mechanics of an atom with a non-coulomb central field. {Part}~{I}. {Theory}
  and methods},}\ }\href {\doibase 10.1017/S0305004100011919} {\bibfield
  {journal} {\bibinfo  {journal} {Proc. Camb. Phil. Soc.}\ }\textbf {\bibinfo
  {volume} {24}},\ \bibinfo {pages} {89--110} (\bibinfo {year}
  {1928})}\BibitemShut {NoStop}%
\bibitem [{\citenamefont {Szabo}\ and\ \citenamefont
  {Ostlund}(1989)}]{Szabo:MQC-89}%
  \BibitemOpen
  \bibfield  {author} {\bibinfo {author} {\bibfnamefont {A.}~\bibnamefont
  {Szabo}}\ and\ \bibinfo {author} {\bibfnamefont {N.~S.}\ \bibnamefont
  {Ostlund}},\ }\href@noop {} {\emph {\bibinfo {title} {Modern Quantum
  Chemistry: Introduction to Advanced Electronic Structure Theory}}},\ \bibinfo
  {edition} {{1st, revised}}\ ed.\ (\bibinfo  {publisher} {McGraw-Hill},\
  \bibinfo {address} {New York},\ \bibinfo {year} {1989})\BibitemShut {NoStop}%
\bibitem [{Sup()}]{SuppData}%
  \BibitemOpen
  \href@noop {} {}\bibinfo {note} {See the Supplementary Data for a
  \textit{Mathematica}~\cite{Mathematica:pgm-V11} file with all computations
  presented in this work.}\BibitemShut {Stop}%
\bibitem [{Mat(2016)}]{Mathematica:pgm-V11}%
  \BibitemOpen
  \href {http://www.wolfram.com} {\emph {\bibinfo {title}
  {\textit{Mathematica}~11}}},\ \bibinfo {organization} {Wolfram Research,
  Inc.},\ \bibinfo {address} {100~Trade Center Drive, Champaign,
  Illinois~61820-7237, USA} (\bibinfo {year} {2016})\BibitemShut {NoStop}%
\bibitem [{Note2()}]{Note2}%
  \BibitemOpen
  \bibinfo {note} {For conceptual ease, I restrict myself here in considering
  the final state to be electron-bare which, however, is not necessarily
  required. Instead the single final state may also be a configuration which
  still has electrons in deep lying shells which is energetically to high in
  energy such that it cannot be ionized by the photons with the given energy
  and is not amenable to REXMI.~\cite {Rudek:UE–12,Rudek:RE-13}}\BibitemShut
  {NoStop}%
\bibitem [{Note3()}]{Note3}%
  \BibitemOpen
  \bibinfo {note} {A fixed photon energy is assumed which needs to be high
  enough such that all electrons can be ionized; the dependence of the
  quantities on the photon energy is omitted throughout. If the x\hbox
  {-}{}ray~photon energy is lower than the energy that is necessary to strip
  the system of all electrons, then the total absorption cross section of this
  last configuration---I assume that there is a single last configuration---is
  zero, although there are still electrons in lower lying shells. Increasing
  the x\hbox {-}{}ray~photon energy renders the total absorption cross section
  nonzero. Please see figure~2 in~\cite {Buth:UA-12} and its explanation. Only
  cationic configurations are included without regarding shake-off or shake-up
  processes and other multielectron effects such as double Auger
  decay.}\BibitemShut {Stop}%
\bibitem [{\citenamefont {Pauli}(1928)}]{Pauli:FS-28}%
  \BibitemOpen
  \bibfield  {author} {\bibinfo {author} {\bibfnamefont {W.}~\bibnamefont
  {Pauli}},\ }\bibfield  {title} {\enquote {\bibinfo {title} {{{\"U}ber das
  H-Theorem vom Anwachsen der Entropie vom Standpunkt der neueren
  Quantenmechanik}},}\ }in\ \href@noop {} {\emph {\bibinfo {booktitle}
  {Probleme der modernen Physik}}},\ \bibinfo {editor} {edited by\ \bibinfo
  {editor} {\bibfnamefont {P.}~\bibnamefont {Debye}}}\ (\bibinfo  {publisher}
  {S.~Hirzel},\ \bibinfo {address} {Leipzig},\ \bibinfo {year} {1928})\ pp.\
  \bibinfo {pages} {30--45},\ \bibinfo {note} {{Arnold} Sommerfeld zum
  60.~Geburtstag, gewidmet von seinen Sch{\"u}lern}\BibitemShut {NoStop}%
\bibitem [{\citenamefont {Louisell}(1990)}]{Louisell:QS-90}%
  \BibitemOpen
  \bibfield  {author} {\bibinfo {author} {\bibfnamefont {W.~H.}\ \bibnamefont
  {Louisell}},\ }\href@noop {} {\emph {\bibinfo {title} {Quantum Statistical
  Properties of Radiation}}},\ Wiley Classics Library\ (\bibinfo  {publisher}
  {John Wiley {\&} Sons},\ \bibinfo {address} {New York},\ \bibinfo {year}
  {1990})\BibitemShut {NoStop}%
\bibitem [{\citenamefont {Buth}(2017)}]{Buth:MC-17}%
  \BibitemOpen
  \bibfield  {author} {\bibinfo {author} {\bibfnamefont {C.}~\bibnamefont
  {Buth}},\ }\bibfield  {title} {\enquote {\bibinfo {title} {Dissociating
  diatomic molecules in ultrafast and intense light},}\ }\href@noop {}
  {\bibfield  {journal} {\bibinfo  {journal} {submitted}\ } (\bibinfo {year}
  {2017})},\ \Eprint {http://arxiv.org/abs/1708.09615} {arXiv:1708.09615}
  \BibitemShut {NoStop}%
\bibitem [{\citenamefont {Krengel}(2005)}]{Krengel:WS-05}%
  \BibitemOpen
  \bibfield  {author} {\bibinfo {author} {\bibfnamefont {U.}~\bibnamefont
  {Krengel}},\ }\href {\doibase 10.1007/978-3-663-09885-0} {\emph {\bibinfo
  {title} {Einf\"uhrung in die Wahrscheinlichkeitstheorie und Statistik}}},\
  \bibinfo {edition} {8th}\ ed.,\ \bibinfo {series} {vieweg studium; Aufbaukurs
  Mathematik}, Vol.~\bibinfo {volume} {59}\ (\bibinfo  {publisher}
  {Vieweg+Teubner Verlag},\ \bibinfo {address} {Wiesbaden},\ \bibinfo {year}
  {2005})\BibitemShut {NoStop}%
\bibitem [{\citenamefont {Kuptsov~(originator)}(2016)}]{Kuptsov:KD-16}%
  \BibitemOpen
  \bibfield  {author} {\bibinfo {author} {\bibfnamefont {L.~P.}\ \bibnamefont
  {Kuptsov~(originator)}},\ }\bibfield  {title} {\enquote {\bibinfo {title}
  {Kronecker symbol},}\ }\href
  {http://www.encyclopediaofmath.org/index.php?title=Kronecker_symbol&oldid=37525}
  {\bibfield  {journal} {\bibinfo  {journal} {Encyclopedia of Mathematics}\ }
  (\bibinfo {year} {2016})},\ \bibinfo {note}
  {\url{http://www.encyclopediaofmath.org/index.php?title=Kronecker_symbol&oldid=37525}.
  Accessed 03 April 2017}\BibitemShut {NoStop}%
\bibitem [{\citenamefont {Giorgio}\ and\ \citenamefont
  {Zuccotti}(2015)}]{Giorgio:ME-15}%
  \BibitemOpen
  \bibfield  {author} {\bibinfo {author} {\bibfnamefont {G.}~\bibnamefont
  {Giorgio}}\ and\ \bibinfo {author} {\bibfnamefont {C.}~\bibnamefont
  {Zuccotti}},\ }\bibfield  {title} {\enquote {\bibinfo {title} {Metzlerian and
  generalized {Metzlerian} matrices: some properties and economic
  applications},}\ }\href {\doibase 10.5539/jmr.v7n2p42} {\bibfield  {journal}
  {\bibinfo  {journal} {J. Math. Res.}\ }\textbf {\bibinfo {volume} {7}},\
  \bibinfo {pages} {42--55} (\bibinfo {year} {2015})}\BibitemShut {NoStop}%
\bibitem [{\citenamefont {Komlenko~(originator)}(2011)}]{Komlenko:FM-11}%
  \BibitemOpen
  \bibfield  {author} {\bibinfo {author} {\bibfnamefont {Y.~V.}\ \bibnamefont
  {Komlenko~(originator)}},\ }\bibfield  {title} {\enquote {\bibinfo {title}
  {Fundamental matrix},}\ }\href
  {http://www.encyclopediaofmath.org/index.php?title=Fundamental_matrix&oldid=14038}
  {\bibfield  {journal} {\bibinfo  {journal} {Encyclopedia of Mathematics}\ }
  (\bibinfo {year} {2011})},\ \bibinfo {note}
  {\url{http://www.encyclopediaofmath.org/index.php?title=Fundamental_matrix&oldid=14038}.
  Accessed 21 March 2017}\BibitemShut {NoStop}%
\bibitem [{\citenamefont {Sorokin}\ \emph {et~al.}(2007)\citenamefont
  {Sorokin}, \citenamefont {Bobashev}, \citenamefont {Feigl}, \citenamefont
  {Tiedtke}, \citenamefont {Wabnitz},\ and\ \citenamefont
  {Richter}}]{Sorokin:PE-07}%
  \BibitemOpen
  \bibfield  {author} {\bibinfo {author} {\bibfnamefont {A.~A.}\ \bibnamefont
  {Sorokin}}, \bibinfo {author} {\bibfnamefont {S.~V.}\ \bibnamefont
  {Bobashev}}, \bibinfo {author} {\bibfnamefont {T.}~\bibnamefont {Feigl}},
  \bibinfo {author} {\bibfnamefont {K.}~\bibnamefont {Tiedtke}}, \bibinfo
  {author} {\bibfnamefont {H.}~\bibnamefont {Wabnitz}}, \ and\ \bibinfo
  {author} {\bibfnamefont {M.}~\bibnamefont {Richter}},\ }\bibfield  {title}
  {\enquote {\bibinfo {title} {Photoelectric effect at ultrahigh
  intensities},}\ }\href {\doibase 10.1103/PhysRevLett.99.213002} {\bibfield
  {journal} {\bibinfo  {journal} {Phys. Rev. Lett.}\ }\textbf {\bibinfo
  {volume} {99}},\ \bibinfo {pages} {213002} (\bibinfo {year}
  {2007})}\BibitemShut {NoStop}%
\bibitem [{\citenamefont {Fischer}(2014)}]{Fischer:LA-14}%
  \BibitemOpen
  \bibfield  {author} {\bibinfo {author} {\bibfnamefont {G.}~\bibnamefont
  {Fischer}},\ }\href {\doibase 10.1007/978-3-658-03945-5} {\emph {\bibinfo
  {title} {{Lineare Algebra}}}},\ \bibinfo {edition} {18th}\ ed.,\ Grundkurs
  Mathematik\ (\bibinfo  {publisher} {Springer Spektrum},\ \bibinfo {address}
  {Wiesbaden},\ \bibinfo {year} {2014})\BibitemShut {NoStop}%
\bibitem [{\citenamefont {Baake}\ and\ \citenamefont
  {Schl{\"a}gel}(2011)}]{Baake:PB-11}%
  \BibitemOpen
  \bibfield  {author} {\bibinfo {author} {\bibfnamefont {M.}~\bibnamefont
  {Baake}}\ and\ \bibinfo {author} {\bibfnamefont {U.}~\bibnamefont
  {Schl{\"a}gel}},\ }\bibfield  {title} {\enquote {\bibinfo {title} {The
  {Peano-Baker} series},}\ }\href {\doibase 10.1134/S0081543811080098}
  {\bibfield  {journal} {\bibinfo  {journal} {Proc. Steklov Inst. Math.}\
  }\textbf {\bibinfo {volume} {275}},\ \bibinfo {pages} {155--159} (\bibinfo
  {year} {2011})},\ \Eprint {http://arxiv.org/abs/1011.1775} {arXiv:1011.1775}
  \BibitemShut {NoStop}%
\bibitem [{\citenamefont {Burton}(1983)}]{Burton:Vi-83}%
  \BibitemOpen
  \bibfield  {author} {\bibinfo {author} {\bibfnamefont {T.~A.}\ \bibnamefont
  {Burton}},\ }\href@noop {} {\emph {\bibinfo {title} {{Volterra} Integral and
  Differential Equations}}},\ \bibinfo {series} {Mathematics in science and
  engineering}\ No.\ \bibinfo {number} {167}\ (\bibinfo  {publisher} {Academic
  Press},\ \bibinfo {address} {New York},\ \bibinfo {year} {1983})\BibitemShut
  {NoStop}%
\bibitem [{\citenamefont {Linz}(1985)}]{Linz:AN-85}%
  \BibitemOpen
  \bibfield  {author} {\bibinfo {author} {\bibfnamefont {P.}~\bibnamefont
  {Linz}},\ }\href {\doibase 10.1137/1.9781611970852} {\emph {\bibinfo {title}
  {Analytical and Numerical Methods for {Volterra} Equations}}},\ \bibinfo
  {series} {SIAM studies in applied mathematics}\ No.~\bibinfo {number} {7}\
  (\bibinfo  {publisher} {Society for Industrial and Applied Mathematics},\
  \bibinfo {address} {Philadelphia},\ \bibinfo {year} {1985})\BibitemShut
  {NoStop}%
\bibitem [{\citenamefont
  {Bakushinskii~(originator)}(2011)}]{Bakushinskii:VE-11}%
  \BibitemOpen
  \bibfield  {author} {\bibinfo {author} {\bibfnamefont {A.~B.}\ \bibnamefont
  {Bakushinskii~(originator)}},\ }\bibfield  {title} {\enquote {\bibinfo
  {title} {Volterra equation},}\ }\href
  {http://www.encyclopediaofmath.org/index.php?title=Volterra_equation&oldid=15168}
  {\bibfield  {journal} {\bibinfo  {journal} {Encyclopedia of Mathematics}\ }
  (\bibinfo {year} {2011})},\ \bibinfo {note}
  {\url{http://www.encyclopediaofmath.org/index.php?title=Volterra_equation&oldid=15168}.
  Accessed 09 November 2016}\BibitemShut {NoStop}%
\bibitem [{\citenamefont {Khvedelidze~(originator)}(2011)}]{Khvedelidze:SA-11}%
  \BibitemOpen
  \bibfield  {author} {\bibinfo {author} {\bibfnamefont {B.~V.}\ \bibnamefont
  {Khvedelidze~(originator)}},\ }\bibfield  {title} {\enquote {\bibinfo {title}
  {Sequential approximation, method of},}\ }\href
  {http://www.encyclopediaofmath.org/index.php?title=Sequential_approximation,_method_of&oldid=14953}
  {\bibfield  {journal} {\bibinfo  {journal} {Encyclopedia of Mathematics}\ }
  (\bibinfo {year} {2011})},\ \bibinfo {note}
  {\url{http://www.encyclopediaofmath.org/index.php?title=Sequential_approximation,_method_of&oldid=14953}.
  Accessed 09 November 2016}\BibitemShut {NoStop}%
\bibitem [{\citenamefont {Fetter}\ and\ \citenamefont
  {Walecka}(1971)}]{Fetter:MP-71}%
  \BibitemOpen
  \bibfield  {author} {\bibinfo {author} {\bibfnamefont {A.~L.}\ \bibnamefont
  {Fetter}}\ and\ \bibinfo {author} {\bibfnamefont {J.~D.}\ \bibnamefont
  {Walecka}},\ }\href@noop {} {\emph {\bibinfo {title} {Quantum Theory of
  Many-Particle Systems}}},\ International Series in Pure and Applied Physics,
  edited by Leonard I.~Schiff\ (\bibinfo  {publisher} {McGraw-Hill},\ \bibinfo
  {address} {New York},\ \bibinfo {year} {1971})\BibitemShut {NoStop}%
\bibitem [{\citenamefont {Golub}\ and\ \citenamefont {van
  Loan}(1996)}]{Golub:MC-96}%
  \BibitemOpen
  \bibfield  {author} {\bibinfo {author} {\bibfnamefont {G.~H.}\ \bibnamefont
  {Golub}}\ and\ \bibinfo {author} {\bibfnamefont {C.~F.}\ \bibnamefont {van
  Loan}},\ }\href@noop {} {\emph {\bibinfo {title} {Matrix Computations}}},\
  \bibinfo {edition} {3rd}\ ed.\ (\bibinfo  {publisher} {Johns Hopkins
  University Press},\ \bibinfo {address} {Baltimore},\ \bibinfo {year}
  {1996})\BibitemShut {NoStop}%
\bibitem [{\citenamefont {Moler}\ and\ \citenamefont
  {Van~Loan}(2003)}]{Moler:DW-03}%
  \BibitemOpen
  \bibfield  {author} {\bibinfo {author} {\bibfnamefont {C.}~\bibnamefont
  {Moler}}\ and\ \bibinfo {author} {\bibfnamefont {C.}~\bibnamefont
  {Van~Loan}},\ }\bibfield  {title} {\enquote {\bibinfo {title} {Nineteen
  dubious ways to compute the exponential of a matrix, twenty-five years
  later},}\ }\href {\doibase 10.1137/S00361445024180} {\bibfield  {journal}
  {\bibinfo  {journal} {SIAM Rev.}\ }\textbf {\bibinfo {volume} {45}},\
  \bibinfo {pages} {3--49} (\bibinfo {year} {2003})}\BibitemShut {NoStop}%
\end{thebibliography}
\end{document}